\documentclass{article}
\usepackage{amsmath,amsfonts,amssymb,amsthm,graphicx}
\usepackage{dsfont}    % it's a bitmap font (on Volodya's computer)
\usepackage{enumerate}
\usepackage{color}
\usepackage{multirow}\usepackage{graphicx}
\usepackage{amsfonts}
\usepackage{amsmath}
\usepackage{amssymb}

\usepackage[margin=1.25in]{geometry} 
\usepackage{color}
\usepackage{dsfont} 
\usepackage{pgfplots}
\usepackage{subcaption}
\usepgfplotslibrary{fillbetween}
\usetikzlibrary{patterns}
\usepackage{amssymb}
\usepackage{graphicx}
\usepackage{amsmath}
\usepackage{bm}
\usepackage{tikz}

\usepackage{caption}
  
\usetikzlibrary{decorations.pathreplacing}

\usepackage{tkz-graph}
\usetikzlibrary{calc,patterns}

\usepackage{subcaption}
\usepackage{url}
\usepackage{fancyhdr}
\usepackage{indentfirst}

\usepackage{enumerate}
\usepackage{dsfont}
\usepackage[british]{babel}
\usepackage[colorlinks=true,citecolor=blue]{hyperref}
\usepackage{amsthm}
\usepackage{color}
\usepackage{natbib}
\usepackage{comment}
\usepackage{tikz-cd}
\usepackage{graphicx}
\usepackage{amsfonts}
\usepackage{amsmath}
\usepackage{amssymb}
\usepackage{url}
\usepackage{bbm}
\usepackage{fancyhdr}
\usepackage{indentfirst}
\usepackage{enumerate}
\usepackage{dsfont}
\usepackage[british]{babel}
\usepackage[colorlinks=true,citecolor=blue]{hyperref}
\usepackage{cleveref}
\usepackage{amsthm}
\usepackage{color}
\renewcommand{\cite}{\citet}
\usepackage{comment}

\usepackage{algorithm}
\usepackage[noend]{algpseudocode}  % algpseudocode loads algorithmicx
\renewcommand{\d}{\,\mathrm{d}}
\newcommand{\dd}{\overset{\mathrm{law}}{=}}

\usepackage{xcolor}

\newcommand{\E}{\mathbb{E}}    % expectation
\newcommand{\R}{\mathbb{R}}    % real numbers
\newcommand{\N}{\mathbb{N}}    % natural numbers

  \newcommand{\id}{\mathds{1}} % indicator with dsfont
\newcommand{\I}{\mathcal{I}}   % the Borel sigma-field on [0,1]
 
\theoremstyle{plain}
\newtheorem{theorem}{Theorem}
\newtheorem{corollary}{Corollary}
\newtheorem{lemma}{Lemma}
\newtheorem{proposition}{Proposition}

\theoremstyle{definition}
\newtheorem{definition}{Definition}
\newtheorem{example}{Example}

\theoremstyle{remark}
\newtheorem{remark}{Remark}

\usepackage{tikz}
\usetikzlibrary{arrows.meta,positioning}

\newcommand{\ee}{\varepsilon}

\newcommand{\n}[1]{\left\lVert#1\right\rVert}
\newcommand{\B}{\mathcal{B}}
\newcommand{\cE}{\mathcal{E}}
\newcommand{\bz}{{\mathbf{z}}}
\renewcommand{\P}{\mathcal{P}}

\newcommand{\bx}{\mathbf{x}}
\newcommand{\by}{\mathbf{y}}

\newcommand{\bmu}{\boldsymbol{\mu}}

\renewcommand{\bm}{\mathbf{m}}
\newcommand{\bn}{\mathbf{n}}
\newcommand{\bta}{\boldsymbol{\eta}}
\newcommand{\bnu}{\boldsymbol{\nu}}
\newcommand{\T}{\mathcal{T}}
\newcommand{\nub}{\bar{\nu}}
\newcommand{\mub}{\bar{\mu}}

\newcommand{\bpsi}{\boldsymbol{\psi}}

\def\d{\mathrm{d}}

\def\laweq{\buildrel \mathrm{law} \over =}

\def\lawis{\buildrel \mathrm{law} \over \sim}

\newcommand{\bone}{ {\mathbbm{1}} }
\renewcommand{\S}{\mathbb{M}}
\renewcommand{\P}{\mathcal{P}}

\usepackage{mathrsfs}
\newcommand{\M}{\mathcal{M}}
\newcommand{\C}{\mathcal{C}}
\newcommand{\cR}{\mathcal{R}}
\newcommand{\K}{\mathcal{K}}

\newcommand{\gicx}{\succeq_{\mathrm{icx}}}

\newcommand{\rd}{\mathbb{R}_+^d}

\newcommand{\nn}{m_{\bnu}}
\newcommand{\mm}{m_{\bmu}}

\newcommand{\W}{\mathcal{W}}
\newcommand{\V}{\mathcal{V}}
\newcommand{\one}{\mathbf{1}}
\newcommand{\lcx}{\preceq_{\mathrm{cx}}}

\newcommand{\gcx}{\succeq_{\mathrm{cx}}}

\renewcommand{\ge}{\geqslant}
\renewcommand{\le}{\leqslant}
\renewcommand{\geq}{\geqslant}
\renewcommand{\leq}{\leqslant}
\renewcommand{\epsilon}{\varepsilon}

 \pgfdeclarepatternformonly{south east lines}{\pgfqpoint{-0pt}{-0pt}}{\pgfqpoint{3pt}{3pt}}{\pgfqpoint{3pt}{3pt}}{
    \pgfsetlinewidth{0.4pt}
    \pgfpathmoveto{\pgfqpoint{0pt}{3pt}}
    \pgfpathlineto{\pgfqpoint{3pt}{0pt}}
    \pgfpathmoveto{\pgfqpoint{.2pt}{-.2pt}}
    \pgfpathlineto{\pgfqpoint{-.2pt}{.2pt}}
    \pgfpathmoveto{\pgfqpoint{3.2pt}{2.8pt}}
    \pgfpathlineto{\pgfqpoint{2.8pt}{3.2pt}}
    \pgfusepath{stroke}}

\begin{document}
\title{Simultaneous optimal transport}
\author{Ruodu Wang\thanks{Department of Statistics and Actuarial Science, University of Waterloo, Canada. Email: \texttt{wang@uwaterloo.ca}.}\and Zhenyuan Zhang\thanks{Department of Mathematics, Stanford University, USA. Email: \texttt{zzy@stanford.edu}.}}
\maketitle 

\begin{abstract}
We propose a general framework of mass transport between non-negative vector-valued measures, which will be called simultaneous optimal transport (SOT). 
The new framework is motivated by the need to transport resources of different types simultaneously, i.e., in single trips, from specified origins to destinations; similarly,
 in economic matching, one needs to couple two groups, e.g., buyers and sellers, by equating supplies and demands of different goods at the same time. 
The mathematical structure of simultaneous transport is very different from the classic setting of optimal transport, leading to many new challenges. 
The Monge and Kantorovich formulations are contrasted and  connected. 
Existence conditions and duality formulas are established. 
More interestingly, 
 by connecting SOT to a natural relaxation of martingale optimal transport (MOT), 
 we introduce the MOT-SOT parity, which allows for explicit solutions of SOT in many interesting cases.
%In particular, the duality theorem gives rise to a  labour market  equilibrium model where each worker has several types of skills and each firm seeks to employ these skills at different levels. We also introduce a notion of Wasserstein distance for non-negative vector-valued measures using a more general decomposition theorem which reduces the study of the two-way simultaneous transport to the classical setting.

\end{abstract}  
\tableofcontents

 \section{Introduction}
 
 Optimal transport theory, originally developed by Monge and Kantorovich (see \cite{V09} for a history), 
 has wide applications in various scientific fields, including economic theory, operations research, statistics, machine learning, and quantitative finance. 
A specialized treatment of optimal transport in economics is given by \cite{G16}. For a mathematical background on optimal transport and its applications, we refer to the textbooks of \cite{A03,S15} and \cite{V03, V09}.
 
%In the recent decade, optimal transport has received increasing attention in economic studies with many relevant applications, such as  contract design (\cite{E10}), robust risk assessment (\cite{EPR13}),  Cournot-Nash equilibria in non-atomic games (\cite{BC16}), multiple-good monopoly (\cite{DDT17}), implementation problems (\cite{NS18}), and team matching (\cite{BTZ21}). 
%A specialized treatment of optimal transport in economics is given by \cite{G16}.

In this paper, we propose a new framework of optimal transport, which will be called \emph{simultaneous optimal transport} (SOT). 
In contrast to the classic optimal transport theory, which studies transports between two measures on spaces $X$ and $Y$, a simultaneous transport (either Monge or kernel, with a precise formulation in Section \ref{sec:model}) moves mass from $d$ measures on $X$ to $d$ measures on $Y$ \emph{simultaneously}.

SOT provides powerful tools for matching problems with multiple distributional constraints.  A considerable amount of new challenges and relevant applications arise, which will gradually be revealed in this paper. 
The new framework,  being mathematically interesting itself, is motivated by several applications from economics, risk management, and stochastic modeling, which are discussed in Section  \ref{sec:model} and Appendix \ref{sec:equilibrium}.   
As a primary example (details in Example \ref{ex:1}), suppose that several factories need to supply $d$ types of products to several retailers, and each factory only has one truck to transport their products to one destination. Since each product type has its own supply and demand, the objective is to make a transport plan such that all demands are met.
In case $d=1$, we speak of the classic optimal transport problem.   
Another natural example is refugee resettlement, where refugee families are resettled to different affiliates while fulfilling various quotas and requirements (Example \ref{ex:refugee}). 
%Two more motivating examples are from risk management and time-homogeneous Markov processes, and the final example is a model of  labour market equilibrium  (Appendix \ref{sec:equilibrium}).
%The third motivating example comes from risk management, where an investor seeks optimal investment subject to distributional constraints under different scenarios (Example \ref{opt}).
%A fourth  example, from stochastic modeling,  concerns constructing optimal time-homogeneous Markov processes (Example \ref{ex:markov}). 
%Further,  a labour market equilibrium model  where  workers have  multi-dimensional skills is discussed in Appendix \ref{sec:equilibrium}.

 We will explain below several sharp contrasts between the new and the classic frameworks, along with our contributions and results.   The following points are ordered by their natural logical appearance, although the  main mathematical results (Theorems \ref{infs}-\ref{2way}) come a bit later. 

First, inspired by the example above, the measures at origin (supplies) do not necessarily have the same mass as the measures at destination (demands to meet). Obviously, there does not exist a possible transport if the  demands (in any product type) are larger than the supplies, but there can be transports if the demands are smaller than the supplies. 
We will say that the SOT problem is \emph{balanced} if the vector of total masses at origin  is equal to the vector of total masses at destination, and otherwise it is \emph{unbalanced} (see Section \ref{sec:model} for a precise definition).
Unbalance is generally not an issue if $d=1$ since one can glue a point at the destination which incurs no transport cost to reformulate the problem as a balanced problem, but such a trick does not work in the SOT setting;  see Section \ref{sec:model} for an explanation. A connection between the balanced and unbalanced settings is established in Section \ref{sec:Kanto} via a continuity result (Proposition \ref{connect}).
 
Second, one needs to specify a reference measure with respect to which the transport cost is computed. In classic transport theory, the cost is integrated with respect to the measure at origin (supply). 
In the example above, it seems that none of the distributions of the product supplies is a natural benchmark for computing the cost; neither are their combinations. A separate benchmark measure needs to be introduced (see Section \ref{sec:model}), and it may cause extra technical subtlety depending on whether it is equivalent to a measure dominating the measures at origin. 

Third, for two given $d$-tuples of (probability) measures, a simultaneous transport may not exist, even if there are no atoms in these measures (transports between atomless probabilities always exist in case $d=1$).
As a trivial example, suppose that there are a continuum of factories, each supplying an equal amount of product A and product B, and a continuum of retailers, half  demanding a ratio of $2:1$ between products A and B and the other half demanding a ratio of $1:2$ between A and B. If the total demand vector is equal to the total supply vector, then there is obviously no possible transport plan; indeed, any transport plan would supply the same amount of A and B to any retailer, leading to over-supplying of one product for each supplier. However, if, instead of a $1:1$ ratio, half of the factories supply in a $3:1$ ratio between A and B, 
and the other half supply in a $1:3$ ratio, then transport plans exist, and we can choose from these plans to minimize the total transport cost. Moreover, it is easy to see from this example that the SOT problems are not symmetric in the measures at origin and the measures at destination, in sharp contrast to the classic problem.
Even if transport plans exist, the set which it can be chosen from is bound to additional constraints. 
The existence issue of simultaneous transport will be studied in Section \ref{S2} using the notions of joint non-atomicity and heterogeneity order (Proposition \ref{prop:ssww}), based on existing results of  \cite{T91} and \cite{SSWW19}. 
Several other interesting inequalities (e.g., Proposition \ref{2infp}) are also discussed in Section \ref{S2}.

Fourth, in the balanced setting, the classic transport problem can be conveniently written in the Kantorovich formulation as each transport corresponds to a joint probability measure with specified marginals but unspecified dependence structure (or a copula, see e.g., \cite{B10} and \cite{J14}). 
In the SOT framework, since there is no ``first marginal" or ``second marginal" of the problem (instead,  two vectors of   marginals), 
the Kantorovich formulation via joint distributions is less clear than in the classic case, and it is studied in Section \ref{sec:Kanto}.
Assuming joint non-atomicity, we prove that the Monge and Kantorovich (kernel) formulations have the same infimum cost (Theorem \ref{infs}). 

Fifth, a duality theorem for balanced SOT is obtained   in Section \ref{Dua}, which has a different form compared with the classic duality formula (Theorem \ref{duality1}).  Using the duality result, 
we construct in Appendix  \ref{sec:equilibrium} a labour market equilibrium model (see e.g., \cite{G16} for a classic equilibrium model in case $d=1$), where workers, each with several types of skills and seeking to optimize their wage, are matched with firms, each seeking to employ these skills of a certain cumulative amount to optimize their profit. The equilibrium wage function and the equilibrium profit function are obtained from the duality formula for  given distributions of the skills that workers supply and firms seek.

% Sixth, the Wasserstein quasi-metric is defined for simultaneous transports, but it does not have symmetry since the transport problem is clearly not symmetric in the two vectors of measures.
% In case transports from and to the measures at origin both exist, 
% we speak of \emph{two-way} transport problems (in the classic setting via the  Kantorovich formulation, all balanced transport problems are two-way).
% In the setting of measures for which two-way transports exist, the Wasserstein distance can be naturally defined in Section \ref{W}. 
% In addition, we provide a  decomposition  formula of the optimal transport and its optimal cost in the two-way transport setting,
% which can be solved explicitly based on existing results (e.g., \cite{GM96}) if the cost function is convex.  

%  Finally, due to the new structure of the simultaneous transport problem, we are able to discuss \emph{uniqueness} of the transport\footnote{Note this is different from \emph{uniqueness} of the \emph{optimal} transport.} for some given vectors of measures. Uniqueness is shown in Section \ref{W} for the two-way transport setting under   additional conditions on the structure of the measures at origin. Note that uniqueness of the Kantorovich transport does not appear at all in  case $d=1$ except for the trivial case where the measure at origin is degenerate. 

Sixth, and most importantly, SOT enjoys a unique connection to the active literature 
 of martingale optimal transport (MOT) between two probability measures, that is, classic optimal transport with a martingale constraint. The study of MOT, initialized by \cite{BHP13} in  discrete time and \cite{GHT14} in continuous time,  is motivated by   applications in  mathematical finance, in particular, in robust option pricing.  The theory is further reinforced by \cite{BJ16}, \cite{BNT17}, and \cite{DT19}, among many others; see also \cite{H17} for a recent survey. In Section \ref{W}, we discover an intriguing connection between SOT and MOT, which we call 
 the {MOT-SOT parity}, that connects SOT in the balanced case with a suitable relaxation of MOT (Theorem \ref{thm:rep}). 
 This connection allows us to apply techniques from MOT to SOT, thus bridging between two rich topics. 
In the special case of two-way transport, i.e., simultaneous transport is possible in both forward and backward directions, the MOT component of the problem is degenerate, and the SOT problem can be completely solved 
(Theorem \ref{2way}).

In Section \ref{sec:7} we conclude the paper with several other promising directions for future research and open challenges. 
In recent years there has been a growing interest in various generalizations of the classic Monge--Kantorovich optimal transport problem. A few generalizations of optimal transport in higher dimensions are related to our paper. To minimize distraction to the reader, we collect them in Appendix \ref{rev} with some detailed discussions. The closest to our framework is perhaps \cite{W19} who considered a similar setting to our simultaneous transport with a different focus and distinctive mathematical results.  
% Our main results, in particular Theorems  \ref{infs}, \ref{duality1} and \ref{2way} require rather long proofs with several lemmas, and they are relegated to  Appendices \ref{D}-\ref{app:Pf6} along with other proofs and additional results. 

\section{Simultaneous optimal transport}
\label{sec:model}

We first briefly review the classic Monge--Kantorovich transport problem. 
 For a measurable space $X$ that is also a Polish space equipped with the Borel $\sigma$-field $\B(X)$, we denote by $\P(X)$ the set of all Borel probability measures on $X$. Consider Polish spaces $X,Y$, and probability measures $\mu\in \P(X) $ and $\nu\in\P(Y)$. 
 Although our results are formulated on general Polish spaces, it does not hurt to think of $X=\R^N$ and $Y=\R^N$ as the primary example.
 We will always equip $X\times Y$ with the product $\sigma$-field. In the following while writing $A\subseteq X,\ B\subseteq Y$, we always assume that $A,B$ are Borel measurable subsets. Given a cost function $c:X\times Y\to [0,\infty]$, the classic optimal transport problem raised by Monge asks for
$$\inf_{T\in \T(\mu,\nu)}\int_Xc(x,T(x))\mu(\d x),$$
    where $\T(\mu,\nu)$ consists of transport maps from $\mu$ to $\nu$,  i.e., measurable functions $T:X\to Y$ such that $\mu\circ T^{-1}=\nu$.
    
    Kantorovich later studied a relaxation of Monge's problem, that is, to solve for
    $$\inf_{\pi\in \Pi(\mu,\nu)}\int_{X\times Y}c(x,y)\pi(\d x,\d y),$$ where $\Pi(\mu,\nu)$ is the set of transport plans from $\mu$ to $\nu$,  i.e., 
    the set of probability measures $\pi\in\P(X\times Y)\text{ such that for any }A\subseteq X\text{ and } B\subseteq Y,\ \pi(A\times Y)=\mu(A)\text{ and } \pi(X\times B)=\nu(B).$ These are the celebrated Monge--Kantorovich optimal transport problems. 
     
\subsection{Simultaneous transport}

%We propose a generalization of the classic Monge--Kantorovich problem. 
Throughout, components of vector-valued measures are non-negative (i.e., they are measures, not signed measures).
We denote by $d\in\N$ the dimension of a vector-valued measure,  where the more interesting case is when $d\geq 2$, and by $[d]=\{1,\dots,d\}$. We work with $d$-tuples of finite Borel measures $\bmu=(\mu_1,\dots,\mu_d)$ on $X$ and $\bnu=(\nu_1,\dots,\nu_d)$ on $Y$ such that for each $j \in [d] $, $\mu_j(X)\geq\nu_j(Y)>0$.  

We propose the new framework of \emph{simultaneous optimal transport} (SOT) by requiring that a certain transport \emph{map} or transport \emph{plan} sends $\mu_j$ to cover $\nu_j$ simultaneously for all $j \in [d]$. In this setup,  the set of all simultaneous transport maps is defined as 
$$\T(\bmu,\bnu):=\{T:X\to Y \mid \bmu\circ T^{-1}\geq \bnu\}.$$
Here and throughout, equalities and inequalities are understood component-wise, and for two measures $\mu$ and $\nu$ on the same space, $\mu\ge \nu$ means that $\mu(A)\ge \nu(A)$ for all measurable $A$.
If $\bmu(X)=\bnu(Y)$, then we speak of \emph{balanced} simultaneous transports.

The most natural and intuitive way to describe the set of all (simultaneous) transport plans is to use $\mathcal K(\bmu,\bnu)$, the set of all stochastic kernels $\kappa$ such that $\kappa_{\#} \bmu\geq \bnu$,
and defined as 
\begin{equation}
    \label{eq:kappadef}
\kappa_{\#} \bmu (\cdot) := \int_X \kappa(x;\cdot)  \bmu(\d   x) \ge \bnu(\cdot) .\end{equation}

Imagine that one would like to distribute goods from a (possibly infinitesimal) point $x\in X$ to different places in $Y$, then the measure $\kappa(x;\cdot)$ describes such a distribution. In view of this definition, the set of stochastic kernels $\mathcal K(\bmu,\bnu)$ can be written as an intersection:
$$\mathcal K(\bmu,\bnu)=\bigcap_{j=1}^d\mathcal K(\mu_j,\nu_j).$$In words, a simultaneous transport plan from $\bmu$ to $\bnu$ sends simultaneously $\mu_j$ to $\nu_j$ for any $j \in [d]$.
The non-emptyness of $\T(\bmu,\bnu)$ and $ \mathcal K(\bmu,\bnu)$ is not guaranteed generally, and will be explained later in Section \ref{exi}.  

In the case $d=1$ and $\mu(X)=\nu(Y)$, our problem reduces to the classic Monge--Kantorovich problem. 
We first illustrate an example of simultaneous transport problems, which sheds some light on the special structure and technical difference of our problem in contrast to the classic  problem.

  \begin{example}[Simultaneous transport of supplies]
  \label{ex:1}
    Suppose that there are $m$ factories; each factory $j$ has $a_j$ units of product A and $b_j$ units of product B.
  There are $m'$ retailers, each demanding $a_k'$ units of A and $b_k'$ units of B.
  We assume that the supply is enough to cover the demand, that is, with normalization,
  $$1=\sum_{j=1}^m a_j \ge  \sum_{k=1}^{m'} a_k' \mbox{~~and~~} 1=\sum_{j=1}^m b_j\ge \sum_{k=1}^{m'} b_k'.$$
  If we assume demand-supply clearance, then, with normalization, \begin{equation}
      \label{eq:supply-demand}\sum_{j=1}^m a_j= \sum_{k=1}^{m'} a_k'=\sum_{j=1}^m b_j=\sum_{k=1}^{m'} b_k'=1. 
  \end{equation} 
  Let $\mu_1$ be a probability such that $\mu_{1}(\{j\})=a_j$ for each $j$, and similarly, $\mu_{2}(\{j\})=b_j$ for each $j$, and $\nu_{1}(\{k\})=a'_k$ and $\nu_{2}(\{k\})=b'_k$ for each $k$. Write $\bmu=(\mu_1,\mu_2)$ and $\bnu=(\nu_1,\nu_2)$. 
   \begin{enumerate}
       \item   A transport in $\mathcal T( \bmu,\bnu)$ or $\mathcal K( \bmu,\bnu)$, if it exists, is an arrangement to send  products from factories to retailers to meet their demand. We cannot transport products within the $m'$ retailers or within the $m$ factories.
   \item  The transport in $\mathcal T( \bmu,\bnu)$  is required to be done in single trips:  One factory can only supply one retailer.
    This is illustrated in  Figure \ref{fig:2} (a). As a practical example, we may think of the situation where each factory only has one truck that goes to one destination in every production cycle.
       \item We may allow each factory to supply multiple retailers, e.g., a factory with multiple trucks. 
       In this case, we can use the formulation of transport kernels in $\mathcal K( \bmu,\bnu)$. 
       We note that a non-trivial constraint imposed by the formulation \eqref{eq:kappadef}  is that the amount of A and that of B are proportional in each truck departing from the same factory (e.g., bundled goods, or worker skills in Appendix \ref{sec:equilibrium} which are not  divisible).
       This is illustrated in  Figure \ref{fig:2} (b).
%\item 
% If, however, we like one retailer to be only supplied by one factory, then we speak of transports in $\mathcal T(   (\nu_1,\nu_2),(\mu_1,\mu_2))$. This leads to two-way transports treated in Section \ref{W}.       
\item If demand-supply clearance \eqref{eq:supply-demand} holds, then the transport is balanced; otherwise it is unbalanced.  In case \eqref{eq:supply-demand} holds, one may consider the backward direction of transporting $\bnu$ to $\bmu$, and this leads to two-way transports treated in Section \ref{W}.    
   \end{enumerate}

\label{fac}
    \end{example}
    
\begin{figure}[h!]  
\begin{center}
\begin{subfigure}[b]{0.49\textwidth}
\centering
\begin{tikzpicture} 
\draw[gray, very thick] (-4.9,0)--(0.9,0);
\draw[gray, very thick] (-4.9,-5)--(0.9,-5);
\filldraw[blue, very thick] (-4.6,0) rectangle (-4.4,0.3);
\filldraw[red,ultra thick] (-4.6,0.3) rectangle (-4.4,1);
\filldraw[blue, very thick] (-3.6,0) rectangle (-3.4,0.8);
\filldraw[red,ultra thick] (-3.6,0.8) rectangle (-3.4,1);
\filldraw[blue, very thick] (-2.6,0) rectangle (-2.4,0.7);
\filldraw[red,ultra thick] (-2.6,0.7) rectangle (-2.4,1);
\filldraw[blue, very thick] (-1.6,0) rectangle (-1.4,0.1);
\filldraw[red,ultra thick] (-1.6,0.1) rectangle (-1.4,1);
\filldraw[blue, very thick] (-0.6,0) rectangle (-0.4,0.5);
\filldraw[red,ultra thick] (-0.6,0.5) rectangle (-0.4,1);
\filldraw[blue, very thick] (0.4,0) rectangle (0.6,0.6);
\filldraw[red,ultra thick] (0.4,0.6) rectangle (0.6,1);

\filldraw[blue, very thick] (-3.2,-5) rectangle (-2.9,-3.9);
\filldraw[red,ultra thick] (-3.2,-3.9) rectangle (-2.9,-2);
\filldraw[blue, very thick] (-1.1,-5) rectangle (-0.8,-3.1);
\filldraw[red,ultra thick] (-1.1,-3.1) rectangle (-0.8,-2);

\draw[very thick, ->](-4.5,-0.1)--(-3.3,-1.9);
\draw[very thick, ->](-3.5,-0.1)--(-1.2,-1.9);
\draw[very thick, ->](-2.5,-0.1)--(-3,-1.9);
\draw[very thick, ->](-1.5,-0.1)--(-2.8,-1.9);
\draw[very thick, ->](-0.5,-0.1)--(-0.9,-1.9);
\draw[very thick, ->](0.5,-0.1)--(-0.7,-1.9);
\end{tikzpicture}
\caption{Monge}
\end{subfigure} 
\begin{subfigure}[b]{0.49\textwidth}
\centering
\begin{tikzpicture}
 \draw[gray, very thick] (-4.9,-1)--(0.9,-1);
\draw[gray, very thick] (-4.9,-5)--(0.9,-5);  
\filldraw[blue,very thick] (-2.2,-1) rectangle (-1.9,0.32);
\filldraw[red,ultra thick] (-2.2,0.32) rectangle (-1.9,0.98);  
\filldraw[blue, opacity=0.2] (-3.2,-1) rectangle (-2.9,-0.3);
\filldraw[red,opacity=0.2] (-3.2,-0.3) rectangle (-2.9,0.98); 
 \filldraw[blue, opacity=0.2] (-4.2,-1) rectangle (-3.9,0.35);
\filldraw[red,opacity=0.2] (-4.2,0.35) rectangle (-3.9,0.98);  
\filldraw[blue,opacity=0.2] (-1.2,-1) rectangle (-0.9,0.2);
\filldraw[red,opacity=0.2] (-1.2,0.2) rectangle (-0.9,0.98);  
\filldraw[blue, opacity=0.2] (-0.2,-1) rectangle (0.1,-0.7);
\filldraw[red,opacity=0.2] (-0.2,-0.7) rectangle (0.1,0.98);  
\filldraw[blue, very thick] (-4,-5) rectangle (-3.8,-3.2);
\filldraw[red,ultra thick] (-4,-3.2) rectangle (-3.8,-2.3);
\filldraw[blue, very thick] (-2.2,-5) rectangle (-2,-4);
\filldraw[red,ultra thick] (-2.2,-4) rectangle (-2,-3.5);
\filldraw[blue, very thick] (-0.3,-5) rectangle (-0.1,-3.6);
\filldraw[red,ultra thick] (-0.3,-3.6) rectangle (-0.1,-2.9); 
\draw[very thick, ->](-2.05,-1.1)--(-3.8,-2.1); 
\draw[very thick, ->](-2.05,-1.1)--(-2.1,-3.2); 
\draw[very thick, ->](-2.05,-1.1)--(-0.3,-2.7); 
\end{tikzpicture}
\caption{Kernel}
\end{subfigure}

\end{center}

\caption{A showcase of simultaneous transport of supplies; red and blue represent different types of products.}
\label{fig:2}
\end{figure}

Example \ref{ex:1} and its continuous version will serve as  a primary example to facilitate the understanding of our new framework. To quantify the cost of simultaneous transports, a cost function will be associated with the simultaneous transport problem, as in the classic formulation. Throughout, we define the normalized average measures \begin{align}
    \mub:=\frac{\sum_{ j=1}^d \mu_j}{\sum_{ j=1}^d\mu_j(X)}\text{ and }\nub:=\frac{\sum_{ j=1}^d\nu_j}{\sum_{ j=1}^d\nu_j(Y)},\label{mubnub}
\end{align}
which are probability measures.
In case $\mu_1,\dots,\mu_d$ are themselves probability measures, $\mub$ is their arithmetic average. 
Consider a measurable function $c:X\times Y\to [0,\infty]$ and a \emph{reference} probability measure $\eta$ on $X$ such that $\eta\ll\mub$. We define the transport costs as follows: for $T\in \T(\bmu,\bnu)$, let
\begin{align}
\mathcal C_{\eta}(T) :=\int_X c(x,T(x))  \eta(\d x). \label{eq:cost}
 \end{align}
 Such a reference measure $\eta$ allows us the greatest generality in view of Example \ref{fac}: We allow nonlinear dependencies of $\eta$ in terms of $\bmu$, for example, when computing the petrol cost which is nonlinear in weights of the transported products. We impose the condition $\eta\ll\mub$ because it would be unreasonable to assign a cost where there is no transport. (For general $\eta\in \M(X)$, we can always normalize it to a probability without loss of generality.) This is also equivalent to multiplying the cost function $c$ by a factor $\d\eta/\d\mub(x)$ and using $\mub$ as the reference measure, but such a change may affect the continuity of the cost function. 
 
 In terms of $\kappa \in\mathcal K(\bmu,\bnu)$, we define the transport cost
    \begin{align}
 \label{eq:cost2}\mathcal C_{\eta}(\kappa):=   \int_{X\times Y} c(x,y)  \eta \otimes \kappa (\d x,\d y).
 \end{align}
 
 The   quantities of interest are the minimum (or infimum) costs 
     \begin{align*}
\inf_{T\in \mathcal T(\bmu,\bnu)}    \mathcal C_{\eta}(T) \mbox{~~~and~~~} \inf_{\kappa\in \mathcal K(\bmu,\bnu)}    \mathcal C_{\eta}(\kappa),
     \end{align*}
     as well as the optimizing transport map and kernel. 
    If $\T(\bmu,\bnu)$ or $\mathcal K(\bmu,\bnu)$ is an empty set, the corresponding minimum cost is set to $\infty$. 
    In dimension $d=1$ and when $\eta=\mub$, this cost coincides with the classic Monge--Kantorovich costs. In case $\eta=\mub$, we omit the subscript $\eta$ in \eqref{eq:cost} and \eqref{eq:cost2}.

\begin{example}[Refugee resettlement]
\label{ex:refugee} Refugee resettlement is an active problem in operations research (\cite{DKT16} and \cite{A21}).   Let $\mathcal F=\{F^1,\dots,F^{I}\}$ denote the set of refugee families, where each family $F^i$ consists of members $F^i=\{f^{i,1},\dots,f^{i,J_i}\}$. Our goal is to resettle these refugee families to affiliates $\mathcal L=\{L^1,\dots,L^{N}\}$, such as different cities across USA. 
There are various quotas  $\mathcal Q=\{Q^1,\dots,Q^{K}\}$ to be fulfilled by each family, such as the numbers of adults and children. Let $q^i_k$ denote the contribution of quota  $k$ by family $i$.  Each quota $Q^k$ must exceed $\underline{q}^\ell_k$ in the affiliate $L^\ell$. The constraints are twofold: each family member in a refugee family must be resettled to the same affiliation, and the quota requirements are satisfied. To each family-affiliation match is attached a quality score $v^i_\ell$, such as the total employment outcome. These lead to the following integer optimization problem:
\begin{align}
    \text{maximize }\ &\sum_i\sum_\ell v^i_\ell z^i_\ell,\nonumber\\
    \text{subject to }\ &z^i_{\ell}\in\{0,1\}\text{ and }\sum_\ell z^i_\ell\leq 1\text{ for all }i;\nonumber\\
    & \sum_i q^i_kz^i_{\ell}\geq \underline{q}^\ell_k\text{ for all }\ell,k.\label{eq:refugee}
\end{align}
 To see this is within the SOT framework \eqref{eq:kappadef}, we let $X=\mathcal F$ and $Y=\mathcal L$, and define measures $(\mu_1,\dots,\mu_{|\mathcal Q|})$ on $X$ by $\mu_k(\{F^{i}\})=q^i_k$ and $(\nu_1,\dots,\nu_{|\mathcal Q|})$ on $Y$ by $\nu_k(\{L^{\ell}\})=\underline{q}^\ell_k$. The cost function is $-v^i_\ell$. The condition $z^i_{\ell}\in\{0,1\}$ asserts that the problem is Monge, corresponding to the fact that each family may be resettled  only in one affiliate. Our formulation differs from the original formulations in \cite{DKT16} and \cite{A21} where the $``\geq"$ in \eqref{eq:refugee} is $``\leq"$, thus a ``dual SOT problem" unbalanced in an opposite direction.
\end{example}

In general, if the supports of $\mub,\nub$ are both finite (e.g., Example \ref{fac}), then the optimal transport problem  is equivalent to a finite-dimensional linear programming problem, which can be handled conveniently by linear programming solvers.
The dimension $d\ge 2$ of $\bmu$ and $\bnu$
leads to more constraints in this linear program compared to the classic case of $d=1$. These additional constraints are highly non-trivial. For instance, the additional constraints may rule out the existence of any transport, in contrast to the case $d=1$; see Section \ref{exi}.

\subsection{Balanced simultaneous transport}

Although we have set up the problem in greater generality with unbalanced measures, in some parts of this paper we will focus on the balanced case where $\bmu(X)=\bnu(Y)$. We may without loss of generality assume that each $\mu_j,\nu_j$ are probability measures. In this case, we have
$$\T(\bmu,\bnu)=\{T:X\to Y \mid \bmu\circ T^{-1}=\bnu\}$$and$$\mathcal K(\bmu,\bnu)=\{\kappa \mid \kappa_{\#} \bmu=\bnu\}.$$
The two examples below illustrate some particular applications of this setting, in addition to the supply-demand clearing case \eqref{eq:supply-demand} of Example \ref{ex:1}. 

\begin{example}[Financial cost efficiency with multiple distributional constraints]\label{opt}
Let $(X,\mathcal F)$ be a measurable space on which $\mu_1,\dots,\mu_d$ are $d$ probability measures and $\mathcal L$ be the set of random variables on $(X,\mathcal F)$. 
Let $\nu_1,\dots,\nu_d$ be $d$ distributions on $\R$ 
and define 
  $$\mathcal L_{\bnu}(\bmu):=\{L\in \mathcal L\mid  L \lawis_{\mu_i} \nu_i  , ~i\in [d] \}$$
where $L\lawis_\mu \nu$ means that $L$ has distribution $\nu$ under $\mu$. 
The set $\mathcal L_{\bnu}(\bmu)$ represents all possible financial positions which have distribution $\nu_j$ under a reference probability $\mu_j$.
As an example in case $d=2$, an investor may seek for an investment $L$ which has a target distribution $\nu_1$ under her subjective probability measure $\mu_1$ and is bound by regulation to have a distribution $\nu_2$ under a regulatory measure $\mu_2$; see \citet[Section 5]{SSWW19}.
The investor is interested in the optimization problem \begin{align}\min\left\{\E^\eta [f(L)] \mid  L\in \mathcal L_{\bnu}(\bmu)\right\},\label{m}\end{align} where $\eta\ll \mub$ and $f$ is a nonnegative measurable function. 
If the probability measure $\eta$ is a pricing measure on the financial market, then the optimization problem \eqref{m} is to find the cheapest financial position $f(L)$ with $L$ satisfying the distributional constraints. 
%Analogous results hold when replacing $\min$ by $\max$ in \eqref{m}.
In case $d=1$, i.e., with only one distributional constraint, this problem  is the cost-efficient portfolio problem studied by \cite{D88}, which can be solved by the classic Fr\'{e}chet-Hoeffding inequality (e.g., \cite{R13}). 
For $d\ge 2$, the problem becomes much more complicated, and a special case of mutually singular  $\mu_1,\dots,\mu_d$ is studied by \cite{WZ21} as the basic tool for representing \emph{coherent scenario-based risk measures}.

Note that  by definition $ \mathcal L_{\bnu} (\bmu) = \T(\bmu,\bnu)$. Hence, $L\in \mathcal L_{\bnu} (\bmu)$ is a balanced Monge transport from $\bmu$ to $\bnu$, and 
$$\E^\eta[f(L)]=\int_X f(L(\omega))  \eta(\d\omega)
%=\int_ X \frac{\d \eta}{\d \mub}(\omega)f(L(\omega)) \mub(\d\omega)
,$$
which is simply the transport cost of $L$ as a Monge transport, with cost function $c(x,y)=f(y)$ and reference measure $\eta$. We will see from Theorem \ref{infs} that if $f$ is continuous and $\bmu$ is jointly atomless, then the infimum of the cost is the same as the infimum cost among the corresponding transport plans. If $\eta \sim \mub$, further duality results from Section \ref{Dua} are applicable.

\end{example}
\begin{example}[Time-homogeneous Markov  processes with specified marginals] \label{ex:markov}
 Let $\mu_1,\dots,\mu_T$ be probability measures on $X=\R^N$  
and  $\xi=(\xi_t)_{t=1,\dots,T}$ be an $\R^N$-valued Markov process with marginal distributions $\mu_1,\dots,\mu_T$. The Markov kernels of $\xi$, $\kappa_t :\R^N\to \mathcal P(\R^N)$ for $t=1,\dots,T-1$,  are  such that 
  $\kappa_t(\bx)$ is the distribution of $\xi_{t+1}$ conditional on $\xi_t=\bx$. 
 Here and throughout conditional distributions (probabilities) should be understood as   regular conditional distributions (probabilities). 
  The Markov process   $\xi $ is time-homogeneous if $\kappa:=\kappa_t$ does not depend on $t$.
  In other words, $\kappa $ needs to satisfy  $$\mu_{t+1}   = \int_{\R^N} \kappa (\bx) \mu_t (\d \bx) ~~~~\mbox{for }t=1,\dots,T-1.$$
  Therefore, the distribution of a time-homogeneous Markov process with marginals $(\mu_1,\dots,\mu_T)$ corresponds to the Markov kernel $\kappa\in \mathcal K(\bmu, \bnu) $ where $\boldsymbol \mu=(\mu_1,\dots,\mu_{T-1})$,
   and $\boldsymbol \nu= (\mu_2,\dots,\mu_T)$, which is a simultaneous transport kernel.
With the tool of SOT,
we can study \emph{optimal} (in some sense) time-homogeneous Markov processes.
   %Therefore,  each time-homogeneous Markov  process with given marginals corresponds to a simultaneous transport.
    A special case of this example will be given in Proposition \ref{prop:gaussian1}.
     \label{1.3}
  %, which corresponds to $\pi\in \mathcal M(\R^2)$  given by
 % $$
 % \pi(A) = \int_{A} \kappa(x; \d y)\mu(\d x),~~~A \in \mathcal %B(\R^2).
 % $$ 
% One can verify that $\pi$ is an element of $\Pi(\boldsymbol \mu, \boldsymbol \nu)$ where $\boldsymbol \mu=(\mu_1,\dots,\mu_{T-1})$,
%  and $\boldsymbol \nu= (\mu_2,\dots,\mu_T)$.  
\end{example}

In the classic optimal transport framework with $d=1$, an unbalanced transport problem can be converted to a balanced transport problem by adjoining a point $y_0$ to the space $Y$ with mass $\mu(X)-\nu(Y)$ and such that $c(x,y_0)=0$ for all $x$. 
However, for $d\geq 2$ the two problems are not equivalent. The reason that the conversion works for $d=1$ is that the set of unbalanced transports
$$\K(\bmu,\bnu)=\{\kappa\mid \kappa_{\#}\bmu\geq \bnu\}=\{\kappa \mid\kappa_{\#}\bmu=\widetilde{\bnu} \mbox{ for some }\widetilde{\bnu}\geq \bnu\}$$
is identical to the set of transports
$$\K'(\bmu,\bnu):=\{\kappa\mid\kappa_{\#}\widetilde{\bmu}=\bnu \mbox{ for some }\widetilde{\bmu}\leq\bmu\}.$$
This is not necessarily true in case $d\geq 2$. For example, take $\mu_1=\mu_2$ be two times the Dirac measure at $0$, $\nu_1$ be uniform on $[-1,0]$ and $\nu_2$ uniform on $[0,1]$. Then the kernel $\kappa$ sending $0$ uniformly to $[-1,1]$ belongs to $\K(\bmu,\bnu)$ while $\K'(\bmu,\bnu)$ is clearly empty.
In other words, even if an unbalanced transport from $\bmu$ to $\bnu$ exists, there may not exist a way to glue mass to $\bnu$ that leads to a balanced transport. 
This subtle issue also hints on the additional technical challenges when dealing with simultaneous transports.

\subsection{Assumptions and standing notation} 

We will focus on different levels of generality in the subsequent sections, with the following 
 hierarchical structure on the imposed assumptions. As we will see, the assumption $\eta\sim\mub$ is necessary for the Kantorovich reformulation to make sense.

\begin{enumerate}[i.]
    \item In Sections \ref{S2} and \ref{Kanto} through \ref{sec:connect}, we will prove general results in the unbalanced setting;
    \item in Sections \ref{Dua} and \ref{W} we work within the balanced setting;
    \item in Section \ref{61} we further require that both $\K(\bmu,\bnu)$ and $\K(\bnu,\bmu)$ are non-empty; that is, we consider two-way transports. 
\end{enumerate} 
In terms of the reference measure, we have the following hierarchy of considerations.
\begin{enumerate}[i.]
    \item In Section  \ref{S2}  we make no further assumption on the reference measure $\eta$ except that $\eta\ll\mub$;
    \item in Section \ref{sec:Kanto}  we assume that $\eta\sim\mub$;
    \item in Section \ref{W} and throughout our examples we assume for simplicity that $\eta=\mub$.
\end{enumerate}   
The  hierarchical structure of assumptions  is summarized in Table \ref{tab:assumptions}. 
\begin{table}[!ht]\centering
\caption{Assumptions across sections}\label{tab:assumptions} \renewcommand{\arraystretch}{1.2}
\begin{tabular}{ccc} 
Section & Tuples of measures $\bmu$ and $\bnu$    & Reference $\eta$ \\  \hline
\ref{S2}       & Possibly unbalanced   &       $\eta\ll\mub$     \\  

\ref{Kanto}-\ref{sec:connect}       & Possibly unbalanced    &    $\eta\sim\mub$       \\  
\ref{Dua}       & Balanced    &    $\eta\sim\mub$       \\ 
\ref{sec:motsot}-\ref{sec:sotmot2}       & Balanced     &  $\eta=\mub$             \\  
\ref{61}      & Balanced and two-way   &   $\eta=\mub$            \\ \hline
\end{tabular}
\end{table}

Throughout, we consider the general setting where $X$ and $Y$ are Polish spaces unless otherwise stated.
We let 
$\bone_A$ stand for the indicator of a set $A$, and $\R_+:=[0,\infty)$.  The set $\M(X)$  is the collection of all finite and non-zero Borel measures on $X$.

    \section{Existence, inequalities,  and examples}\label{S2}
In the study of SOT and its structure, the   Radon--Nikodym derivatives of $\bmu,\bnu$ with respect to $\mub,\nub$ play a crucial role.  For this reason, we recall \eqref{mubnub} and introduce the shorthand notation \begin{align}
    \bmu'=\frac{\d \bmu}{\d \mub}\qquad\text{ and }\qquad\bnu'=\frac{\d \bnu}{\d \nub}.\label{mupnup}
\end{align}
We also denote by $\mm$ and $\nn$ the laws of $\bmu'$ under $\mub$ and of $\bnu'$ under $\nub$. Note that both $\mm$ and $\nn$ are probability measures on $\R^d$.

\subsection{Existence of  simultaneous transports} \label{exi}

We first state a condition to guarantee that 
$ \mathcal K(\bmu,\bnu)$
and 
$\mathcal T(\bmu,\bnu)$ are non-empty.   
  The following definition is adapted from \cite{SSWW19} where $\bmu(X)=\bnu(Y)$ is assumed. Let us emphasize that the paper \cite{SSWW19} is only related to the existence of simultaneous transports, and is independent of everything else discussed in this paper.

  \begin{definition}\label{def2} We say that  $\bmu\in\M(X)^d$  is \emph{jointly atomless} if there exists a  random variable $\xi:X\to\R $ such that  under $\mub$, $\xi$ is atomless and independent of $\bmu'$. 
 \end{definition}

   \begin{remark}
  \cite{SSWW19} called the notion of joint non-atomicity in Definition \ref{def2} as ``conditional non-atomicity".    We choose the term ``joint non-atomicity" because this notion is  indeed a collective property of $(\mu_1,\dots,\mu_d)$, and it is stronger than non-atomicity of each $\mu_j$.  There are many parallel results between non-atomicity for  $d=1$ and joint non-atomicity  for $d\ge 2$; see Remark \ref{rem:parallel}.
    \end{remark}
\begin{proposition}[\cite{T91,SSWW19}]\label{prop:ssww}
Let $\bmu\in\M(X)^d$ and $\bnu \in \M(Y)^d$.
\begin{enumerate}[(i)]
\item The set  $\mathcal K(\bmu,\bnu)$ is non-empty if and only if 
 $\mm\gicx\nn$, where $\gicx$ is the multivariate increasing convex order.\footnote
{This means $
 \int f (\bmu')\, \d   \mub  \ge   \int f   (\bnu') \,\d   \nub $  for all increasing convex $f:\R^n\to \R$ such that the integrals are well-defined. 
} 
 \item Assume that $\bmu$ is jointly atomless.  
The set $\mathcal T(\bmu,\bnu)$ is non-empty  if and only if 
 $\mm\gicx\nn$. 
\end{enumerate}  

\end{proposition}

In particular, it follows from Proposition \ref{prop:ssww} that $\mathcal K(\bmu,\bnu)$ and $\mathcal K(\bnu,\bmu)$ are  both non-empty if and only if $\mm=\nn$.   
We also note that $\mm\gicx\nn$ implies  $\bmu(X)\ge \bnu(Y)$ by taking a linear function $f(x_1,\dots,x_d) = x_j$ for $j\in [d]$ in the definition of the increasing convex order. Hence, it makes sense to discuss the set $\K(\bmu,\bnu)$ under this condition. 
  
   \begin{remark}
In Definition \ref{def2}, if $\bmu(X)=\bnu(Y)$,
then $\mm\gicx \nn$  is equivalent to 
$\mm\gcx \nn$,
where $\gcx$ is the multivariate  convex order.\footnote
{This means $
 \int f (\bmu')\, \d   \mub  \ge   \int f   (\bnu') \,\d   \nub $  for all convex functions $f:\R^n\to \R$ such that the integrals are well-defined. 
} 
 \end{remark}

\begin{remark}
To understand  $\mm\gicx\nn$  intuitively,  one could look at some special cases (treated in Proposition 3.7 of \cite{SSWW19}), by assuming $\bmu(X)=\bnu(Y)$.
Suppose that   $\K(\bmu,\bnu)$ is non-empty. 
Then if $\bmu$ has identical components, so does $\bnu$; if $\bmu$ has equivalent components,  so does $\bnu$;
 if $\bnu$ has mutually singular components, so does $\bmu$.
  Moreover, if $\bmu$ has mutually singular components  or 
    $\bnu$ has identical components, then $\K(\bmu,\bnu)$ is non-empty.
    
    In case $\bmu(X)\ge \bnu(Y)$ in which equality does not hold, simple sufficient conditions exist. For example, suppose that
    $$\min_{j \in [d]}(\mu_j(X))\geq \left(\max_{j \in [d]}\nu_j\right)(Y),$$ then $\K(\bmu,\bnu)$ is non-empty.\footnote{For a collection of (signed) measures $\mu_j$, $j\in J$ on $X$, their maximum (or supremum) is defined as $\sup_{j\in J} \mu_j (A) =\sup \{\sum_{j\in J} \mu_j (A_j)\mid \bigcup_{j\in J} A_j= A \mbox{ and $A_j$ are disjoint}\}$ for $A\subseteq X$. Moreover, the positive part of $\mu$, denoted by $\mu_+$, is   $\max\{\mu,0\}$ where $0$ is the zero measure.}  To see this, we may assume each $\mu_j$ is a probability measure. For 
    $\nu:=(\max_j \nu_j)/((\max_j \nu_j)(Y))$, we have that $\K(\bmu,(\nu,\dots,\nu))$ is non-empty since the constant kernel $x\mapsto \nu$ is in $\K(\bmu,(\nu,\dots,\nu))$. Then for $\kappa\in\K(\bmu,(\nu,\dots,\nu))$, 
    $$\kappa_{\#}\mu_j=\nu=\frac{\max_{j \in [d]}\nu_j}{(\max_{j \in [d]}\nu_j)(Y)}\geq \max_{j \in [d]}\nu_j\geq \nu_j.$$This shows $\kappa\in\K(\bmu,\bnu)$.
\end{remark}

\begin{remark}
    The converse of Proposition \ref{prop:ssww}(ii) does not hold even if $\bmu(X)=\bnu(Y)$. There are examples where $\bmu$ is not jointly atomless but there exists a unique Kantorovich transport that is also Monge. A trivial example could be that $d=1$ and $\mu=\nu=\delta_0$. A more meaningful example is given by Theorem 6.1 of \cite{W19}, which states that if the space $Y$ is finite, the jointly atomless condition can be relaxed to the atomless condition.  
\end{remark}

We record an immediate corollary of Proposition \ref{prop:ssww} for the subsequent analysis, which can also be shown by directly using definition.
\begin{corollary}\label{comp}
Suppose that $\bmu,\bnu,\bta$ are $\R^d_+$-valued probability measures on Polish spaces such that $\K(\bmu,\bnu)$ and $\K(\bnu,\bta)$ are non-empty. Then $\K(\bmu,\bta)$ is non-empty. 
\end{corollary}

Proposition \ref{prop:ssww} can also be applied to give a necessary condition for the existence of a time-homogeneous Markov process (see Example \ref{1.3}) for centered Gaussian marginals on $\R$. 

\begin{proposition}\label{prop:gaussian1}
Suppose that  $\mu_t = \mathrm{N}(0,\sigma^2_t)$, $\sigma_t >0$, $t=1,\dots,T$.
For the existence of a transport from $(\mu_1,\dots,\mu_{T-1})$ to $(\mu_2,\dots,\mu_T)$, it is necessary that the mapping
$t\mapsto \sigma_t$ on $ \{1,\dots,T\}$ is increasing log-concave or decreasing log-convex.  If $T=3$, this condition is also sufficient.
%(It is one of the following cases:
%(i)  $ \sigma_1\le \sigma_2\le \sigma_3$  and $\sigma_1 \sigma_3 \le \sigma_2^2$; 
% (ii) $ \sigma_1\ge \sigma_2\ge \sigma_3 $ and $\sigma_1\sigma_3 \ge \sigma_2^2$.)
\end{proposition}

The necessary condition in Proposition \ref{prop:gaussian1} is not sufficient for $T>3$. See Appendix \ref{D} for a counterexample.
In the case $T=3$, the Markov process in Proposition \ref{prop:gaussian1} can be realized by an AR(1) process with Gaussian noise.

    \subsection{Some simple lower bounds for on the minimum cost}\label{21}
    
        We  collect some lower bounds for the infimum cost based only on classic ($d=1$) transports. Since every $\kappa\in \mathcal K(\bmu,\bnu)$ transports each $\mu_j$ to cover $\nu_j$, it must transport each $\boldsymbol\lambda\cdot\bmu$ to cover  $\boldsymbol\lambda\cdot\bnu$ for $\boldsymbol\lambda\in \R_+^d$. Denoting by $\Delta_d$ the standard simplex in $\R^d$, we have
    $$\mathcal K(\bmu,\bnu)\subseteq \bigcap_{\boldsymbol \lambda \in  \Delta_d}\mathcal K(\boldsymbol\lambda\cdot\bmu,\boldsymbol\lambda\cdot\bnu).$$Therefore, we obtain
 \begin{align}\inf_{\kappa\in \mathcal K(\bmu,\bnu)}  \mathcal C_\eta(\kappa) 
 &\ge \sup_{\boldsymbol \lambda \in \Delta_d} \inf_{\kappa\in \mathcal K(\boldsymbol\lambda\cdot\bmu,\boldsymbol\lambda\cdot\bnu)}  \mathcal C_\eta(\kappa).\label{ineq1}
 \end{align}In particular, if $\kappa\in\K (\bar{\mu},\bar{\nu})$ is an optimal transport from $\bar{\mu} $ to $\bar{\nu}$ and $\kappa\in\K(\bmu,\bnu)$, then $\kappa$ is also an optimal transport from $\bmu$ to $\bnu$. 
However, as we will see in Example \ref{ex1}, the inequality (\ref{ineq1}) is not sharp in general.

We record yet another lower bound for the minimum cost as an application of the kernel formulation. For simplicity we consider the balanced setting. The following proposition follows intuitively by observing that, for example in case $d=2$, the parts where $\nu_1\geq \nu_2$ must be transported from the parts where $\mu_1\geq \mu_2$ (see Figure \ref{fig:1}).

\begin{figure}[hbtp]\begin{center} 
\vspace{-1cm} \begin{tikzpicture}
\centering
\begin{axis}[width=\textwidth,
      height=0.45\textwidth,
  axis lines=none,% 
  domain=-10:10,
  xmin=-3, xmax=10,
  ymin=0, ymax=5,
  samples=80
]
  \addplot[name path=f,blue,domain={-3:3},mark=none]
    {3*exp(-9*x^2)}node[pos=0.5,above]{$\mu_1$};
  \addplot[name path=g,purple,domain={-3:3},mark=none]
    {exp(-x^2)}node[pos=0.5,below]{$\mu_2$};
    \addplot[pattern=south east lines, pattern color=black!50]fill between[of=f and g, soft clip={domain=-0.37:0.37}];
    \addplot[gray!50]fill between[of=f and g, soft clip={domain=-3:-0.37}];
    \addplot[gray!50]fill between[of=f and g, soft clip={domain=0.37:3}];
    
      \addplot[name path=p,blue,domain={4:10},mark=none]
    {2*exp(-4*(x-7)^2)}node[pos=0.5,above]{$\nu_1$};
  \addplot[name path=q,purple,domain={4:10}]
    {exp(-(x-7)^2)}node[pos=0.5,below]{$\nu_2$};
    \addplot[pattern=south east lines, pattern color=black!50]fill between[of=p and q, soft clip={domain=6.52:7.48}];
    \addplot[gray!50]fill between[ of=p and q, soft clip={domain=7.48:10}];
    \addplot[gray!50]fill between[of=p and q, soft clip={domain=4:6.52}];

\end{axis}
\draw[->](4.6,1.3)--(8.8,1.3); 
\draw[->](5.5,0.4)--(7.9,0.4);
\end{tikzpicture}\end{center}
\caption{Part of the shaded region $(\mu_1-\mu_2)_+$ on the left is transported to cover all of the shaded region $(\nu_1-\nu_2)_+$ on the right; similarly, part of the gray region $(\mu_2-\mu_1)_+$ is transported to cover all of the gray region $(\nu_2-\nu_1)_+$.}
\label{fig:1}
\end{figure}
%\com{I revised it. Check whether it is correct. I need $c(x,y)=0$ for some $y$ for each $x$ to make sure the cost is $0$ on $X_j^c$ for each $j$. Otherwise it is not correct. The original formulation is also not correct. Complete the proof please.......I don't understand what you mean here: the supports of $(\mu_j-\mub)_+$ are not disjoint, so how could you sum up the costs?? You are right, I revised..... I understand now. First, $\eta$ should be $\eta_j$ restricted on each pieces like (8) below; second, if $d\geq 3$, is this inequality stronger than the bound achieved by picking two of the measures and applying (8)? I feel it's not (for example, when anywhere you have two measures among the $d$ that attain maximum.), that's why I didn't write down this version before; third: why you need $c(x,y)=0$ for some $y$ for each $x$??}
\begin{proposition}
\label{2infp}
Suppose that $\bmu(X)=\bnu(Y) $, and for each $x\in X$, $c(x,y)=0$ for some $y\in Y$. Then \begin{align} 
    &\inf_{\kappa\in \mathcal K(\bmu,\bnu)}  \mathcal C_{\eta}(\kappa)   \ge\max_{i,j\in [d]} \left(\inf_{\substack{\kappa\in \K((\mu_i- \mu_{j})_+,(\nu_i-   \nu_{j})_+)}} \C_{\eta}(\kappa) +\inf_{\substack{\kappa\in \K((\mu_j-  \mu_{i})_+,(\nu_j-  \nu_{i})_+)}} \C_{\eta}(\kappa) \right) .\label{2inf}
\end{align} 
%where $\widetilde  \mu_{j}=\max_{i\ne j} \mu_i$ and 
%$\widetilde  \nu_{j}=\max_{i\ne j} \nu_i$ for each $j\in [d]$.
In particular, if $(\mu_i-\mu_j)_+(X)<(\nu_i- \nu_j)_+(Y)$ for some $i,j\in [d]$, then both sides of \eqref{2inf} are equal to $\infty$.
\end{proposition}
% \begin{proposition}
% \label{2infp}
% Consider $d=2$ and $\bmu(X)=\bnu(Y)$. Denote by $\eta_1$ the restriction of $\eta$ on the set $\{x\in X\mid \mu_1'(x)\geq 1\}$ and $\eta_2=\eta-\eta_1$. Then\begin{align} 
%     \inf_{\kappa\in \mathcal K(\bmu,\bnu)}  \mathcal C_{\eta}(\kappa) 
%  \ge \inf_{\substack{\kappa\in \K((\mu_1-\mu_2)_+,(\nu_1-\nu_2)_+)}} \C_{\eta_1}(\kappa)    +\inf_{\substack{\kappa\in \K((\mu_2-\mu_1)_+, (\nu_2-\nu_1)_+) }} \C_{\eta_2}(\kappa)    .\label{2inf}
% \end{align} 
% In particular, if $(\mu_1-\mu_2)_+(X)<(\nu_1-\nu_2)_+(Y)$ or $(\mu_2-\mu_1)_+(X)<(\nu_2-\nu_1)_+(Y)$, then the left-hand side of \eqref{2inf} is equal to $\infty$.
% \end{proposition}

 Note that the quantities on the right-hand side of \eqref{2inf} arise from two separate one-dimensional transport problems. Such problems are well-studied in the optimal transport literature; see \cite{S15,V03,V09}.

If $\mu_1,\dots,\mu_d$ have mutually disjoint supports (in particular, if $d=1$), then the simultaneous transport problem is reduced to $d$ classic transport problems and the optimal cost is the sum of corresponding optimal costs. In this case,  \eqref{ineq1} is  sharp, and \eqref{2inf} is also sharp when $d=2$.

\subsection{Peculiarities of the simultaneous transport}
    
 We consider a few simple but instructional examples showing that simultaneous transport is very different from classic transport ($d=1$). We will focus on the balanced case (i.e., $\bmu(X)=\bnu(Y)$) for simplicity.
    
    \label{sec:12}
    We first provide some immediate observations which help to explain some novel features of simultaneous transport and its connection to the classic optimal transport problem. Denote by $\mu'_j(x),\nu'_j(y)$ the corresponding Radon--Nikodym derivatives of $\mu_j,\nu_j$ with respect to $\mub,\nub$, respectively.

    First, suppose that $\eta\sim\mub$. If there exist measurable functions $\phi$ on $X$ and $\bpsi=(\psi_1,\dots,\psi_d)$ on $Y$ such that 
    $$c(x,y)=\phi(x)+\bpsi(y)^{\top}\frac{\d\bmu}{\d\eta}(x),$$ then all transports (should any exist) from $\bmu$ to $\bnu$ have the same cost 
    \begin{align}
        \label{eq:decompose}
    \int_{X\times Y} c \, \d (\eta\otimes \kappa)= \int_X\phi\,\d\eta+\int_Y\bpsi^{\top}\d\bnu,\end{align}
    because $\kappa_\#\bmu=\bnu$.
    This extends the fact that in the case $d=1$, the cost functions of the form $c(x,y)=\phi(x)+\psi(y)$ are trivial and can be ``decomposed into marginal costs". We now have a larger class of such cost functions.  If $\eta=\mub$, then 
    a term $\psi(y)$ for $\psi:Y\to \R$
    can also be included in $c(x,y)$, by noting that
    $$\psi(y) = \frac{1}{d}\psi(y)\one\cdot \frac{\d \bmu}{\d\mub}(x).$$
    Moreover, \eqref{eq:decompose} also hints on how a duality result would look like in this setting, which will be discussed in Section \ref{Dua}.

    \begin{example}\label{xy}
    Consider $X=\R$ on which Borel probability measures $\bmu$ are supported and $\eta=\mub$. Assume that $\mu_1'$ is linear in $x\in\R$ on the support of $\mub$, say, equal to $ax+b$, $ a\ne 0$.   Let $\bnu$ be probability measures on $\R$ such that $\K(\bmu,\bnu)$ is non-empty. Consider a quadratic cost function $c(x,y)=(x-y)^2$. Then we may write
    $$c(x,y)=x^2+(ax+b)\left(-\frac{2y}{a}\right)+\left(y^2+\frac{2by}{a}\right).$$
    Therefore, for any $\kappa\in\K(\bmu,\bnu)$,
    $$\C(\kappa)=\int_Xx^2\mub(\d x)+\int_Y\left(y^2+\frac{2by}{a}\right)\nub(\d y)-\frac{2}{a}\int_Yy\, \nu_1(\d y).$$
    \end{example}
    
    \begin{example}
As a concrete but slightly more general example, we consider $X=Y=[0,1]$ on which Borel probability measures $\mu_j, \nu_j,\ j=1,2$ are supported. Assume $\mu_1$ has density $2x$ and $\mu_2$ has density $2-2x$ with respect to Lebesgue measure on $[0,1]$, and $\nu_1=\nu_2$ be any identical probability measures on $[0,1]$ such that $\nu_1((1/4,3/4))=1/2$ (see Figure \ref{fig:4}). Thus the Radon--Nikodym derivatives are $\mu'_1(x)=2x$ and $\nu'_1(y)=1$. Denote the set $A=(1/4,3/4)\times([0,1/4)\cup (3/4,1])$ and consider the cost function $$c(x,y)=(x-y)^2+\alpha\bone_{A},\ \alpha>0.$$
For any $s\in[0,1/2]$ and any $S$ such that $\nu_1(S)=1-2s$, the transport kernel
$$\kappa(x;B):=\frac{\nu_1(B\cap S)}{\nu_1(S)}\bone_{\{x\in(s,1-s)\}}+\frac{\nu_1(B\setminus S)}{\nu_1([0,1]\setminus S)}\bone_{\{x\in[0,s]\cup[1-s,1]\}}$$
 belongs to $\K(\bmu,\bnu)$.
 In case $S=(1/4,3/4)$ and $s=1/4$, we denote such a transport by $\kappa_0$. 

We show that $\kappa_0$ is indeed an optimal transport. Similarly as in Example \ref{xy}, for a kernel $\kappa\in\K(\bmu,\bnu)$ we compute its transport cost
$$\C(\kappa)=\frac{1}{3}+\int_0^1(y^2-y)\nub(\d y)+\alpha(\mub\otimes \kappa)(A)\geq \frac{1}{3}+\int_0^1(y^2-y)\nub(\d y),$$
where inequality holds if and only if $(\mub\otimes \kappa)(A)=0$. Since by definition $\kappa_0(x;(1/4,3/4))=1$ for $x\in(1/4,3/4)$, we have $(\mub\otimes \kappa_0)(A)=0$. Therefore, $\kappa_0$ is an optimal transport.

 \label{ex1}
% {Can generalize to $c(x,y)=|x-y|^p$?}
\begin{figure}[h!]\begin{center}\resizebox{0.8\textwidth}{0.33\textwidth}{
\begin{tikzpicture}

\begin{axis}[axis lines=middle,
            enlargelimits,
            ytick=\empty,
            xtick=\empty,]
\addplot[name path=F,purple,domain={0:1}] {2*x} node[pos=.9, above]{$\mu_1$};

\addplot[name path=G,blue,domain={0:1}] {2-2*x}node[pos=.1, above]{$\mu_2$};

\addplot[name path=H,black,domain={0:1}] {1}node[pos=0.85, above]{$\eta=\bar\mu$};

\end{axis}
\end{tikzpicture}\hspace{1cm}\begin{tikzpicture}

\begin{axis}[axis lines=middle,
            enlargelimits,
            ytick=\empty,
            xtick=\empty, samples=180,]
\addplot[name path=F,purple,very thick,domain={0:1}] {0.8+0.4*(sin(deg(6.28*(x-3))))+0.3*(sin(deg(6.28*5*x)))+x^2} node[pos=.999, above]{$\nu_1=\nu_2$};
\addplot[name path=e,blue,domain={0:1}] {0.8+0.4*(sin(deg(6.28*(x-3))))+0.3*(sin(deg(6.28*5*x)))+x^2} node[pos=.999, above]{};
%\addplot[name path=e,blue,domain={0:1}] {1+0.7*(sin(deg(6.28*(x-3))))+0.3*(sin(deg(6.28*5*x)))+x} node[pos=.9, above]{};
\addplot[name path=e2,black,domain={0.25:0.75}] {0} node[pos=.85, above]{\begin{minipage}{0.3\textwidth} total mass of the\\ gray area is $1/2$\end{minipage}};
    % Shade the area between the function and the x-axis
    \addplot[black!20] fill between[of=F and e2, soft clip={domain=0.25:0.75}];

\end{axis}
\end{tikzpicture}

}\end{center}

\caption{Densities of $\bmu$ and $\bnu$ in Example \ref{ex1}.}
\label{fig:4}
\end{figure}
\end{example}
Trivial as it looks, Examples \ref{xy} and \ref{ex1} provide us with some interesting aspects of the SOT in contrast to the classic optimal transport.
\begin{enumerate}[i.]
    \item It is well-known that for the classic Kantorovich transport problem in $\R$, if the cost function is a convex function in $y-x$,  then the comonotone map is always optimal (see e.g., Theorem 2.9 of \cite{S15}). However, this effect no longer exists in simultaneous transport, since there may not exist an admissible comonotone map.
    \item It is also easy to see that the equality in (\ref{ineq1}) may not hold, for example when $\nu_1$ is uniform on $[0,1]$. In addition, the inequality (\ref{2inf}) becomes trivial since it gives a lower bound $0$.
\end{enumerate} 

    After developing our theory, we discuss a few more interesting examples in Section \ref{61}.

 \section{General properties of simultaneous optimal transport}\label{sec:Kanto}

Recall we consider $d$-tuples of finite measures $\bmu=(\mu_1,\dots,\mu_d)$ on $X$ and $\bnu=(\nu_1,\dots,\nu_d)$ on $Y$, and a reference measure $\eta\sim\mub$. Also  recall that  $\bmu',\bnu'$ denote the Radon--Nikodym derivatives of $\bmu,\bnu$ with respect to $\mub,\nub$. 

\subsection{The Kantorovich formulation}\label{Kanto}
Sometimes it is mathematically more convenient to adopt the Kantorovich formulation, which describes the set of all transport plans as probability measures in $\P(X\times Y)$. More precisely, for a probability measure $\eta\sim\mub$, we define
$$\Pi_\eta(\bmu,\bnu):=\{\eta\otimes\kappa\mid\kappa_\#\bmu\geq\bnu\}.$$The subscript $\eta$ incorporates the way we calculate costs: see (\ref{eq:cost}) and (\ref{eq:cost2}). It is immediate that 
\begin{align}\inf_{\pi\in\Pi_\eta(\bmu,\bnu)}\C(\pi):=\inf_{\pi\in\Pi_\eta(\bmu,\bnu)}\int_{X\times Y}c(x,y)\pi(\d x,\d y)=\inf_{\kappa\in \K(\bmu,\bnu)}\mathcal C_\eta(\kappa).\label{sameinf}\end{align}

Equivalently, we have the following reformulation for $\Pi_\eta(\bmu,\bnu)$. We call this the Kantorovich reformulation, whose reasons are explained below.
\begin{proposition}\label{prop4.1}
For each $\eta\sim\mub$, we have
\begin{align}
\Pi_\eta(\bmu,\bnu)&=\Bigg\{\pi\in\mathcal P(X\times Y)\mid  \pi(\d x \times Y)=\eta (\d x) \text{   and} \int_{X  }\frac{\d \bmu}{\d\eta}(x)\pi(\d x,\d y)\geq\bnu (\d y) \Bigg\}.\label{newpi}
\end{align} 
\end{proposition}

In a  way  similar to Proposition \ref{prop4.1}, in the balanced case, i.e., $\bmu(X)=\bnu(Y)$, we have
\begin{align}
\Pi_\eta(\bmu,\bnu)&=\Bigg\{\pi\in\mathcal P(X\times Y)\mid  \pi(\d x\times Y)=\eta(\d x)\text{  and }\int_{X}\frac{\d \bmu}{\d\eta}(x)\pi(\d x,\d y)=\bnu(\d y)\Bigg\}.\label{o1}
\end{align} In particular, if  $\eta=\mub$, we denote by $\Pi(\bmu,\bnu)=\Pi_{\mub}(\bmu,\bnu)$,
and \eqref{o1} reads as
\begin{align}
    \Pi(\bmu,\bnu)&=\Bigg\{\pi\in\mathcal P(X\times Y)\mid  \pi(\d x\times Y)=\mub(\d x)\text{  and }\int_{X}\bmu'(x)\pi(\d x,\d y)=\bnu(\d y)\Bigg\}.\label{o}
\end{align} 

%Proposition \ref{eta} implies that the set of transports in the Kantorovich formulation does not depend on which reference reference measure $\eta$ we choose as long as they are equivalent measures. This is as expected since $\T(\bmu,\bnu)$ and $\K(\bmu,\bnu)$ do not explicitly involve the measure $\eta$. On the other hand, the Kantorovich formulation has a limitation compared to the kernel formulation we introduced before: it cannot take care of measures $\eta$ with respect to which $\mu$ is not absolutely continuous.
    
It seems worthwhile to explain the similarities and differences of (\ref{o}) compared to the classic definition $\Pi(\mu,\nu)$ in the case $d=1$ (see (\ref{c}) below). First, by summing over and normalizing the second constraint in (\ref{o}), we see that $\pi$ is a transport from $\mub$ to $\nub$. Thus, one may think of $\pi(A\times B)$ as the amount of $\mub$-mass moving from $A$ to $B$. With $j\in [d]$ fixed, the second constraint in (\ref{o}) means that the mass sent from the contribution of $\mu_j$ covers exactly the corresponding portion of $\nu_j$ in $Y$. 

We can reformulate \eqref{o} as \begin{align}\Pi(\bmu,\bnu)&=\Bigg\{\pi\in\mathcal P(X\times Y)\mid \int_{X\times Y}f(x)\pi(\d x,\d y)=\int_Xf\, \d \mub \text{ and}\nonumber\\
 &\qquad \int_{X\times Y}\bmu'(x)g(y)\pi(\d x,\d y)=\int_Yg \,\d \bnu \text{ for all measurable }f,g\Bigg\}.\label{oo}\end{align}
In the case $d=1$, our formulation coincides with the classic Kantorovich formulation, where the admissible transports are defined as
\begin{align}
    \widetilde{\Pi}(\mu,\nu):=\{\pi\in\mathcal P(X\times Y)\mid  \pi(\d x\times Y)=\mu(\d x)\text{ and } \pi(X\times \d y)=\nu(\d y)\}.\label{c}
\end{align}In some sense, one can also recover transports in $\widetilde{\Pi}(\mu_j,\nu_j)$ from $\Pi(\bmu,\bnu)$. For example, taking $f(x)=\bone_{\{x\in A\}}\mu'_j(x)$ and $g(y)=\bone_{\{y\in B\}}$ in (\ref{oo}), we have for any $j \in [d]$, the measure $\mu'_j(x)\pi(\d x,\d y)$ belongs to $\widetilde{\Pi}(\mu_j,\nu_j)$.

Unlike the classic Kantorovich optimal transport problem in the case $d=1$, the simultaneous transport problem is not symmetric with respect to the measures $\bmu,\bnu$, as expected from \Cref{prop:ssww}. It seems unlikely that $\Pi(\bmu,\bnu)$ can be defined in a similar way as (\ref{c}) using only projections of measures. The fact that the classic Kantorovich formulation uses projections and is symmetric, is nothing more than a nice consequence of the kernel formulation and does not reflect the general structure.

%In general, it is not required that $\mu_j$ have the same total mass for each $j \in [d]$. As long as  $\bmu(X)=\bnu(Y)$, it is still possible that $\Pi(\bmu,\bnu)$ is non-empty.

\subsection{Equivalence between Monge and Kantorovich costs}  \label{rel}
 
 Below, we prove that under suitable conditions, the set of transport maps and plans have the same infimum cost.  This serves as an extension of Theorem 2.1 of \cite{A03} in the case $d=1$ and  we also assume for simplicity that $\mub,\nub$ have compact supports. As expected from Proposition \ref{prop:ssww}, joint non-atomicity plays an important role since it guarantees the existence of Monge transports.

 We first prove the following more general result using the kernel formulation. Observe that for a Monge transport $T\in\T(\bmu,\bnu)$, we can associate a kernel $\kappa_T\in\K(\bmu,\bnu)$ defined by $\kappa_T(x;B):=\bone_{\{T(x)\in B\}}$. In view of (\ref{eq:cost}) and (\ref{eq:cost2}), they have the same transport cost.
\begin{theorem}[Cost equality]
\label{infs}
Let $\eta\sim\mub$. Suppose that $X,Y$ are compact spaces on which $\bmu,\bnu$   are supported, $\bmu$ is jointly atomless, and $c$ is continuous. Then the transport plans and transport maps admit the same infimum cost. That is,
$$\inf_{\kappa\in \K(\bmu,\bnu)}\C_\eta(\kappa)= \inf_{T\in \mathcal T(\bmu,\bnu)} \mathcal C_\eta(T).$$
\end{theorem}

Combining with (\ref{sameinf}) yields the following.
\begin{corollary}\label{2infss}
Consider a reference measure $\eta\sim\mub$. Suppose that $X,Y$ are compact spaces on which $\bmu,\bnu$ are supported, $\bmu$ is jointly atomless, and $c$ is continuous, then Monge and Kantorovich transport costs have the same infimum value. That is,
$$\inf_{\pi\in\Pi_\eta(\bmu,\bnu)}\C(\pi)= \inf_{T\in \mathcal T(\bmu,\bnu)} \mathcal C_\eta(T).$$
\end{corollary} 

The proof of Theorem \ref{infs} follows a similar path as the classic result in the case $d=1$, except that we need a few new lemmas on joint non-atomicity. Recall from the classic proof that non-atomicity allows us to  approximate a transport plan using a transport map on each small piece of $X$. In our setting, we need joint non-atomicity to achieve this; see Proposition \ref{prop:ssww}.

\begin{remark}\label{rem:parallel}
Heuristically, there is a parallel between  non-atomicity  in the classic setting   and joint non-atomicity  in our setting. For example, \begin{enumerate}[i.]
    \item Under joint non-atomicity, a Monge transport exists if and only if a Kantorovich transport exists (Proposition \ref{prop:ssww}). In the case $d=1$ with non-atomicity, this equivalence also holds, although a Kantorovich transport between $\mu$ and $\nu$ exists as soon as $\mu$ and $\nu$ have the same mass. 
%    Under joint non-atomicity we have the equivalence of existences of Monge and Kantorovich transports, which is the case when $d=1$ assuming only marginal non-atomicity. 
    
    \item Marginal non-atomicity is equivalent to the existence of a uniform random variable and joint non-atomicity is equivalent to the existence of a uniform random variable independent of a $\sigma$-field (Lemma \ref{lem:useful}).
    \item The joint non-atomicity condition enables us to conclude Monge and Kantorovich problems have the same infimum (Corollary \ref{2infss}), which is true in the case $d=1$ given marginal non-atomicity.
\end{enumerate}
\end{remark}

 \subsection{Connecting the balanced and unbalanced settings}
 \label{sec:connect}
 So far we have discussed SOT in the unbalanced setting.
In real applications such as the setting of Example \ref{ex:1}, it likely holds that $\bmu(X)\ge \bnu(Y)$ with strict inequality in some components. 
 For instance, in an economy, the total demand for each product may be approximately $95\%$ of the total supply, leading to $\bnu(Y)\approx 0.95\times  \bmu(X)$.

 As we will see in the subsequent sections, 
results on duality, equilibria,  and the MOT-SOT parity  
will be obtained in the setting of balanced transport, since the balanced setting has a much richer mathematical structure than the unbalanced setting. 

Nevertheless, we show below that the balanced setting of simultaneous transport can be used as an approximation of the unbalanced setting. 
A special situation is when $\bnu(Y)\approx  (1-\epsilon)\times  \bmu(X)$ for a small $\epsilon>0$, which is more realistic in applications. 

Suppose that  $\bnu^n\le \bnu$ for $n\in \N$ and $\bnu^n\to \bnu$ weakly as $n\to \infty$. 
By definition, $\K(\bmu,\bnu)\subseteq \K(\bmu,\bnu^n)$, which means that each transport from $\bmu$ to $\bnu$ is also a transport from $\bmu$ to $  \bnu^n$.  
Moreover,  under a continuity assumption, the minimum transport cost from $\bmu$ to $\bnu$ is the limit of that from $\bmu$ to $\bnu^n$.  
Therefore, an optimal transport from $\bmu$ to $\bnu$ can be seen as a nearly optimal transport from  $\bmu$ to $\bnu^n$. 
Note that $\bmu(X)=\bnu(Y)$ is not needed for this continuity result.

 \begin{proposition}\label{connect}
 Suppose that  $X,Y$ are compact Polish spaces, $\bmu\in\M(X)^d$, and
  $\d\bmu/\d\eta$ and $c $ are  continuous. Suppose that $(\bnu^n)_{n\in\N}\subseteq \M(Y)^d$ is a sequence of measures converging weakly to $\bnu\in\M(Y)^d$ such that $\bnu^n\leq \bnu$ for each $n\in \N$. Then
 $$\lim_{n\to\infty}\inf_{\pi\in\Pi_{\eta}(\bmu,\bnu^n)}\C(\pi)=\inf_{\pi\in\Pi_{\eta}(\bmu,\bnu)}\C(\pi).$$
 \end{proposition}
 
Proposition \ref{connect} provides a link between two settings, allowing us to use results in  the balanced setting to approximate the unbalanced setting.  
 Starting from the next section, we concentrate on the balanced setting.

\subsection{Duality for simultaneous optimal transport}\label{Dua}
Consider $\R^d_+$-valued measures $\bmu,\bnu$ on Polish spaces $X,Y$ satisfying $\bmu(X)=\bnu(Y)$ (e.g., when they are probability measures), a reference probability  $\eta\sim\mub$, and $\bmu,\bnu$ are absolutely continuous with respect to $\mub, \nub$ with densities  $\bmu'$ on $X$ and $\bnu'$ on $Y$ respectively. Also, recall that (\ref{o1}) is the set of all transport plans from the vector-valued measure $\bmu$ to the vector-valued measure $\bnu$. 

We give a duality theorem for SOT on Polish spaces. A detailed proof will be provided in Appendix \ref{AA}. 
\begin{theorem}[Duality]
Suppose that $X,Y$ are Polish spaces, $\eta\sim\mub$ with both $\d\bmu/\d\eta$ and $\d\eta/\d\mub$ bounded  continuous, and $c:X\times Y\to[0,\infty]$ is lower semi-continuous.\footnote{Recall that a function $f$ is lower semi-continuous if and only if for any $y\in\R$, $\{\bx\mid f(\bx)>y\}$ is open.} 
Duality holds as
\begin{align} \inf_{\pi\in\Pi_\eta(\bmu,\bnu)}\int_{X\times Y}c \, \d \pi =\sup_{(\phi,\bpsi)\in \Phi_c}\int_X\phi\,  \d\eta+\int_Y \bpsi^\top \d \bnu, \label{dua}\end{align}
where  \begin{align*} \Phi_c  =\Bigg\{(\phi,\bpsi)\in C(X)\times C (Y)^d\mid \phi (x)+\bpsi(y)\cdot 
\frac{\d\bmu}{\d\eta}(x)\le c(x,y)\Bigg\}.\end{align*}  
Moreover, the infimum in \eqref{dua} is attained.
\label{duality1}
\end{theorem}

\begin{remark}
   Under more restrictive assumptions ($X,Y$ compact and $c$ continuous), Theorem \ref{duality1} also follows from an abstract duality theorem established in \cite{G20} for Banach space-valued measures. 
\end{remark}

In the case $d=1$ and $\eta=\mu$, this recovers Theorem 1.3 in \cite{V03} under the assumption of compactness. If $\eta=\mub$ and $\bmu'$ is upper semi-continuous, this result is a special case of the more general moment-type duality formula; see \cite{R98}.

 If $\d\eta/\d\mub$ is bounded (e.g., when $\eta=\mub$), even if $X,Y$ are not compact, we still have $\Pi_{\eta}(\bmu,\bnu)$ is tight and hence weakly relatively compact.  This follows from the definition of tightness and$$\nub(B)=\int_{X\times B}\frac{\d\mub}{\d\eta}(x)\pi(\d x,\d y)\geq \left(\sup_{x\in X}\frac{\d\eta}{\d\mub}(x)\right)^{-1}\pi(X\times B).$$ 
Furthermore, if $\eta=\mub$, $\Pi(\bmu,\bnu)$ is weakly compact if $\bmu'$ is assumed to be continuous, as can be seen by taking limits in (\ref{oo}).
% In particular, we will see in the proof in Section \ref{pf51} that in this case the attainability of the infimum in \eqref{dua} does not require the spaces $X,Y$ to be compact. 

% First, a similar duality result holds if  $c(x,y)\frac{\d\eta}{\d\mu}(x)$ is lower semi-continuous.\com{How is the LHS well-defined when $\eta,\mu$ are not equivalent??????} In this case, we have 
% \begin{align*}\inf_{\pi\in\Pi_\eta(\bmu,\bnu)}\int_{X\times Y}c(x,y)\pi(\d x,\d y)=\sup_{(\phi_0,\dots,\phi_d)\in \Phi_c}\int_X\phi_0\d\eta+\int_Y \bpsi^\top\d\bnu,\end{align*}where $\bphi=(\phi_1,\dots,\phi_d)$ and \begin{align*}\Phi_c=\Bigg\{(\phi_0,\bphi)\in C(X)\times C^d(Y):\ \phi_0(x)+\bphi(y)^{\top}\frac{\d\bmu}{\d \eta}(x)\leq c(x,y)\frac{\d \eta}{\d\mu}(x)\Bigg\}.\end{align*}This can be seen by applying Theorem \ref{duality1} with $\eta=\mu$ and replacing $c(x,y)$ by $c(x,y)\frac{\d\eta}{\d\mu}(x)$. 

In the case where $\eta\ll\mub$ but $\mub\not\ll\eta$, we still have the lower bound
\begin{align*}\inf_{\kappa\in\mathcal K(\bmu,\bnu)}\int_{X\times Y}c(x,y)\eta \otimes \kappa(\d x,\d y)\geq \sup_{(\phi,\bpsi)\in \Phi_c}\int_X\phi\, \d\eta +\int_Y \bpsi^\top \d \bnu  ,\end{align*}
where
\begin{align*} \Phi_c 
&:=\Bigg\{ (\phi,\bpsi)\in C(X)\times C^d(Y)\mid  \phi(x)\d\eta(x)+\bpsi(y)^{\top} \d\bmu(x) \leq c(x,y)\d\eta(x)\Bigg\}.\end{align*}
This is because for $(\phi,\bpsi)\in \Phi_c$ and $\kappa\in\mathcal K(\bmu,\bnu)$,
\begin{align*}
    \int_X\phi\, \d\eta +\int_Y \bpsi^{\top}\d \bnu &= \int_{X\times Y}\phi(x)\eta\otimes \kappa(\d x,\d y)+\int_{X\times Y}\bpsi(y)^{\top}\bmu\otimes \kappa(\d x,\d y)\\
    &\leq \int_{X\times Y}c(x,y)\eta \otimes \kappa(\d x,\d y).
\end{align*}

\section{MOT-SOT parity}\label{W}

Consider $\bmu\in\M(X )^d$ and $\bnu\in\M(Y )^d$ where for simplicity $X,Y$ are Euclidean spaces, and recall our notation  $\mub,\nub,\bmu',\bnu',\mm,\nn$ from \eqref{mubnub} and \eqref{mupnup}. By the disintegration theorem, there exist measures $\{\mu_{\bf z}\}_{{\bf z}\in\R_+^d}$ such that
$$\mu_{\bf z}(X \setminus A_{\bf z}):=\mu_{\bf z}\left(X \setminus (\bmu')^{-1}({\bf z})\right)=0$$
and for any Borel measurable function $f:X \to[0,\infty)$,
$$\int_X  f(x)\mub(\d x)=\int_{\R_+^d}\int_{A_{\bf z}}f(x)\mu_{\bf z}(\d x)\mm(\d {\bf z}).$$
Moreover, the family of measures $\{\mu_{\bf z}\}_{{\bf z}\in\R_+^d}$ is uniquely determined for $\mm$-a.s.~${\bf z}\in\R_+^d$. Similarly for ${\bf z'}\in\R_+^d$ we define $B_{\bf z'}\subseteq Y $ and a probability measure $\nu_{\bf z'}$ on $Y $.

We recall that given probability measures $\mu,\nu$, a coupling $(\xi_\mu,\xi_\nu)$ of $\mu,\nu$ is called a \emph{martingale transport} if $(\xi_\mu,\xi_\nu)$ forms a martingale, and we denote by $\M(\mu,\nu)$ the set of all such couplings, which can be further identified as stochastic kernels.

We have seen from Proposition \ref{prop:ssww} that the existence of a (balanced) simultaneous transport from $\bmu$ to $\bnu$ is equivalent to 
 $\mm\gcx\nn$. 
This naturally gives rise to a martingale transport from $\nn$ to $\mm$ in view of Strassen's theorem (\cite{S65}). Such a martingale transport, seen as a coupling, encodes the way we take combinations of the ($\R^d_+$-valued) derivatives $\bmu'$ to form the derivatives $\bnu'$. The martingale constraint is equivalent to the constraint of mixing $\bmu'$ to get $\bnu'$. Essentially, at this step, we do not ``distinguish" the points in $A_\bz$ since $\bmu'$ is constant there, but treat the set $A_{\bz}$ as a single point; the same applies to $B_\bz$.  Of course, such a martingale transport may not be unique (in fact, for two-way transports it is unique). This choice of a martingale transport can be seen
as the first layer of freedom for a simultaneous transport from $\bmu$ to $\bnu$. The second layer of freedom is how to  transport on each slice from  $A_\bz$ to $B_\bz$, where now we do not treat them as single points, but equip the measure $\mu_\bz,~\nu_{\bz}$ on them. In comparison, the martingale transport treats them as points, which can be regarded as an ``integrated version". In particular, this extends our results on the two-way transport. In some nice cases, the optimization problem reduces to classic optimization problems that admit explicit solutions.

In this section, we will often encounter transport 
from $(\rd\times[0,1],\mm\times \tau)$ to  $(\rd\times[0,1],\nn\times \tau)$. To simplify formulas, 
we will write $\kappa^x(\cdot)$ as $\kappa(x;\cdot)$ for a stochastic kernel $\kappa$, where $x$ often has two components. 

\subsection{Connecting MOT and SOT}
\label{sec:motsot}

Let $\tau$ be the Lebesgue measure on $[0,1]$ and write $\cR=\rd\times[0,1]$. 
The set of couplings between $(\cR ,\mm\times \tau)$ and  $(\cR ,\nn\times \tau)$ that are backward martingale in the first marginal is given by
$$\{((X,U),(X',U'))\in\widehat \Pi(\mm\times\tau,\nn\times\tau)\mid \E[X|(X',U')]= X'\},$$
where $\widehat \Pi$ is the set of random vectors having distributions in $ \Pi$. 
We disintegrate such couplings into stochastic kernels, and denote by $\S_{b,1}$ the corresponding collection of stochastic kernels. Formally,
$ \S_{b,1}$ is the set
\begin{align*}
     \Bigg\{\hat{\kappa}\in\K(\mm\times\tau,\nn\times\tau) \mid \int_{\cR }\hat{\kappa}^{(\bz,u)}(Z'\times V)\bz\tau(\d u)\mm(\d\bz)= \int_{Z'}\bz'\tau(V)\nn(\d\bz')~\forall Z'\times V\subseteq \cR \Bigg\}.
\end{align*}
For $\bz,\bz'\in\rd$, define also the sets of stochastic kernels
$$\K_{\bz}=\K(\mu_\bz,\delta_\bz\times\tau)\quad\text{ and }\quad\widetilde{\K}_{\bz'}=\K(\delta_{\bz'}\times\tau,\nu_{\bz'}).$$
Since $\tau$ is atomless, there exist kernels $\kappa_{\bz}\in \K_{\bz},~\widetilde{\kappa}_{\bz'}\in \widetilde{\K}_{\bz'}$ that are backward Monge and Monge, respectively. 

% couplings from $(A_\bz,\mu_{\bz})$ to $(\rd\times[0,1],\delta_\bz\times\tau)$ by
% $$\Pi_\bz=\{(X,(\bz,U))\in\Pi(\mu_\bz,\delta_\bz\times\tau)\}$$
% and similarly from $(\rd\times[0,1],\delta_{\bz'}\times\tau)$ to $(B_{\bz'},\nu_{\bz'})$ as
% $$\widetilde{\K}_{\bz'}=\{((\bz',U'),Y)\in\Pi(\delta_{\bz'}\times\tau,\nu_{\bz'})\}.$$We may identify $\Pi_\bz$ with kernels, and with  the set of transports from $\mu_\bz$ to $\tau$. \com{Why dont we just write $\Pi_{\mathbf z}$ as a set of kernels?}

\begin{theorem}[MOT-SOT parity]
\label{thm:rep}
Suppose that $\bmu(X)=\bnu(Y)$. Fix arbitrary  collections of kernels $\kappa_{\bz}\in \K_{\bz},~\widetilde{\kappa}_{\bz'}\in \widetilde{\K}_{\bz'}$ indexed by $\bz,\bz'\in\rd$,   where $\kappa_{\bz}$ is backward Monge,  $\widetilde{\kappa}_{\bz'}$ is Monge, and $\bz\mapsto \kappa_{\bz}$ and $\bz'\mapsto\widetilde{\kappa}_{\bz'}$ are measurable.    Every $\kappa\in\K(\bmu,\bnu)$ can be represented as
\begin{align}
    \label{eq:rep}
\kappa^x(B)=\int_{\cR }\int_{[0,1]}\kappa_{\bmu'(x)}^x(\bmu'(x),\d u)\hat{\kappa}^{(\bmu'(x),u)}&(\d \bz',\d u')\widetilde{\kappa}_{\bz'}^{(\bz',u')}(B),~~x\in X,\,B\subseteq Y,
\end{align}
for some $\hat{\kappa}\in \S_{b,1}$.
 Conversely, given any  $\hat{\kappa}\in \S_{b,1}$, the equation \eqref{eq:rep} defines a simultaneous transport kernel $\kappa\in\K(\bmu,\bnu)$ from $\bmu$ to $\bnu$.
\end{theorem}

% Also, I treat $\kappa_{\mathbf z}^x$ as a measure on $[0,1]$ instead of $\mathbf z\times [0,1]$. It seems that I can write
% \begin{align}
%  \kappa^x =\int_{\rd\times[0,1]}\int_{[0,1]}\widetilde{\kappa}_{\bz'}^{(\bz',u')} \kappa_{\bz}^x (  \d u)  \hat{\kappa}^{(\bz ,u)} ( \d \bz' , \d u')  \mbox{ for $x\in X$, where $\bz=\bmu'(x)$}.
% \end{align}
\begin{remark} Write $f:\cR  \to \mathcal P(Y)$, $   (\bz',u') \mapsto \widetilde{\kappa}_{\bz'}^{(\bz',u')} $. 
% and $g_x:  [0,1] \to \mathcal P(\rd\times[0,1] )$, $ u \mapsto  \hat{\kappa}^{(\bmu'(x) ,u)}$.
Then \eqref{eq:rep} can be written as, for $x\in X$, 
% \begin{align}
%  \kappa^x =\int_{\rd\times[0,1]}\int_{[0,1]} f \d \kappa_{\bz}^x &\d g \mbox{ for $x\in X$, where $\bz=\bmu'(x)$}.
% \end{align} 
\begin{align*}
 \kappa^x =
 \E[ f (\zeta, \xi)   ]  = \E\left [  \widetilde{\kappa}_{\zeta }^{(\zeta ,\xi)}\right]
\end{align*} 
where $(\zeta,\xi)  \lawis \int_0^1   \hat{\kappa}^{(\bmu'(x) ,u)}  \kappa_{\mathbf z}^x  (\d u) $.
\end{remark} 

Let us first explain the intuition.
The uniform measure $\tau$ on $[0,1]$ can be regarded as a parameterization space to keep the information of the space $(A_\bz,\mu_{\bz})$ when we map it to a single point on $\rd$.\footnote{When $\bmu'$ (resp.~$\bnu'$) is injective, we may remove the parameterization space on the $\mm$ (resp.~$\nn$) side.} It can be replaced by any atomless measure. The Monge property of $\kappa_\bz,\widetilde{\kappa}_{\bz'}$ will allow us to reconstruct the original simultaneous transport because it guarantees that no information is lost at the step where we encode $(A_\bz,\mu_{\bz})$ using a single point.\footnote{To see the Monge property is crucial, imagine we use independent couplings for both---it will not yield the set of all simultaneous transports.} Theorem \ref{thm:rep} says that the way we parameterize this information does not matter---it is possible to fix two collections of parameterizations a priori, as long as they have the Monge properties and are measurable.  See also Figure \ref{fig:commute} for a pictorial representation. 
\begin{figure}[t]
    \centering
    \begin{tikzcd}[row sep=2em, column sep = 6em]
X \arrow{r}{\kappa} \arrow[shift right, swap]{d}{x\mapsto(\bmu'(x),\kappa_{\bmu'(x)}^x)}& Y  \arrow[shift left,swap]{d}{(\bz',u')\mapsto S_{\bz'}(u')~\,}\\
(\cR ,\mm\times\tau) \arrow{r}{\hat{\kappa}} \arrow[shift right, swap]{u}{(\bz,u)\mapsto T_\bz(u)} & (\cR ,\nn\times\tau)\arrow[shift left,swap]{u}{~\, y\mapsto (\bnu'(y),\widetilde{\kappa}_{\bnu'(y)}^y)}
\end{tikzcd}

    \caption{MOT-SOT parity illustrated with commutative diagram: the double arrows connects the upper half (simultaneous transport) and lower half (martingale transport).  Note that the downward arrows are not  given by a map in general, which explains why  Monge transport may not always exist.}
    \label{fig:commute}
\end{figure}

We also need a few technical considerations. To see that Theorem \ref{thm:rep} actually makes sense, we need the following.
\begin{enumerate}[(i)]
    \item Existence of a measurable selection of $\{\kappa_\bz\}_{\bz\in\rd}$ and $\{\widetilde{\kappa}_{\bz'}\}_{\bz'\in\rd}$ satisfying the Monge properties. When $X ,Y $ are Euclidean spaces, this is guaranteed by a measurable selection of optimal plans (\cite[Corollary 5.22]{V09}) for   the quadratic cost, which are given by deterministic maps; see \cite{GM96}.\footnote{More precisely, we replace $\tau$ by $[0,1]^\ell$, where $\ell$ is the larger dimension of $X$ and $Y$. As commented above, this will not affect the result.} A sufficient condition for Polish spaces is given by \cite[Theorem 5.30]{V09}.
    \item Joint measurability of $\widetilde{\kappa}_{\bz'}^{(\bz',u')}(B)$ in $(\bz',u')$ (so that the integral in \eqref{eq:rep} makes sense).  To see this, denote the transport plan corresponding to $\widetilde{\kappa}_{\bz'}$ by $\widetilde{\pi}_{{\bz'}}$, which is a probability measure on $\{\bz'\}\times[0,1]\times B_{\bz'}$. Define $\widetilde{\pi}=\int\widetilde{\pi}_{\bz'}\nn(\d \bz')$, which is a probability measure on $\cR \times Y $ and is well-defined by (i). Next, disintegrate $\widetilde{\pi}$ in the first two coordinates, $(\bz',u')$ to get a family of measures $\{\widetilde{\kappa}_{\bz',u'}(\cdot)\}$ that is jointly measurable in $(\bz',u')$. By uniqueness of disintegration and since $B_{\bz'},\,\bz'\in\rd$ are disjoint, we must have $\widetilde{\kappa}_{\bz',u'}(\cdot)=\widetilde{\kappa}_{\bz'}^{(\bz',u')}(\cdot)$.
\end{enumerate}

\begin{example}
Let $X =Y =\{0,1\}$ and $\bmu(\{0\})=(1/3,2/3)$, $\bmu(\{1\})=(2/3,1/3)$, $\bnu(\{0\})=(1/3,1/3)$, and $\bnu(\{1\})=(2/3,2/3)$. We have $\mm=(\delta_{(4/3,2/3)}+\delta_{(2/3,4/3)})/2$ and $\nn=\delta_{(1,1)}$. The  backward  martingale transport from $\mm$ to $\nn$ is unique. Therefore by Theorem \ref{thm:rep}, the simultaneous transport from $\bmu$ to $\bnu$ is unique, and we can easily check that  it is given by $\kappa(0,\{0\})=\kappa(1,\{0\})=1/3$ and $\kappa(0,\{1\})=\kappa(1,\{1\})=2/3.$

% $\mu_1(\{0\})=\mu_2(\{1\})=\nu_1(\{0\})=\nu_2(\{0\})=1/3,~ \mu_1(\{1\})=\mu_2(\{0\})=\nu_1(\{1\})=\nu_2(\{1\})=2/3$.
% Then the only transport is given by $$\kappa(0,\{0\})=\kappa(1,\{0\})=\frac{1}{3}.$$
% \begin{center}
    
% \includegraphics[width=12cm, height=5cm]{1.jpg}
% \end{center}
\end{example}

% \begin{example}\label{ex2}This example shows that in Theorem \ref{thm:rep}, it is not sufficient if we only parameterize (arbitrarily) with $\kappa_\bz$ and $\widetilde{\kappa}_{\bz'}$ but fix  $\hat{\kappa}$. Consider $X=\{0,1,2\}$ and $Y=\{0,1\}$, and $\bmu(\{0\})=(1/4,3/4)$, $\bmu(\{1\})=\bmu(\{2\})=(3/8,1/8)$, and $\bnu(\{0\})=\bnu(\{1\})=(1/2,1/2)$.

% \begin{center}
    
% \includegraphics[width=12cm, height=6cm]{2.jpg}
% \end{center}
% \end{example}

% \begin{example}\label{ex:3}
% Note that here a feasible transport is the identity (which is bijective), while the corresponding (backward) submartingale transport is not Monge. This makes it difficult to infer the ST from the (backward) submartingale transport, which is straightforward for balanced ST (corresponding to backward martingale transport).
% \begin{center}
    
% \includegraphics[width=12cm, height=5cm]{3.jpg}
% \end{center}
% \end{example} 

A first immediate consequence is the following commutative relation. This can be seen as a special case of a more general commutative relation illustrated by Figure  \ref{fig:commute} below.

\begin{corollary}\label{coro:commute}
    Let $\bmu\in\P(X)^d$ and $\bnu \in \P(Y)^d$ satisfy that $\K(\bmu,\bnu)$ is non-empty.  Suppose that $\bmu$ is jointly atomless and $\mm$ is atomless. Then there exists a backward martingale coupling between $\mm$ and $\nn$ that is also Monge.\footnote{A coupling $(\xi_1,\xi_2)$ is \emph{backward martingale} if $\E[\xi_1|\xi_2]=\xi_2$, that is, $(\xi_2,\xi_1)$ forms a martingale.} Moreover, if we denote by $h$ the map that induces this Monge transport, then there exists a simultaneous transport map $f\in\T(\bmu,\bnu)$ satisfying
$$\bnu'(f(x))=h(\bmu'(x)),~x\in X.$$
\end{corollary}

Finally, we mention a more general version of Theorem \ref{thm:rep} for the unbalanced setting. Apart from the constraints $\kappa\in\K(\bmu,\bnu)$ we add an extra constraint that $\kappa\in\K(\mub,\nub)$ (which is automatically satisfied in the balanced case; see \eqref{step}), i.e., define 
$$\widetilde{\K}(\bmu,\bnu)=\{\kappa\in\K(\bmu,\bnu)\mid \kappa\in\K(\mub,\nub)\}.$$
There exists a parity relation between $\widetilde{\K}(\bmu,\bnu)$ and the set of stochastic kernels from $(\cR ,\mm\times \tau)$ to $(\cR ,\nn\times \tau)$ that is backward \emph{submartingale}  in the first marginal. The proof is very similar to the proof of Theorem \ref{thm:rep} and we omit the details.

\subsection{Optimality of simultaneous transport and examples}\label{sec:sotmot2}
In this section, for simplicity we will keep the reference measure $\eta=\mub$. The general case $\eta\ll\mub$ follows by modifying the cost function $c(x,y)$.

Suppose that we are given $\kappa_{\bz}\in \K_{\bz},~\widetilde{\kappa}_{\bz'}\in \widetilde{\K}_{\bz'},~\bz,\bz'\in\rd$ measurable as in the setup of Theorem \ref{thm:rep}. For $\kappa\in\K(\bmu,\bnu)$ with representation \eqref{eq:rep}, let us compute the cost of the associated simultaneous transport:
\begin{align}\C(\kappa) &=\int_{X \times Y }c(x,y)\kappa^x(\d y)\mub(\d x)\nonumber\\
    &=\int_{\rd}\int_{A_\bz}\int_Y  c(x,y)\kappa^x(\d y)\mu_{\bz}(\d x)\mm(\d \bz)\nonumber\\
    &=\int_{\rd}\int_{A_\bz}\int_Y  c(x,y)\int_{\cR }\int_{[0,1]}\kappa_\bz^x(\{\bz\}\times \d u)\hat{\kappa}^{(\bz,u)}(\d \bz',\d u') \widetilde{\kappa}_{\bz'}^{(\bz',u')}(\d y )\mu_{\bz}(\d x)\mm(\d \bz)\nonumber\\
    &=\int_{\cR }\int_{\cR }\hat{\kappa}^{(\bz,u)}(\d \bz',\d u') \left(\int_{A_\bz}\int_Y  c(x,y)\kappa_\bz^x(\{\bz\}\times \d u)\widetilde{\kappa}_{\bz'}^{(\bz',u')}(\d y )\mu_{\bz}(\d x)\right)\mm(\d \bz)\nonumber\\
    &=\int_{\cR }\int_{\cR }\hat{\kappa}^{(\bz,u)}(\d \bz',\d u')\left(\int_{A_\bz}\int_{B_{\bz'}} c(x,y)\kappa_\bz^x(\{\bz\}\times \d u)\widetilde{\kappa}_{\bz'}^{(\bz',u')}(\d y)\mu_{\bz}(\d x)\right)\mm(\d \bz),\label{eq:costafterrep}
\end{align}
where the third equality follows since   $ \bmu'(x)=\bz$ for  $x\in A_\bz$. 
The infimum cost $\inf\C(\kappa)$ may be computed by taking an infimum  over  all $\hat{\kappa}$ while fixing $\kappa_{\bz}\in \K_{\bz}$ and $\widetilde{\kappa}_{\bz'}\in \widetilde{\K}_{\bz'}$ for $\bz,\bz'\in\rd$. Alternatively, taking an infimum  over all $\kappa_\bz,~\widetilde{\kappa}_{\bz'}$,  and $\hat{\kappa}$ leads to the same value. 

Note that the measure
$$\gamma_{\bz,\bz',u'}(V):=\int_{A_\bz}\int_{B_{\bz'}} c(x,y)\kappa_\bz^x(\{\bz\}\times V)\widetilde{\kappa}_{\bz'}^{(\bz',u')}(\d y)\mu_{\bz}(\d x)$$
satisfies $\gamma_{\bz,\bz',u'}\ll \tau$ for all $\bz,\bz'\in\rd\text{ and }u'\in[0,1]$. By Theorem 58 in \cite{DM11}, there exists
%Therefore, there exists a family of  density functions $\{\d \gamma_{\bz,\bz',u'}/\d\tau\}_{\bz,\bz',u'}$
% Note that the measure
% $$\gamma_{y,\bz}(V)=\int_{A_\bz}c(x,y)\kappa_\bz^x(\{\bz\}\times V)\mu_{\bz}(\d x)$$
% satisfies $\gamma_{y,\bz}\ll\tau$ for all $y,\bz$. Therefore, there exists a density $\d\gamma_{y,\bz}/\d\tau$
 a jointly measurable  function $\hat{c}:\cR^2\to\R$ such that $$\int_{A_\bz}\int_{B_{\bz'}} c(x,y)\kappa_\bz^x(\{\bz\}\times \d u)\widetilde{\kappa}_{\bz'}^{(\bz',u')}(\d y)\mu_{\bz}(\d x)=\hat{c}(\bz,u,\bz',u')\tau(\d u).$$
% where
% $$\hat{c}(\bz,u,\bz',u')=\int_{B_{\bz'}}\widetilde{\kappa}_{\bz'}^{(\bz',u')}(\d y)\frac{\d\gamma_{y,\bz}}{\d\tau}(u)$$
% is Borel measurable. 
We are then left with
\begin{align*}
    \inf_{\kappa\in\K(\bmu,\bnu)}\C(\kappa)    & =\inf_{\hat{\kappa}\in\S_{b,1}}\int_{\cR }\int_{\cR }\hat{\kappa}^{(\bz,u)}(\d \bz',\d u')\hat{c}(\bz,u,\bz',u')\mm\times\tau(\d \bz,\d u)
\\& = \inf \E\left [\hat c (Z,U; Z',U')  \right ]
\end{align*}
where the last infimum is taken over all possible couplings  $(Z,U;Z',U')\lawis  (\mm\times \tau)\otimes \hat \kappa$ where $\E[Z|(Z',U')]=Z'$.
This becomes an optimal transport problem on ``backward martingale over the first marginal" in  $\R^{d+1}$. In fact, we may reduce the dimension to $(d-1)+1$, simply because the Radon--Nikodym derivatives sum up to a constant.   The connection to MOT also explains some bizarre behaviors of SOT. For example, in Example \ref{xy}, the transport cost is a constant if $\bmu$ is linear and the cost function is quadratic. This stems from the well-known fact that the quadratic cost is trivial for MOT. We next discuss a few special classes and explicitly solvable examples below.

\begin{example}
When $\bmu'$ is injective and $Y =\R$, we may pick $\kappa_{\bz}$ to map the point $(\bmu')^{-1}(\bz)$ to $\tau$ and $\widetilde{\kappa}_{\bz'}$ the comonotone coupling between $\nu_{\bz'}$ and $\tau$. We arrive at
$$\hat{c}(\bz,u,\bz',u')=\int_{B_{\bz'}} c((\bmu')^{-1}(\bz),y)\widetilde{\kappa}_{\bz'}^{u'}(\d y)=c((\bmu')^{-1}(\bz),F^\leftarrow_{\bz'}(u')),$$where $F^\leftarrow_{\bz'}$ is the left-continuous inverse of the measure $\nu_{\bz'}$.
\end{example}

\begin{example}\label{4}
Assume that $c(x,y)$ depends only on $\bmu'(x)$ and $\bnu'(y)$ (for instance, when both $\bmu'$ and $\bnu'$ are injective), then with $\hat{c}(\bmu'(x),\bnu'(y))=c(x,y)$, we have  for any $V\subseteq [0,1]$,
\begin{align*}
    \int_{A_\bz}\int_{B_{\bz'}} &c(x,y)\kappa_\bz^x(\{\bz\}\times V)\widetilde{\kappa}_{\bz'}^{(\bz',u')}(\d y)\mu_{\bz}(\d x)\\
    &=\hat{c}(\bz,\bz')\int_{A_\bz}\int_{B_{\bz'}} \kappa_\bz^x(\{\bz\}\times V)\widetilde{\kappa}_{\bz'}^{(\bz',u')}(\d y)\mu_{\bz}(\d x)=\hat{c}(\bz,\bz')\tau(V).
\end{align*}
This yields
\begin{align}\inf_{\kappa\in\K(\bmu,\bnu)}\C(\kappa)=\inf_{\hat{\kappa}\in\M(\nn,\mm)}\int_{\rd}\int_{\rd}\hat{\kappa}^{\bz'}(\d \bz)\hat{c}(\bz,\bz')\nn(\d\bz').\label{eq:inj~mot}\end{align}Intuitively, in the optimal transport problem, we may remove the extra dimension where $\tau$ is supported, because the cost function depends only on the $\rd$-coordinate.  In particular,  \eqref{eq:inj~mot} is now equivalent to an MOT problem. 

In this setting, it is possible to recover an optimal simultaneous transport from an optimal martingale coupling $\pi$: let $\hat{\kappa}^{(\bz,u)}$ follow $\pi$ in the first coordinate and identity in the second, then the kernel $\kappa$ defined in \eqref{eq:rep} gives an optimal simultaneous transport.
\end{example}

An immediate consequence of Example \ref{4} is that MOT on compact Euclidean spaces can be realized as a special case of SOT. This is called the \emph{MOT-SOT parity} as suggested by the title of this section.

\begin{example}
    We may strengthen Theorem \ref{thm:rep} when both $\bmu'(x)$ and $\bnu'(y)$ are injective, as follows. It can be checked that the same result goes through if we remove our parameterization space $([0,1],\tau)$. In other words, there is a correspondence between simultaneous transport and backward martingale transport (on $\R^d$). Using Theorem 2.1 of \cite{NWZ22}, we thus obtain a stronger version of Proposition \ref{prop:ssww}, that a Monge simultaneous transport exists when $\mm$ is atomless.\footnote{This is also true if we only assume $\bnu'$ is injective, by removing the parameterization space on the $\nn$ side.}   If moreover $c$, $(\bmu')^{-1}$, and $(\bnu')^{-1}$ are  continuous and $c$ is bounded, we conclude using Corollary 2.4 of \cite{NWZ22} the equivalence between Monge and Kantorovich costs, complementing Theorem \ref{infs}. 
\end{example}

\begin{remark}\label{sup?}
    Recall that the supremum in the MOT duality formula may not always be attained. Indeed, a simple  counterexample is given by \cite{BNT17}, Example 8.2. By the MOT-SOT parity, the supremum in our SOT duality formula \eqref{dua} is not always attained either. For details, see the dual MOT-SOT parity discussions in Section \ref{dualparity}. A direct construction of a counterexample can also be found in \cite{G20}.
\end{remark}

% The double integral inside the bracket is exactly the (classic) transport cost from $\mu_\bz$ to $\nu_\bz$. Denoting 
% $$\hat{c}(\bz,\bz'):=\inf_{\kappa_{\bz,\bz'}\in\K(\mu_{\bz},\nu_{\bz'})}\int_{A_\bz}\int_Yc(x,y)\kappa_{\bz,\bz'}(x;A)\mu_{\bz}(\d x),$$we get that
% \begin{align*}
%     &\hspace{0.5cm}\inf_{\kappa\in\K(\bmu,\bnu)}\C(\kappa)\\
%     &=\inf_{\kappa_{\bz,\bz'}\in\K(\mu_{\bz},\nu_{\bz'})}\inf_{\hat{\kappa}\in \M_b(\bmu'|_{\mub},\bnu'|_{\nub})}\int_{\rd}\int_{\rd}\left(\int_{A_\bz}\int_Yc(x,y)\kappa_{\bz,\bz'}(x;A)\mu_{\bz}(\d x)\right)\hat{\kappa}(\bz;\d\bz')\mm(\d \bz)\\
%     &=\inf_{\hat{\kappa}\in \M_b(\bmu'|_{\mub},\bnu'|_{\nub})}\int_{\rd}\int_{\rd}\hat{c}(\bz,\bz')\hat{\kappa}(\bz;\d\bz')\mm(\d \bz)\\
%     &=\inf_{\pi\in\M(\bnu'|_{\mub},\bmu'|_{\nub})}\int_{\rd}\int_{\rd}\hat{c}(\bz',\bz)\pi(\d\bz',\d\bz).
% \end{align*}
% This is yet another martingale optimal  transport problem. Meanwhile, this lies in $\R^d$ and there is not too much we can say about $\hat{c}$ even if $c$ is nice. However, in a few special cases we can still say something.
\begin{example}
 Consider $\bmu'$ taking values on only two points, say $\bz_1$ and $\bz_2$ (so that $\bnu'$ takes values only on the line segment joining these two points). In this case, there is a further decomposition of the transport $\hat{\kappa}\in\M_{b,1}$, into a collection of independent transports from $([0,1],\tau)$ to $([0,1],\tau)$. This is because any such $\hat{\kappa}$ must transport a positive fraction of the measure on $\bz'\times[0,1]$ to $\bz_1\times[0,1]$ and the rest to $\bz_2\times[0,1]$. Further solutions are available if $\hat{c}$ is nicely behaved (however, this might be difficult to achieve in general).
In the special case where $\bmu'$ is also injective,  every simultaneous transport has the same cost.
\end{example}

\begin{example}
Let $d=2$ and consider measures $\bmu,\bnu$ on $\R$ such that $\d\mu_1/\d\mub(x)$ and $\d\nu_1/\d\nub(y)$ are affine in $x,y$ respectively with positive slopes  (the cases with negative slopes are analogous). 
%Let $c(x,y)$ be a Borel \emph{backward martingale Spence-Mirrlees map}, i.e., $c(\cdot,y_2)-c(\cdot,y_1)$ is strictly increasing and strictly convex for all $y_1<y_2$ (in particular, this holds if $c$ is smooth and $c_{xxy},c_{xy}>0$),
Assume that $c(x,y)=h(x-y)$ for some differentiable $h$ with $h'$ strictly convex, and such that $|c(x,y)|\leq a(x)+b(y)$ for some $a\in L^1(\mm),~b\in L^1(\nn)$.
This SOT problem is then reduced to a MOT problem on $\R$, in the form of \eqref{eq:inj~mot}, with a martingale Spence-Mirrlees cost function.
Using Theorem 1.7 in \cite{BJ16}, the MOT problem \eqref{eq:inj~mot} is uniquely  solved by the \emph{left-curtain transport} from $\nn$ to $\mm$
%which can be defined  as the mirror of the \emph{increasing supermartingale transport} as defined in \cite{NuS18}.\footnote{For example, the decreasing submartingale transport now corresponds to right-curtain and antitone kernels, which reflects the conditions $c_{xxy},c_{xy}>0$.} 
This coupling uniquely induces a simultaneous transport from $\bmu$ to $\bnu$ since $\bmu',\bnu'$ are injective.
This easily generalizes to when $f_1:=\d\mu_1/\d\mub$ and $g_1:=\d\nu_1/\d\nub$ not being linear. For example, assuming $f_1'',f_1',g_1',c_{xxy}$ all being positive suffices.
\end{example}

\subsection{Decomposition of   two-way transport}\label{61}

Define an equivalence relation $\simeq$ among $\R^d_+$-valued probability measures as follows: $\bmu\simeq\bnu$ if $\mm=\nn$  (or equivalently, both $\Pi(\bmu,\bnu)$ and $\Pi(\bnu,\bmu)$
are non-empty). 
For $P\in\P(\R^d)$, we define $$\cE_P=\{\bmu\in\Pi(X)^d \mid \mm=P\},$$ the equivalence class under $\simeq$.
The transitivity of $\simeq$ follows from Corollary \ref{comp}.  
 We also define the minimum transport cost between $\bmu$ and $\bnu$ as
$$
\I_c(\bmu,\bnu):=
\inf_{\pi\in\Pi(\bmu,\bnu)}\int_{X\times Y}c(x,y)\pi(\d x,\d y)=\inf_{\kappa\in\K(\bmu,\bnu)}\C(\kappa).$$

\begin{theorem}[Decomposition of  two-way  transport]\label{2way}
 Suppose that $c$ is continuous. For $\bmu,\bnu\in \cE_P$ and $\kappa\in\K(\bmu,\bnu)$, we have $\kappa\in\K(\mu_{\bz},\nu_{\bz})$ for $P$-a.e.~$\bz$. Moreover, the following are equivalent:
\begin{enumerate}[(i)]
    \item $\kappa$ is an optimal transport from $\bmu$ to $\bnu$;
    \item $\kappa$ is an optimal transport from $\mu_{\bz}$ to $\nu_{\bz}$ for $P$-a.s.~$\bz$;
    \item we have 
    \begin{align}\C(\kappa)=\int_{\R_+^d}\I_c(\mu_{\bz},\nu_{\bz}) P (\d {\bf z}).\label{2ic}\end{align}
\end{enumerate}
In particular,
$$\I_c(\bmu,\bnu) =\int_{\R_+^d}\I_c(\mu_{\bz},\nu_{\bz}) P (\d {\bf z}).$$ 
\end{theorem}

\begin{remark}
That the right-hand side of \eqref{2ic} is indeed well-defined will be discussed in the proof using a measure selection argument.
\end{remark}

\begin{remark}
    If $\bmu,\bnu\in\cE_P$ and $\bmu'$ is injective, the proof of Theorem \ref{2way} also indicates that $\K(\bmu,\bnu)$ consists of a single element, i.e., the simultaneous transport from $\bmu$ to $\bnu$ is unique.  
\end{remark}

A few comments are in place. Roughly speaking, two-way transports exist if and only if  each transport from $\bmu$ to $\bnu$ (provided it exists) can be inverted to produce a transport from $\bnu$ to $\bmu$. This inversion is not in general possible, because multiple points with different Radon--Nikodym derivatives may be transported to the same point in the destination, while any inversion transports back with the same Radon--Nikodym derivative as the destination point; see Figure \ref{fig:2} (a).

If both $X,Y$ are discrete, the two-way transports exist if and only if for each $\bz\in\R_+^d$,
\begin{align}\sum_{x\in X:\ \bmu'(x)=\bz}\mub(\{x\})=\sum_{y\in Y:\ \bnu'(y)=\bz}\nub(\{y\}).\label{2xy}\end{align}

 Theorem \ref{2way} provides us with an explicit expression of the minimum cost $\I_c(\bmu,\bnu)$. Intuitively, it amounts to optimizing a (possibly infinite) collection of individual classic transport problems with the same cost function. For the important case of a convex cost, i.e., $c(x,y) = h( y-x )$ with $h$ strictly convex,  existing techniques can be applied to solve these individual problems; see \cite{GM96}. 
In the special case where $X=Y=\R$ and 
$c$ is continuous and strictly submodular\footnote{A function $c$ on $X\times Y$ is submodular  if 
$c(x,y) + c(x',y') \le c(x,y') + c(x',y)$ whenever $x\le x'$ and $y\le y'$. 
It is strictly submodular if the above inequality is strict as soon as $(x,y)\ne (x',y')$. An example of a (strictly) submodular function on $\R^2$ is $(x,y)\mapsto h(y-x)$ for a  (strictly) convex $h$.
} on $\R^2$, this transport problem is uniquely optimized by taking comonotone transport plans from $\mu_{\bf z}$ to $\nu_{\bf z}$ for each ${\bf z}\in\R_+^d$ by the Fr\'{e}chet-Hoeffding theorem. We summarize this in the following corollary.

\begin{corollary}\label{convex}
Suppose that  $X=Y=\R$, $\bmu,\bnu\in\mathcal E_P$ and $c$ is continuous and submodular. Then
$$\I_c(\bmu,\bnu)=\int_{\R_+^d}\int_0^1 c\left(F^{-1}_{\bz}(t),G^{-1}_{\bz}(t)\right)\d t\, P(\d {\bf z}),$$
where $F^{-1}_{\bz},G^{-1}_{\bz}$ are the distribution functions of $\mu_{\bz},\nu_{\bz}$ respectively.
\end{corollary}

In Appendix \ref{sec:W}, we further  discuss an application of Corollary \ref{convex} to Wasserstein distances between $\R^d_+$-valued probability measures. Each of the distances will be defined on some equivalence class $\mathcal E_P$.

As another consequence of  Theorem \ref{2way}, we obtain the following slightly stronger  duality result. This duality formula appears in a different form compared to the one in Theorem \ref{duality1}, and it is  similar to the duality formula in the classic setting ($d=1$). This is due to  Theorem \ref{2way}, which is only possible in case of two-way transport problems. 
Recall that in case $d=1$, all transport problems are two-way.

\begin{proposition}
Suppose that both $\Pi(\bmu,\bnu)$ and $\Pi(\bnu,\bmu)$ are non-empty and $c:X\times Y\to[0,\infty)$ is uniformly continuous and bounded, then duality holds as
\begin{align} \inf_{\pi\in\Pi(\bmu,\bnu)}\int_{X\times Y}c \, \d \pi =\sup_{(\phi,\psi)\in \widetilde{\Phi}_c}\int_X\phi\,  \d\mub+\int_Y \psi\, \d \nub, \label{dua2}\end{align}
where  \begin{align*} \widetilde{\Phi}_c  =\Bigg\{(\phi,\psi)\in L^1(\mub)\times L^1(\nub)\mid \phi (x)+\psi(y)\le c(x,y)\text{ if }\bmu'(x)=\bnu'(y)\Bigg\}.\end{align*}  
Moreover, both the infimum and supremum in \eqref{dua2} are attained.\label{newdua}
\end{proposition}

Using the Decomposition Theorem, we discuss a few interesting examples illustrating the peculiarities of simultaneous transport (complementing Section \ref{sec:12}) on an  equivalence class $\cE_P$. 
From classic optimal transport theory $(d=1)$, we first recall the following result.

%{Does this have a name? Brenier's Theorem?}-- Brenier is for quadratic cost only.
\begin{proposition}[Theorem 1.17 in \cite{S15}]
Suppose that probability measures $\mu,\nu$ are supported on a compact domain $\Omega\subseteq \R^N$ where $\partial\Omega$ is $\mu$-negligible, $\mu$ is absolutely continuous, and $c(x,y)=h(y-x)$ with $h$ strictly convex, then there exists a unique transport that is optimal among all Kantorovich transports and such a transport is Monge.\label{oldp}
\end{proposition}
 Example \ref{ex3} below shows that, in the setting of simultaneous transport $(d=2)$,
there may not exist an optimal Monge transport 
even if we assume moreover that both $\bmu$ and $\bnu$ are absolutely continuous with respect to the Lebesgue measure on $[0,1]^2$ and are jointly atomless. 

We first recall 
from Exercise 2.14 in \cite{V03} that if we remove the absolute continuity condition of $\mu$ while still assuming $\mu$ is atomless, Proposition \ref{oldp} may fail to hold. A counterexample is given by $\mu$ being  uniform on $[0,1]\times \{0\}$, and  $\mu$  uniformly distributed on $[0,1]\times\{a,b\}$ where $a\neq b$,  with $N=2$ and $c(\bx,\by)=\Vert \bx-\by\Vert^2$.

\begin{example}
%Say $(\widetilde{X},\widetilde{\mu})$ and $(\widetilde{Y},\widetilde{\nu})$ is such an example and such that $\widetilde{X},\widetilde{Y}$ are compact and there is a unique optimal Kantorovich  transport. Denote  by $\gamma_1,\gamma_2$ two measures supported on $[0,1]$ having densities  $2x$ and $2-2x$ with respect to Lebesgue. Define $X=[0,1]\times \widetilde{X}$, $\mu_j=\gamma_j\times\tmu$, and $Y=[0,1]\times\widetilde{Y}$, $\nu_j=\gamma_j\times \tnu$ for $j=1,2$. 
Consider $N=2$ and $c(\bx,\by)=\n{\bx-\by}^2$. Define $\mu_1,\nu_1$ being uniformly distributed on $[0,1]\times[0,1]$ and $[0,1]\times [2,3]$ respectively. Define $\mu_2$ supported on $[0,1]\times [0,1]$ such that $\d\mu_2/\d\mu_1(x,y)=2y$ and $\nu_2$ supported on $[0,1]\times [2,3]$ such that $\d\nu_2/\d\nu_1(x,y)=2-4|y-5/2|$. 

Observe that $\mub,\nub$ are compactly supported and $\bmu,\bnu$ are jointly atomless (e.g., the uniform distribution on $[0,1]\times \{0\}$ and $\mu_1'$ are independent). For each $z\in\R_+$, using notations similarly as in Section \ref{W}, we have $A_z:=(\mu_1')^{-1}(z)=[0,1]\times\{(1-z)/2z\}$ and $B_z:=(\nu_1')^{-1}(z)=[0,1]\times\{(5/2)\pm((3z-1)/2z)\}$. Moreover, $\mu_z,\nu_z$ are uniformly distributed on $A_z,B_z$ respectively. Thus, from the counterexample mentioned above, the unique optimal transport from $\mu_z$ to $\nu_z$ is not Monge unless $z=1/3$. This proves that the unique optimal transport from $\bmu$ to $\bnu$ is not Monge.
\label{ex3}
\end{example}

We next discuss an example of simultaneous transport between Gaussian measures. For simplicity we focus on the case $d=2$ with $L^2$ cost. First, we record a general result stating that for $\bmu,\bnu$ in the same equivalence class $\cE_P$, there must exist a linear transport between them. That is, $\bmu$ and $\bnu$ differ by a nonsingular linear transformation. %This can be seen as a complement to Example \ref{ex0} above.

\begin{proposition}\label{d2}
If $\bmu,\bnu$ are $\R^2$-valued  Gaussian measures on $\R^N$ with positive densities everywhere, then $\bmu,\bnu$ belong to the same equivalence class $\cE_P$ if and only if $\bmu\circ T^{-1}=\bnu$ where $T(\mathbf{x})=A\mathbf{x}+\mathbf{b}$ with $A$ invertible. 
\end{proposition}

% \begin{enumerate}[i.]
% \item Weak convergence\com{What is the other name for convergence in $P_P$ you mentioned?? Change the words on the next page too.} in $\cE_P$ does not imply Wasserstein convergence, so that they induce different topologies. {Convergence in $P_P$ instead of weak convergence?}
% \item Even with the $L^2$ cost, the optimal transport may not be the linear transport given in Proposition \ref{d2}.
% \end{enumerate}

%{Write the example on board: first, W convergence doesn't imply weak converge ($\delta\to 0$); second, linear exists but may not be optimal (the ``flip" transport is not optimal. Consider general convex costs?)}

\begin{example}
     We discuss an example where the optimal transport may not be the linear transport given in Proposition \ref{d2}. 
Consider $\delta>0$ and Gaussian measures $\mu_1,\nu_1\sim N(0,I_2)$, $\mu_2\sim N(0,\Sigma)$, and $\nu_2\sim N(0,\Omega)$ where 
$$\Sigma=\begin{pmatrix}1+\delta & 0\\
0&1
\end{pmatrix}\text{ and }\Omega=\begin{pmatrix}1 & 0\\
0&1+\delta
\end{pmatrix}.$$
It is straightforward to compute all the linear transports in $\K(\bmu,\bnu)$. These are given by reflections along $y=\pm x$ axes and rotations of $\pm \pi/2$ degrees at zero. Our goal is to show that  these transports are not optimal, in contrast to the case $d=1$ where optimal transports are linear. Observe that $\bmu,\bnu$ belong to the same equivalence class $\cE_P$ (since two-way transports exist), so that we may apply  Theorem \ref{2way}. Consider $z\in(0,2)$, then computing the density yields that $\d\mu_1/\d\mub((x,y))=z$ if and only if 
$$x=\pm\sqrt{\frac{2(1+\delta)^{3/2}}{\delta}\log\left(\frac{2}{z}-1\right)}=:\pm h_\delta(z).$$
Similarly, $\d\nu_1/\d\nub((x,y))=z$ if and only if $y=\pm h_\delta(z)$. The optimal transport problem from $\bmu$ to $\bnu$ is then reduced to transporting from  
$\{(x,y)\mid x=\pm c\}$ to   $\{(x,y)\mid y = \pm c\}$ for each $c =h_\delta(z)\ge 0$ on which some copies of Gaussian measures are equipped. Direct computation shows that a transport 
from $ \mu_z$ to   $\nu_z$ 
is given by
$$T((x,y))=(\text{sgn}(x)|y|,\text{sgn}(y)  |x|).$$
This is illustrated by the following Figure \ref{fig:7}. 

\begin{figure}[h!]
\begin{center}
    
 \begin{tikzpicture}
    \begin{axis}[axis lines=middle,
            enlargelimits,
            ytick=\empty,
            xtick=\empty,]
\addplot[name path=F,blue,domain={-2:2}] {1} node[pos=.575, above]{$h_\delta(z)$};
\addplot[name path=F2,blue,domain={-2:2}]{-1}node[pos=.6, below]{$-h_\delta(z)$};
\addplot[domain={-2:2}
,purple, name path=three] coordinates {(-1,-2)(-1,2)};
\addplot[domain={-2:2},purple, name path=threee] coordinates {(1,-2)(1,2)};
\draw[   ->](axis cs:1,1.5)--(axis cs:1.5,1);  
\draw[   ->](axis cs:1,0.5)--(axis cs:0.5,1);  
\draw[   ->](axis cs:1,-1.5)--(axis cs:1.5,-1);  
\draw[   ->](axis cs:1,-0.5)--(axis cs:0.5,-1);  
\draw[   ->](axis cs:-1,1.5)--(axis cs:-1.5,1);  
\draw[   ->](axis cs:-1,0.5)--(axis cs:-0.5,1);  
\draw[   ->](axis cs:-1,-1.5)--(axis cs:-1.5,-1);  
\draw[   ->](axis cs:-1,-0.5)--(axis cs:-0.5,-1);  

\draw[   ->](axis cs:1, 1.3)--(axis cs: 1.3,1);  
\draw[   ->](axis cs:1, 0.7)--(axis cs: 0.7,1);  
\draw[   ->](axis cs:1,- 1.3)--(axis cs: 1.3,-1);  
\draw[   ->](axis cs:1,- 0.7)--(axis cs: 0.7,-1);  
\draw[   ->](axis cs:-1, 1.3)--(axis cs:- 1.3,1);  
\draw[   ->](axis cs:-1, 0.7)--(axis cs:- 0.7,1);  
\draw[   ->](axis cs:-1,- 1.3)--(axis cs:- 1.3,-1);  
\draw[   ->](axis cs:-1,- 0.7)--(axis cs:- 0.7,-1); 

\draw[   ->](axis cs:1, 1.7)--(axis cs: 1.7,1);  
\draw[   ->](axis cs:1, 0.3)--(axis cs: 0.3,1);  
\draw[   ->](axis cs:1,- 1.7)--(axis cs: 1.7,-1);  
\draw[   ->](axis cs:1,- 0.3)--(axis cs: 0.3,-1);  
\draw[   ->](axis cs:-1, 1.7)--(axis cs:- 1.7,1);  
\draw[   ->](axis cs:-1, 0.3)--(axis cs:- 0.3,1);  
\draw[   ->](axis cs:-1,- 1.7)--(axis cs:-1.7,-1);  
\draw[   ->](axis cs:-1,- 0.3)--(axis cs:-0.3,-1); 
     \draw [dashed] (axis cs:0,1) -- (axis cs:1,0);
     \draw [dashed] (axis cs:0,1) -- (axis cs:-1,0);
     \draw [dashed] (axis cs:0,-1) -- (axis cs:1,0);
     \draw [dashed] (axis cs:0,-1) -- (axis cs:-1,0);
% \addplot[name path=G,purple,domain={0:1}] {0.64-0.28*x}node[pos=.1, above]{$\mu_1$};
% \addplot[name path=H,domain={0:1}] {0};

% \addplot[pattern=south east lines, pattern color=black!50]fill between[of=G and H, soft clip={domain=0:0.5}]
% ;

\end{axis}
    
    \end{tikzpicture}
\end{center}\caption{Optimal transport from $\mu_z$ to $\nu_z$: the red and blue lines indicate the supports of $\mu_z$ and $\nu_z$ respectively. Black arrows indicate the transports.}
\label{fig:7}
\end{figure}

Recall from Theorem 3.2.9 of \cite{R98} that $T$ is optimal if and only if 
$$(\text{sgn}(x)|y|,\text{sgn}(y) |x|)\in\partial f(x,y)$$
for some lower semi-continuous convex function $f$ on $\R^2$, where the subdifferential $\partial f$ is given by 
$$\partial f(\bx):=\{\bx^*\in X^*\mid f(\bx)-f(\by)\geq \langle \bx-\by,\bx^*\rangle \text{ for all }\by\in X\}, \ \bx\in X=\R^2.$$
It is straightforward to check that $f(x,y)=|y|$ meets these criteria, so that $T$ is indeed an optimal transport from $\mu_z$ to $\nu_z$. Since this holds for all $z\in(0,2)$, by Theorem \ref{2way}, the   transport 
$ T $ on $\R^2$ 
is an optimal transport from $\bmu$ to $\bnu$. Evidently, this is not given by a linear map. Intuitively, even if there always exists an optimal linear transport map when considering transports of a single measure, in our case the measures are weaved in such a way that under both constraints, none of the   linear maps become optimal.

\label{ex:gaussian2}

\end{example}

\section{Concluding remarks}
\label{sec:7}

The simultaneous optimal transport is introduced and studied in this paper.   
In view of the wide applications of optimal transport in  economic studies, such as  contract design (\cite{E10}), Cournot-Nash equilibria in non-atomic games (\cite{BC16}), multiple-good monopoly (\cite{DDT17}), implementation problems (\cite{NS18}), and team matching (\cite{BTZ21}, there are many directions  of SOT for future exploration,
in addition to our motivating examples and   equilibrium analysis in Section \ref{sec:model} and Appendix  \ref{sec:equilibrium}.
More broadly, optimal transport also has strong presence in 
robust risk assessment (\cite{EPR13}), option pricing (\cite{BHP13}),  machine learning (e.g., \cite{PC19}), operations research (e.g., \cite{BM19}) and statistics (e.g., \cite{CCG16}), which offer natural locations to look for   applications of our new framework.

The framework is shown to be technically  much more complicated than 
the classic setting which corresponds to $d=1$ and many new mathematical results are obtained. 
Due to the additional technical richness,
there are many directions to explore within the framework of SOT. We discuss a few  directions below.
\begin{enumerate}[(i)]
\item The MOT-SOT parity (Theorem \ref{thm:rep}) could potentially pave the path to many future developments of SOT. For example, some results on MOT such as complete duality may be translatable to SOT, shedding light on some of our open questions  below.  Computational methods for  SOT may be developed based on those of MOT; see \cite{D182} and \cite{GO19}. Exploration along these directions is left for future study. 
\item  We have focused on the case where $d$ is an integer. 
The problem can be naturally formulated for infinite dimension, by looking at 
$ 
\K(\bmu,\bnu):=\bigcap_{j \in J} \K(\mu_j,\nu_j) 
$ 
where $J$ is an infinite set which is possibly a continuum.
The optimal transport problem in this setting can be seen as a limit in some sense
of our setting as $d\to\infty$. 
A significant technical challenge arises because $\{\mu_j\mid j\in J\}$ may not admit a dominating measure.
For studies involving collections of probabilities without a dominating measure, see  e.g., \cite{STZ11} in the context of stochastic analysis with applications to mathematical finance. Assuming existence of a dominating measure and under some additional assumptions, \cite{G20} proved duality results in infinite dimension via an abstract duality theorem.

\item  The setting of this paper involves two tuples of measures to transport between. 
A natural question is how to generalize the framework to accommodate multiple marginals $\bmu^1,\dots,\bmu^n \in \P(X)^d$. For simplicity, assume all marginals are probabilities and   defined on the same space $X$.
In case $d=1$, such a generalization can be conveniently described via the Kantorovich formulation such that the optimal transport problem is 
$$
\inf_{\pi\in \Pi(\mu^1,\dots,\mu^n)} \int_{X^n} c \, \d \pi,
$$
where $c: X^n\to \R$ is the cost function and $\Pi(\mu^1,\dots,\mu^n)$ is the collection of measures with marginals $\mu^1,\dots,\mu^n\in \P(X)$; 
see e.g., \cite{R13} and \cite{P15} for results in multi-marginal transports for $d=1$. 
In contrast to the case  $d=1$ or $n=2$, such a generalization cannot be easily described via the Kantorovich formulation for   $d\ge 2$ and $n\ge 3$. A possible formulation  via kernels is given by defining, for each $j\in [d]$,
$ 
\mathcal K(\mu^1_j,\dots,\mu^n_j) = \{\kappa: X\to \P(X^{n-1})\mid \kappa_\# \mu^1_j \in \Pi(\mu_j^2,\dots,\mu_j^n)  \} 
$ 
and letting 
 $\mathcal K(\bmu^1,\dots,\bmu^n ) =\bigcap_{j=1}^d   \mathcal K(\mu^1_j,\dots,\mu^n_j).$  
Each $\kappa \in \mathcal K(\bmu^1,\dots,\bmu^n ) $ corresponds to a  multi-marginal simultaneous transport, with $n=2$ corresponding  our setting and $d=1$ corresponding to the classic multi-marginal transport setting. 
\item  Recall that in the Monge formulation,
the objective is to minimize
\begin{equation}
 \mathcal C_\eta (T) = \int_Xc(x,T(x))\eta (\d x).
\label{eq:non-linear}
\end{equation}
One may consider a nonlinear reference, i.e., $\eta$ in  \eqref{eq:non-linear} is a Choquet capacity\footnote{A Choquet capacity $\eta$ on a $\sigma$-field $\mathcal B$ of $X$ is a   function $\eta:\mathcal B \to [0,\infty]$ such that $\eta(\emptyset)=0$ and $\eta(A)\le\eta (B)$ for $A\subseteq B\subseteq X$, and the integration of $L:X\to \R$ with respect to $\eta$ is defined as $\int L \d \eta = \int_0^\infty \eta (L >t)\d t + \int_{-\infty}^0( \eta (L >t)-\eta(X))\d t  $.}  instead of a measure. 
The motivation of this formulation can be easily explained in the context of Example \ref{opt}, where the objective is
\begin{equation}
\mbox{to minimize~}\int f(L) \d \eta, \mbox{~~~~subject to $L\in \T(\bmu,\bnu)$}.
\label{eq:non-linear2}
\end{equation} 
By taking $\eta$ as a capacity, \eqref{eq:non-linear2}  includes many popular objectives
in risk management and decision analysis.
For instance, if $\eta$ is given by $\eta: A\mapsto \id_{\{\mathbb{P}(A)>1-\alpha\}}$ where $\mathbb P \in \P(X)$, 
then $\int f(L)\d \eta$ is the (left) $\alpha$-quantile of $f(L)$,
and the problem \eqref{eq:non-linear2} is a quantile optimization problem; see e.g., \cite{R10} for an axiomatization of quantile optimization in decision theory.
This formulation also includes optimization of risk measures (\cite{FS16}) or rank-dependent utilities (\cite{Q93}) of the financial position $f(L)$.
More generally, one may optimize $\mathcal R(L)$ subject to $L\in \mathcal T(\bmu,\bnu)$
for a general mapping $\mathcal R:\mathcal L\to \R$, such as many other quantities developed in decision theory (e.g., \cite{HS01, MMR06}).
Alternatively, instead of choosing $\eta$ as a capacity,  $\bmu$ and $\bnu$  may also be  chosen as tuples of capacities instead of measures.

\item Our optimal transport is allowed to be 
chosen from the entire set of transports $\K(\bmu,\bnu)$ (kernel) or $\T(\bmu,\bnu)$ (Monge).
There is an active stream of research on optimal transport with constraints such as MOT and directional optimal transport (e.g., \cite{NW22}).
Adding these constraints to the simultaneous transport gives rise to many new challenges and requires further studies. 

\item There are a few  technical open questions related to results in this paper.
\begin{enumerate}[(a)]
 \item We have explained in Remark  \ref{sup?} that the supremum may not always be attained in the duality formula \eqref{dua}. Establishing a complete duality remains a challenging problem, where we refer to \cite{BNT17}, \cite{NuS18}, and \cite{DT19} for relevant results for MOT. Sufficient conditions for attainability were established by \cite{G20}.\label{open}
\item There are several places in the paper where compactness of $X$ and $Y$ is assumed. For instance, compactness is  used in Theorem \ref{infs} and Proposition \ref{connect}. We expect that this condition can be removed. In particular, we note that Theorem \ref{infs} for $d=1$ holds without the compactness assumption as shown by \citet{P07}.

%  \item Duality, equilibrium, uniqueness, and decomposition results in Sections \ref{Dua} and \ref{W} are obtained only in the balanced setting. Whether analogous results hold in the unbalanced case requires different techniques and further study. 

   \end{enumerate} 

\end{enumerate}

 %   \item Generalization of (\ref{2inf}) to $d>2$ and to multiple-steps.
%   \item If $\kappa\in\Pi(\bmu,\bnu)$ and there ``almost exists" $\kappa'\in\Pi(\bnu,\bmu)$ (in some sense), then $\kappa$ cannot change the RN derivative too much. Can one use this to give upper and lower bounds for the optimal cost? This should serve as a generalization of Theorem \ref{2way}.
%   \item Counterexamples of Proposition \ref{d2} in the case $d>2$.
%   \item A counterexample where weak convergence does not imply Wasserstein convergence. 

\subsubsection*{Acknowledgements}
The authors thank Itai Ashlagi, Jose Blanchet, Job Boerma, Ibrahim Ekren, Henry Lam, Fabio Maccheroni, Max Nendel, Marcel Nutz, Giovanni Puccetti, Ludger R\"uschendorf, Aleh Tsyvinski,  and 	Kelvin Shuangjian Zhang  for helpful comments on a previous version of the paper. We are also grateful to three anonymous referees for their helpful feedback. RW acknowledges financial support from the Natural Sciences and Engineering Research Council of Canada (NSERC, RGPIN-2018-03823, RGPAS-2018-522590).

\appendix

\setcounter{lemma}{0}
\renewcommand{\thelemma}{A.\arabic{lemma}}
\setcounter{proposition}{0}
\renewcommand{\theproposition}{A.\arabic{proposition}}
\setcounter{theorem}{0}
\renewcommand{\thetheorem}{A.\arabic{theorem}}
\setcounter{definition}{0}
\renewcommand{\thedefinition}{A.\arabic{definition}}
\setcounter{corollary}{0}
\renewcommand{\thecorollary}{A.\arabic{corollary}}
\setcounter{example}{0}
\renewcommand{\theexample}{A.\arabic{example}} 
\setcounter{equation}{0}
\renewcommand{\theequation}{A.\arabic{equation}} 

\begin{center}
\bf \Large Appendices 
%``Simultaneous optimal transport"
\end{center}

In the appendices, we first present proofs and some  additional results   in Appendix \ref{A:pf}. We then collect a small review of literature on various generalizations of optimal transport in Appendix \ref{rev}. Finally, we discuss an application of SOT duality to a labour market equilibrium model in Appendix \ref{sec:equilibrium}.

\section{Proofs and additional results}\label{A:pf}

\subsection{Proofs  of results  in Section \ref{S2}}\label{A2}

\begin{proof}[Proof of Proposition \ref{prop:ssww}]
The first statement is implied by Proposition 9.7.1 of \cite{T91} and the remarks that follow it. The second statement can be shown  by  the same arguments as in Theorem 3.17 of \cite{SSWW19} where $\bmu(X)=\bnu(Y)$ is assumed. 
\end{proof}

 \begin{proof}[Proof of Proposition \ref{prop:gaussian1}]
 Note that it suffices to prove the case $T=3$, as the necessity statement for $T>3$ follows from that for $T=3$.  
% , where we prove that $\K((\mu_1,\mu_2),(\mu_2,\mu_3))$ is non-empty if and only if it is one of the following cases:
% \begin{enumerate}[(i)]
%\item $\sigma_1\leq\sigma_2\leq \sigma_3$ and $\sigma_1\sigma_3\leq\sigma_2^2$;
%\item $\sigma_1\geq\sigma_2\geq \sigma_3$ and $\sigma_1\sigma_3\geq\sigma_2^2$.
% \end{enumerate} 
Write $\alpha= \sigma_2^2/\sigma_1^2>0  $ and $\beta=\sigma_3^2/\sigma_2^2>0  $.  
  Increasing log-concavity of $t\mapsto \sigma_t$ 
means $\alpha\ge \beta\ge 1$ (case i) and   decreasing log-convexity of $t\mapsto \sigma_t$ means $\alpha \le \beta\le 1$ (case ii). 

Using Lemma 3.5  of \cite{SSWW19}, the following are equivalent:
\begin{enumerate}[(a)]
\item $\K((\mu_1,\mu_2),(\mu_2,\mu_3))\neq\emptyset  $;
 \item $\frac{\d \mu_2}{\d \mu_3} \big |_{\mu_3} \lcx \frac{\d \mu_1}{\d \mu_2} \big|_{\mu_2}$;
 \item  $\frac{\d \mu_3}{\d \mu_2} \big |_{\mu_2} \lcx \frac{\d \mu_2}{\d \mu_1} \big|_{\mu_1},$ 
 \end{enumerate} 
 where $\lcx$ is the one-dimensional convex order on $\mathcal P$.
  We shall use the equivalent condition (b) for the case $\alpha,\beta\ge 1$ and the condition (c) for the case $\alpha,\beta \le 1$.
Writing $\xi$ as a standard Gaussian random variable, 
 and $\laweq$ as  equality in distribution, by direct calculation,  
 \begin{align*}
 \frac{\d \mu_1}{\d \mu_2} \Big |_{\mu_2} 
\laweq \frac{\sigma_2}{\sigma_1} e^{-\frac{Z^2}{2\sigma_1^2} + \frac{ Z^2 }{2 \sigma_2^2} } \Big |_{Z\sim \mu_2}  
 \laweq \sqrt{ {\alpha}} e^{ {\xi ^2}   ( \frac{1} 2 - \frac \alpha  {2 }  ) } ,
\end{align*} 
\begin{align*}
 \frac{\d \mu_2}{\d \mu_3} \Big |_{\mu_3} 
  \laweq \frac{\sigma_3}{\sigma_2} e^{-\frac{Z^2}{2\sigma_2^2} + \frac{ Z^2 }{2 \sigma_3^2} } \Big |_{Z\sim \mu_3}  
 \laweq \sqrt{ {\beta}} e^{ {\xi ^2}   ( \frac{1} 2 -  \frac \beta   {2}   ) } ,
\end{align*}   \begin{align*} 
 \frac{\d \mu_2}{\d \mu_1} \Big |_{\mu_1} 
\laweq \frac{\sigma_1}{\sigma_2} e^{-\frac{Z^2}{2\sigma_2^2} + \frac{ Z^2 }{2 \sigma_1^2} } \Big |_{Z\sim \mu_1}  
\laweq \sqrt{\frac1 {\alpha}} e^{ {\xi ^2}   ( \frac{1} 2 - \frac1  {2\alpha}  ) } ,
\end{align*}
and 
\begin{align*}
 \frac{\d \mu_3}{\d \mu_2} \Big |_{\mu_2} 
\laweq\frac{\sigma_2}{\sigma_3} e^{-\frac{Z^2}{2\sigma_3^2} + \frac{ Z^2 }{2 \sigma_2^2} } \Big |_{Z\sim \mu_2}  
\laweq \sqrt{ \frac 1 {\beta}} e^{ {\xi ^2}   ( \frac{1} 2 -  \frac 1  {2\beta }   ) } .
\end{align*} 
Therefore,   $\K((\mu_1,\mu_2),(\mu_2,\mu_3))\neq\emptyset  $   is equivalent to  
\begin{equation}
\label{eq:condition}
 \beta^{1/2} e^{ {\xi ^2}   ( \frac{1} 2 -  \frac \beta   {2}   ) }  \lcx \alpha^{1/2}  e^{ {\xi ^2}   ( \frac{1} 2 - \frac \alpha  {2 }  ) }  ~   \Longleftrightarrow ~
 \beta^{-1/2}  e^{ {\xi ^2}   ( \frac{1} 2 -  \frac 1  {2\beta }   ) }  \lcx \alpha^{-1/2}  e^{ {\xi ^2}   ( \frac{1} 2 -  \frac 1  {2\alpha }   ) }  . 
 \end{equation}
% \com{When using the peacock property, the expectation should be wrt the same measure, but here you take expectations wrt different measures $\mu_2,\mu_3$. To resolve this, use isomorphism to transform $(\sigma_2,\sigma_3)$ to $(\sigma_2^2/\sigma_3,\sigma_2)$.}
 A convenient result  we use here is Corollary 1.2 of \cite{HPRY11},  which says that
the stochastic process $( ( {1+2t})^{1/2} e^{ - {\xi ^2} t }  )_{t\ge 0} $ is a peacock; moreover, it is obvious that 
this process is non-stationary. This implies that, for $x,y\ge 1$,
$ \sqrt{ {y}} e^{ {\xi ^2}   ( \frac{1} 2 -  \frac y   {2}   ) }  \lcx \sqrt{ {x}} e^{ {\xi ^2}   ( \frac{1} 2 - \frac x  {2 }  ) }    $
if and only if $y\le x$. 
Hence, if $\alpha,\beta \ge 1$, then \eqref{eq:condition} is equivalent to $\beta \le \alpha$, thus case (i), and 
if $\alpha,\beta \le 1$,  then \eqref{eq:condition} is equivalent to $\beta \ge \alpha$, thus case (ii).

To show that   (i) and (ii) are the only cases where a transport from $(\mu_1,\mu_2)$ to $(\mu_2,\mu_3)$ exists, it suffices to exclude the case $\alpha<1< \beta$ or $\beta<1<\alpha$.
Note that in this case 
 $ \sqrt{ {\beta}} e^{ {\xi ^2}   ( \frac{1} 2 -  \frac \beta   {2}   ) }  $  and  $ \sqrt{ {\alpha}} e^{ {\xi ^2}   ( \frac{1} 2 - \frac \alpha  {2 }  ) }    $ 
 have mismatch supports  (one bounded away from $-\infty$ and one bounded away from $\infty$), and the hence either order  in \eqref{eq:condition}  is not possible.    
\end{proof} 

\sloppy The condition in Proposition \ref{prop:gaussian1} is not sufficient when $T>3$. For example, consider $(\sigma_t)=(8,4,2,\sqrt{2},1)$. If $\kappa\in\K((\mu_1,\dots,\mu_{d-1}),(\mu_2,\dots,\mu_d))$, then by Theorem \ref{2way} in Appendix \ref{app:Pf6}, $\kappa(x;\{\pm x/2\})=\kappa(x;\{\pm x/\sqrt{2}\})=1$, a contradiction.

\begin{proof}[Proof of Proposition \ref{2infp}]
By symmetry, it suffices to consider the case $d=2$ and we assume that $i=1,j=2$.

 Consider the decomposition $Y=Y_1\cup Y_2$ where $Y_1=\{y\in Y\mid \nu_1'(y)\geq 1\}=Y_2^c$. Also fix an arbitrary $\kappa\in \K(\bmu,\bnu)$. Then for $B\subseteq Y_1$, we have
 \begin{align*}
     \kappa_\#(\mu_1-\mu_2)_+(B)&=\int_X\kappa(x;B)(\mu_1-\mu_2)_+(\d x)\\
     &\geq \int_X\kappa(x;B)(\mu_1-\mu_2)(\d x)=(\nu_1-\nu_2)(B)=(\nu_1-\nu_2)_+(B).
 \end{align*}
 In fact, this holds with $\kappa$ replaced by $\kappa|_{X_1}$, where we define $X_1=\{x\in X\mid \mu_1'(x)\geq 1\}=X_2^c$. Similarly, for $B\subseteq Y_2$, we have
 $$(\kappa|_{X_2})_\#(\mu_2-\mu_1)_+(B)\geq (\nu_2-\nu_1)_+(B).$$
 Therefore,
 denoting by $\eta_1$ the restriction of $\eta$ on the set $\{x\in X\mid \mu_1'(x)\geq 1\}$ and $\eta_2=\eta-\eta_1$, we obtain \begin{align*}
    \C_{\eta}(\kappa) &=\int_{X\times Y}c(x,y)\kappa^x(\d y)\eta(\d x)\\
    &=\int_{X_1\times Y}c(x,y)\kappa|_{X_1}(x;\d y)\eta_1(\d x)+\int_{X_2\times Y}c(x,y)\kappa|_{X_2}(x;\d y)\eta_2(\d x)\\
    &\geq   \inf_{\kappa\in \K((\mu_1-\mu_2)_+,(\nu_1-\nu_2)_+} \C_{\eta_1}(\kappa)  +\inf_{\kappa\in \K((\mu_2-\mu_1)_+,(\nu_2-\nu_1)_+}     \C_{\eta_2}({\kappa}) \\
    &=   \inf_{\kappa\in \K((\mu_1-\mu_2)_+,(\nu_1-\nu_2)_+) }\C_{\eta}(\kappa)+\inf_{{\kappa}\in \K((\mu_2-\mu_1)_+,(\nu_2-\nu_1)_+)}\C_{\eta}({\kappa}),
\end{align*}where in the last step we used the condition that for any $x\in X$, there exists $y\in Y$ such that $c(x,y)=0$.
Taking infimum over $\kappa\in \K(\bmu,\bnu)$ proves (\ref{2inf}). 
\end{proof}

%  Define $\mu'=(\mu_1-\mu_2)_+$ and for $\kappa_0\in\mathcal K(\bmu,\bnu)$, define $\mu$ by $$\frac{\d\mu}{\d\mu'}(x)=\kappa_0(x;Y_1).$$
% Then $\mu\leq \mu'=(\mu_1-\mu_2)_+$. Consider $\kappa(x;B):=\kappa_0(x;B)/\kappa_0(x;Y_1)$, the normalized kernel on $Y_1$.\footnote{Define $0/0=0$.} It follows that for $\nu:=\kappa_\#\mu$ and $B\subseteq Y_1,$
%     $$\nu(B)=\int_{X}\kappa_0(x;B)\mu'(\d x)\geq \int_{X}\kappa_0(x;B)(\mu_1-\mu_2)(\d x)=(\nu_1-\nu_2)(B).$$
% This gives
% \begin{align*}
%     \C_{\bmu}(\kappa_0)    &\geq   \inf\{\C_{\eta}(\kappa)\mid\ \kappa\in \K(\mu,(\nu_1-\nu_2)_+),\ \mu\leq (\mu_1-\mu_2)_+\}.
% \end{align*}
% A similar inequality holds with $(\mu_1-\mu_2)_+$ replaced by $(\mu_2-\mu_1)_+$: consider $\widetilde{\mu}'=(\mu_2-\mu_1)_+$ and define $\widetilde{\mu},\widetilde{\kappa}$ similarly, we have
% $$\C_{\bmu}(\kappa_0)\geq \inf\{\C_{\eta}(\widetilde{\kappa})\mid\ \widetilde{\kappa}\in \K(\widetilde{\mu},(\nu_2-\nu_1)_+),\ \widetilde{\mu}\leq (\mu_2-\mu_1)_+\}.$$
% Since the supports of $\mu,\widetilde{\mu}$ are disjoint (hence so are $\kappa_\#\mu$ and $\widetilde{\kappa}_\#\widetilde{\mu}$), we have

\label{D}

\subsection{Proofs  of results in Section \ref{sec:Kanto}}\label{AA}

\begin{proof}[Proof of Proposition \ref{prop4.1}]

For each stochastic kernel $\kappa\in \mathcal K(\bmu,\bnu)$, we can define a measure $\pi\in \P(X\times Y)$ such that \begin{align}
    \pi(A\times B)=\int_A\kappa(x;B)\eta(\d x)\text{ for all }A\subseteq X,\ B\subseteq Y. \label{pidef}
 \end{align} Such a measure $\pi$ exists and is unique by Carath\'{e}odory's extension theorem. It follows that for a nonnegative measurable function $f:X\to\R$, 
$$\int_{X\times B}f(x)\pi(\d x,\d y)=\int_X\kappa(x;B)f(x)\eta(\d x),$$
which can be proved by considering indicator functions first and then using monotone convergence. Plugging in $f:={\d\mu_j}/{\d\eta}$ we obtain for any $B\subseteq Y,$
$$\nu_j(B)\leq \int_X\kappa(x;B)\mu_j(\d x)=\int_X\kappa(x;B)\frac{\d\mu_j}{\d\eta}(x)\eta(\d x)=\int_{X\times B}\frac{\d\mu_j}{\d\eta}(x)\pi(\d x,\d y).$$ In addition, for any $A\subseteq X$, $\pi(A\times Y)=\int_A\kappa(x;Y)\eta(\d x)=\eta(A)$, so by definition, $\pi\in \Pi_\eta(\bmu,\bnu)$.

On the other hand, given $\pi\in \Pi_\eta(\bmu,\bnu)$, we have by definition $\pi\circ\pi_1^{-1}=\eta$ where $\pi_1$ is projection onto $X$. By the disintegration theorem for product spaces, there exists a stochastic kernel $\kappa:\R\to\P(Y)$ such that for $A\subseteq X,\ B\subseteq Y,$
$$\pi(A\times B)=\int_A\kappa(x;B)\pi\circ\pi_1^{-1}(\d x)=\int_A\kappa(x;B)\eta(\d x),$$
which is exactly (\ref{pidef}). Similarly as above, we have
$$\nu_j(B)\leq \int_{X\times B}\frac{\d\mu_j}{\d\eta}(x)\pi(\d x,\d y)=\int_X\kappa(x;B)\frac{\d\mu_j}{\d\eta}(x)\eta(\d x)=\int_X\kappa(x;B)\mu_j(\d x),$$
thus $\kappa\in \mathcal K(\bmu,\bnu).$ 
%\com{in the above example I used $\kappa(x;B)$ instead of $\kappa(x;B)$ to emphasize that $\kappa(x)$ is understood as a measure.}. 
\end{proof}

Next, we turn to the proof of Theorem \ref{infs}. We first show a useful lemma, Lemma \ref{ca} below, which will be used to show that joint non-atomicity is sufficient for the equality between the optimal values of Monge and Kantorovich formulations of simultaneous transport in Section \ref{sec:Kanto}.

In what follows,  $\mathcal B$ is always  the Borel $\sigma$-field on $\R$.
We  first define another notion of  joint non-atomicity   introduced by \cite{D21}. This notion is similar to our Definition \ref{def2}, which was proposed by \cite{SSWW19}, but this time defined for $\sigma$-fields.  
%To avoid terminological collision, we will replace ``joint" used by \cite{D21} with ``relative" in our next definition.  
Both \cite{SSWW19} and \cite{D21} called their properties as being ``conditionally atomless" (and they are indeed equivalent in some sense as discussed by \cite{D21}; see Lemma \ref{lem:useful}). Recall that we renamed the notion from \cite{SSWW19} as  joint non-atomicity. All inequalities below involving conditional expectations are in the almost sure sense. 
\begin{definition}\label{def2prime}
Let  $(\Omega,\mathcal G,\mu)$ be a measure space. 
We say that $(\mathcal G,\mu)$ is atomless conditionally to the sub-$\sigma$-field $\mathcal F\subseteq \mathcal G$, if 
for all $A\in \mathcal G$ with $\mu(A)>0$, there exists $A'\subseteq A$, $A'\in \mathcal G$, such that 
 $$
\E^{\mu}[\bone_A|\mathcal F]>0 ~\Longrightarrow~ 0< \E^{\mu}[\bone_{A'}|\mathcal F]<\E^{\mu}[\bone_A|\mathcal F].
 $$
\end{definition}

Intuitively, the requirement in Definition \ref{def2prime} means that any set $A$ can be divided into smaller (measured by $\mu$) sets, conditionally on $\mathcal F$.
\cite{D21} showed that the two notions of conditional non-atomicity are equivalent in the   sense of Lemma \ref{lem:useful}.
This equivalence is anticipated because, in the unconditional setting, any set being divisible (corresponding to Definition \ref{def2prime}) is equivalent to the existence of a continuously distributed random variable (corresponding to Definition \ref{def2}); see e.g., Lemma D.1 of \cite{VW21}. 

\begin{lemma}\label{lem:useful}
Let   $ \mu$ be any strictly positive convex combination of $\bmu \in \mathcal P(X)^d$.
Then $\mathbf \bmu$ is jointly atomless if and only if $(\mathcal B(X),\mu)$ is atomless conditionally to $\sigma (\d \bmu/\d \mu)$.
%For $\mathbf \bmu=(\mu_1,\dots,\mu_d)\in \M_1^d$, define  $\mathcal F=\sigma (\d \mu_1/\d \mu,\dots,\d \mu_m/\d \mu)$, where $ \mu$ is any strictly positive convex combination of $\mu_1,\dots,\mu_d$.
%Then $\mathbf \bmu$ is jointly atomless if and only if $(\mathcal B,\mu)$ is atomless conditionally to $\mathcal F$.
\end{lemma}
\begin{proof}
This statement follows from Theorem 2.3 of \cite{D21}.
The connection between the two notions of conditional non-atomicity is   discussed in Remark 2.11 of \cite{D21}.
\end{proof}

Next, we are ready to give a useful lemma for non-atomicity on a subset of the sample space. 

\begin{lemma}
Let $\mathbf \bmu=(\mu_1,\dots,\mu_d)\in \P(X)^d$ be  jointly atomless. Consider an arbitrary Borel set $B\subseteq \R$ and, without loss of generality, assume $\mu_j(B)>0$ for $1\leq j\leq m$ where $m\leq d$. The normalized tuple $\bmu_B $ of probability measures on $B$, given by
$$\bmu_B=\left(\frac{\mu_1|_B}{\mu_1(B)},\dots,\frac{\mu_m|_B}{\mu_m(B)}\right),$$
is again jointly atomless.\label{ca}
\end{lemma}

\begin{proof} 
Let $\mu=(\mu_1+\dots+\mu_m)/m $ and $\mathcal F=\sigma (\d \bmu/\d \mu)=\sigma (\d \mu_1/\d \mu,\dots,\d \mu_m/\d \mu)$.
Define $\mathcal F_B= \{A\cap B\mid A\in \mathcal F\}$ and similarly for $\mathcal B_B$,
 and $\mu_B(A)=\mu(A\cap B)/\mu(B)$ for $A\in \mathcal B$.  
 
Take $A\in \mathcal B_B$ with $\mu (A) = \mu(B)\mu_B(A)>0$. 
 Note that $(\mu_1,\dots,\mu_m)$ is jointly atomless. 
Using Lemma \ref{lem:useful}, $(\mathcal B,\mu)$ is atomless conditionally to $\mathcal F$. 
By definition, there exists $A'\subseteq A$, $A'\in \mathcal B$  such that 
\begin{equation}
\E^{\mu}[\bone_A|\mathcal F]>0 ~\Longrightarrow~ 0< \E^{\mu}[\bone_{A'}|\mathcal F] <\E^{\mu}[\bone_A|\mathcal F].
\label{eq:cal1}
 \end{equation}
 Since $A'\subseteq A\subseteq B$, we have 
 $$ \E^{\mu_B}[\bone_{A}|\mathcal F_B]= \E^{\mu_B}[\bone_{A}|\mathcal F]= \E^{\mu}[\bone_{A}|\mathcal F],$$
 and the same holds for $A'$ in place of $A$. 
 As a consequence,  \eqref{eq:cal1} leads to  
\begin{equation}
\E^{\mu_B}[\bone_{A}|\mathcal F_B]>0  ~ \Longrightarrow~ 0< \E^{\mu_B}[\bone_{A'}|\mathcal F_B] <\E^{\mu_B}[\bone_{A'}|\mathcal F_B] 
\label{eq:cal2}
 \end{equation} 
 Note also that $A'\in \mathcal B_B$ by definition. 
Therefore,  by treating $\mu_B$ as a probability measure on $\mathcal B_B$,   \eqref{eq:cal2} implies that $(\mathcal B_B,\mu_B)$ is   atomless conditionally to $\mathcal F_B$.
Noting that $\mu_B$ is a strictly positive convex combination of components of $\boldsymbol\mu_B$,
and using Lemma \ref{lem:useful} again,   we conclude that $\bmu_B$ is jointly atomless.
\end{proof}

\begin{proof}[Proof of Theorem \ref{infs}]
We can without loss of generality assume that $X=Y$ by considering $\bmu,\bnu$ as measures on the compact space $X\times Y$, and that each $\mu_j$ is a probability measure. We have shown above that Monge transports are special cases as Kantorovich transports, thus the infimum cost among Monge transports is bounded below by that among Kantorovich transports. 

To prove the other direction, we first assume that there is $\delta>0$ such that $\frac{\d\eta}{\d\mub}(x)\geq\delta$ for all $x\in X$. For each $n\in\N$ we partition $X$ into countably many Borel sets $\{K_{i,n}\}_{i\in\N}$ of diameter smaller than $1/n$ and such that for each $i$, $$\frac{\sup_{x\in K_{i,n}}\frac{\d\eta}{\d\mub}(x)}{\inf_{x\in K_{i,n}}\frac{\d\eta}{\d\mub}(x)}\leq 1+\frac{1}{n}.$$ Consider a transport plan $\kappa\in \K(\bmu,\bnu)$. Define \begin{align}\bmu^{i,n}:=\bmu|_{K_{i,n}}\text{ and for }B\subseteq X,\ \bnu^{i,n}(B):=\int_{K_{i,n}}\kappa(x;B)\bmu(\d x).\label{2in}\end{align}It is then obvious that $\kappa_{i,n}:=\kappa|_{K_{i,n}}\in \K(\bmu^{i,n},\bnu^{i,n})$. Consider the normalized probability measures
$$\d \widetilde{\mu}_{j}^{i,n}=\frac{\d\mu^{i,n}_j}{\mu^{i,n}_j(K_{i,n})};\ \d \widetilde{\nu}_{j}^{i,n}=\frac{\d\nu^{i,n}_j}{\nu^{i,n}_j(X)}.$$ 
It is also easy to check that $\kappa_{i,n}\in \K(\widetilde{\bmu}^{i,n},\widetilde{\bnu}^{i,n})$. By Proposition \ref{prop:ssww}, $(\widetilde{\bmu}^{i,n})'|_{\mub^{i,n}}\gcx(\widetilde{\bnu}^{i,n})'|_{\nub^{i,n}}$. By Lemma \ref{ca}, $\widetilde{\bmu}^{i,n}$ is jointly atomless, so that applying Proposition \ref{prop:ssww} again, we conclude that $\mathcal T(\widetilde{\bmu}^{i,n},\widetilde{\bnu}^{i,n})$ is non-empty.\footnote{We can forget about the components $j$ where $\widetilde{\mu}_j^{i,n}(K_{i,n})=0$ because the transport condition is trivially satisfied there.} That is, there exist Monge transports $T_{i,n}:K_{i,n}\to X$ such that $\bmu^{i,n}\circ T_{i,n}^{-1}=\bnu^{i,n}$. By gluing these, we obtain a Monge transport $T_n:X\to X$. Note that $T_n\in\T(\bmu,\bnu)$ since 
$$\sum_{i\in\N}\bnu^{i,n}(B)=\int_X\kappa(x;B)\bmu(\d x)\geq \nub(B).$$

Define $\kappa_n(x;B):=\bone_{\{T_n(x)\in B\}}$, then $\kappa_n\in \K(\bmu,\bnu)$. Our goal now is to show that 
\begin{align}
    \C_\eta(T_n)=\int_{X\times X}c(x,y)\eta\otimes \kappa_n(\d x,\d y)\to \int_{X\times X}c(x,y)\eta\otimes \kappa(\d x,\d y).\label{q3}
\end{align}
Let us define cost functions 
$$\bar{c}_n(x,y):=\sup_{x_0\in K_{i,n},y_0\in K_{\ell,n}}c(x_0,y_0)\text{ if }x\in K_{i,n}\text{ and }y\in K_{\ell,n}.$$ Then since $c$ is uniform continuous on $X\times X$, we have \begin{align}
    \int_{X\times X}|\bar{c}_n(x,y)-c(x,y)|\eta\otimes \kappa(\d x,\d y)\to 0.\label{q1}
\end{align}
On the other hand, 
\begin{align*}
    \int_{K_{i,n}\times K_{\ell,n}}\eta\otimes \kappa_n(\d x,\d y)&=\int_{K_{i,n}}\bone_{\{T_n(x)\in K_{\ell,n}\}}\eta(\d x)\\
    &=\int_{K_{i,n}}\bone_{\{T_n(x)\in K_{\ell,n}\}}\frac{\d\eta|_{K_{i,n}}}{\d\mub^{i,n}}(x)\mub^{i,n}(\d x)\\
    &\leq \sup_{x\in K_{i,n}}\frac{\d\eta}{\d\mub}(x)\nub^{i,n}(K_{\ell,n})\\
    &\leq \left(1+\frac{1}{n}\right)\inf_{x\in K_{i,n}}\frac{\d\eta}{\d\mub}(x)\nub^{i,n}(K_{\ell,n})\\
    &\leq \left(1+\frac{1}{n}\right)\int_{K_{i,n}\times K_{\ell,n}}\eta\otimes \kappa(\d x,\d y).
\end{align*}
Applying this in the second inequality below yields that
\begin{align}
    \C_\eta(T_n)&\leq \int_{X\times X}\bar{c}_n(x,y)\eta\otimes \kappa_n(\d x,\d y)\nonumber\\
    &=\sum_{i\in\N}\sum_{\ell\in \N}\sup_{x\in K_{i,n},y\in K_{\ell,n}}c(x,y)\int_{K_{i,n}\times K_{\ell,n}}\eta\otimes \kappa_n(\d x,\d y)\nonumber\\
    &\leq \left(1+\frac{1}{n}\right)\sum_{i\in\N}\sum_{\ell\in \N}\sup_{x\in K_{i,n},y\in K_{\ell,n}}c(x,y)\int_{K_{i,n}\times K_{\ell,n}}\eta\otimes \kappa(\d x,\d y)\nonumber\\
    &=\left(1+\frac{1}{n}\right)\int_{X\times X}\bar{c}_n(x,y)\eta\otimes \kappa(\d x,\d y).\label{q2}
\end{align}
Combining (\ref{q1}) and (\ref{q2}), and since $c\geq 0$, we obtain $$\limsup_{n\to\infty} \C_\eta(T_n)\leq  \int_{X\times X}c(x,y)\eta\otimes \kappa(\d x,\d y).$$
The liminf part is similar. We have thus proved (\ref{q3}).

In the general case where $\d\eta/\d\mub$ is not bounded below by $\delta>0$, we consider $\eta_\delta:=\eta+\delta\mub$. Since $c$ is bounded, we have uniformly for $T\in \T(\bmu,\bnu)$ and $\kappa\in\K(\bmu,\bnu)$,
$$\C_{\eta_\delta}(T)\to \C_{\eta}(T)\text{ and }\C_{\eta_\delta}(\kappa)\to \C_{\eta}(\kappa)\text{ as }\delta\to 0.$$
This completes the proof.
\end{proof}

\begin{proof}[Proof of Proposition \ref{connect}]
% Since $(\bnu^n)$ is increasing in $n$, the set $\Pi_{\eta}(\bmu,\bnu^n)$ is decreasing in $n$. Thus the quantity $\inf_{\pi\in\Pi_{\eta}(\bmu,\bnu^n)}\C(\pi)$ is non-decreasing in $n$ hence admits a limit. 
Denote by $$J_n:=\inf_{\pi\in\Pi_{\eta}(\bmu,\bnu^n)}\C(\pi).$$It suffices to show for each subsequence $\{n_k\}$ there exists a further subsequence $\{n_{k_\ell}\}$ such that $J_{n_{k_\ell}}\to \inf_{\pi\in\Pi_{\eta}(\bmu,\bnu)}\C(\pi).$

%Evidently, $$\lim_{n\to\infty}\inf_{\pi\in\Pi_{\eta}(\bmu,\bnu^n)}\C(\pi)\leq\inf_{\pi\in\Pi_{\eta}(\bmu,\bnu)}\C(\pi).$$ 
Consider for each $n$ a measure $\pi_n\in\Pi_\eta(\bmu,\bnu^n)$. Then since $X,Y$ are compact, the sequence $(\pi_{n_k})$ is tight, so a subsequence $(\pi_{n_{k_\ell}})$ converges weakly to some $\pi\in\P(X\times Y)$. Since $\d\bmu/\d\eta$ is continuous, the operations defining $\Pi_\eta(\bmu,\bnu)$ is continuous with respect to weak topology in \eqref{newpi}, thus we have $\pi\in\Pi_\eta(\bmu,\bnu)$. Since $c(x,y)$ is continuous, this gives that
$$\lim_{k\to\infty}\C(\pi_{n_k})=\C(\pi).$$
Taking infimum yields that $$\liminf_{\ell\to\infty} J_{n_{k_\ell}}\geq \inf_{\pi\in\Pi_{\eta}(\bmu,\bnu)}\C(\pi).$$
Since $\bnu^n\leq\bnu$, we also have$$J_{n_{k_\ell}}\leq \inf_{\pi\in\Pi_{\eta}(\bmu,\bnu)}\C(\pi),$$ thus $J_{n_{k_\ell}}\to \inf_{\pi\in\Pi_{\eta}(\bmu,\bnu)}\C(\pi),$ completing the proof.
\end{proof}

The key to proving Theorem \ref{duality1} is the following minimax theorem,  which could be found in \cite[Theorem 2.4.1]{AH99} and sometimes referred to as Sion's minimax theorem.

\begin{lemma}\label{thm:minimax}
Let $X$ be a compact Hausdorff space, $Y$ be an arbitrary set, and $f:X\times Y\to \R\cup\{\infty\}$. Assume that $f$ is lower semi-continuous in $x$ for each fixed $y$, convex in $x$, and concave in $y$. Then 
$$\min_{x\in X}\sup_{y\in Y}f(x,y)=\sup_{y\in Y}\min_{x\in X}f(x,y).$$
\end{lemma}

\begin{proof}[Proof of Theorem \ref{duality1}]
The $\geq$ direction of \eqref{dua} being obvious, we focus on the $\leq$ part. We first assume $c$ is bounded continuous. For a Polish space $X$, we denote by $C_{\rm b}(X)$ the space of all bounded continuous functions on $X$.  First observe that by definition \eqref{o1}, for $\pi\in\Pi_\eta(\mub,\nub)$,
\begin{align}
&\hspace{0.5cm}\sup_{\substack{\phi\in C_{\rm b}(X)\\ \bpsi\in C_{\rm b}^d(Y)}}\Bigg\{\int_X\phi\, \d\eta -\int_{X\times Y}\phi(x)\pi(\d x,\d y) + \int_Y \bpsi^\top\d\bnu-\int_{X\times Y}\bpsi(y)^{\top}\frac{\d\bmu}{\d \eta}(x)\pi(\d x,\d y)\Bigg\}\nonumber\\
&\qquad\qquad=\begin{cases}0&\text{ if }\pi\in\Pi_\eta(\bmu,\bnu);\\
\infty &\text{ elsewhere}.\end{cases}\label{3}
\end{align}
 For $\pi\in\Pi_\eta(\mub,\nub)$, we have by using (\ref{3}) that
\begin{align}
&\hspace{0.5cm}\sup\Bigg\{\int_{X\times Y}p(x,y)\pi(\d x,\d y)+\int_X\phi\, \d\eta+\int_Y \bpsi^\top\d\bnu   \mid p\in C_{\rm b}(X\times Y), \nonumber\\
&\hspace{1.5cm}\phi\in C_{\rm b}(X),\ \bpsi\in C_{\rm b}^d(Y),\ \phi(x)+\bpsi(y)^{\top}\frac{\d\bmu}{\d \eta}(x)\leq c(x,y)-p(x,y)\Bigg\}\nonumber\\
&=\sup\Bigg\{\int_{X\times Y}c(x,y)\pi(\d x,\d y)+\int_X\phi\, \d\eta-\int_{X\times Y}\phi(x)\pi(\d x,\d y)\nonumber\\
&\hspace{0.5cm}+\int_Y \bpsi^\top\d\bnu-\int_{X\times Y}\bpsi(y)^{\top}\frac{\d\bmu}{\d \eta}(x)\pi(\d x,\d y)\mid \phi\in C_{\rm b}(X),\ \bpsi\in C_{\rm b}^d(Y)\Bigg\}\nonumber\\
&=\begin{cases}\int_{X\times Y}c(x,y)\pi(\d x,\d y)&\text{ if }\pi\in\Pi_\eta(\bmu,\bnu);\\
\infty&\text{ elsewhere}.\end{cases}\label{eq:supinf}
\end{align} Since $\d\mub/\d\eta$ is bounded continuous, the set $\Pi_\eta(\mub,\nub)$ is weakly compact. 
Using \eqref{eq:supinf} and Lemma  \ref{thm:minimax}, we obtain
\begin{align*}
    &\hspace{0.5cm}\min_{\pi\in\Pi_{\eta}(\bmu,\bnu)}\int_{X\times Y}c(x,y)\pi(\d x,\d y)\\
    &=\min_{\pi\in \Pi_\eta(\mub,\nub)}\sup\Bigg\{\int_{X\times Y}p(x,y)\pi(\d x,\d y)+\int_X\phi\, \d\eta+\int_Y \bpsi^\top\d\bnu\mid \phi\in C_{\rm b}(X),\\
&\hspace{1cm} \bpsi\in C_{\rm b}^d(Y),\ p\in C_{\rm b}(X\times Y),\  \phi(x)+\bpsi(y)^{\top}\frac{\d\bmu}{\d \eta}(x)\leq c(x,y)-p(x,y)\Bigg\}\\
&=\sup\Bigg\{\min_{\pi\in \Pi_\eta(\mub,\nub)}\int_{X\times Y}p(x,y)\pi(\d x,\d y)+\int_X\phi\, \d\eta+\int_Y \bpsi^\top\d\bnu\mid\phi\in C_{\rm b}(X), \\
&\hspace{1cm} \bpsi\in C_{\rm b}^d(Y),\ p\in C_{\rm b}(X\times Y),\  \phi(x)+\bpsi(y)^{\top}\frac{\d\bmu}{\d \eta}(x)\leq c(x,y)-p(x,y)\Bigg\}.
\end{align*}
By duality for classic optimal transport, 
\begin{align*}&\hspace{0.5cm}\min_{\pi\in \Pi_\eta(\mub,\nub)}\int_{X\times Y}p(x,y)\pi(\d x,\d y)\\
&=\sup\Bigg\{\int_X\widetilde{\phi}\,\d\eta+\int_Y\widetilde{\psi}\,\d\nub\mid \widetilde{\phi}\in C_{\rm b}(X),\ \widetilde{\psi}\in C_{\rm b}(Y),\  \widetilde{\phi}(x)+\widetilde{\psi}(y)\frac{\d\mub}{\d\eta}(x)\leq p(x,y) \Bigg\}.\end{align*}
Rearranging the terms we have
\begin{align*}
&\hspace{0.5cm}\min_{\pi\in\Pi_{\eta}(\bmu,\bnu)}\int_{X\times Y}c(x,y)\pi(\d x,\d y)\\
&=\sup\Bigg\{\int_X\widetilde{\phi}\,\d\eta+\int_Y\widetilde{\psi}\,\d\nub +\int_X\phi\, \d\eta+\int_Y \bpsi^\top\d\bnu\mid  \widetilde{\phi}\in C_{\rm b}(X),\ \widetilde{\psi}\in C_{\rm b}(Y),\  \widetilde{\phi}(x)+\widetilde{\psi}(y)\leq p(x,y);\\
&\hspace{4cm}\phi\in C_{\rm b}(X),\ \bpsi\in C_{\rm b}^d(Y),\ \phi(x)+\bpsi(y)^{\top}\frac{\d\bmu}{\d \eta}(x)\leq c(x,y)-p(x,y)\Bigg\}\\
&\leq \sup\Bigg\{\int_X\widetilde{\phi}\,\d\eta+\int_Y\widetilde{\psi}\,\d\nub+\int_X\phi\, \d\eta+\int_Y \bpsi^\top\d\bnu\mid \widetilde{\phi}\in C_{\rm b}(X),\ \widetilde{\psi}\in C_{\rm b}(Y),\\
&\hspace{3cm}\phi\in C_{\rm b}(X),\ \bpsi\in C_{\rm b}^d(Y),\ \phi(x)+\bpsi(y)^{\top}\frac{\d\bmu}{\d \eta}(x)+\widetilde{\phi}(x)+\widetilde{\psi}(y)\leq c(x,y)\Bigg\}\\
&\leq \sup\Bigg\{\int_X\phi\,  \d\eta+\int_Y \bpsi^\top \d \bnu\mid (\phi,\bpsi)\in C_{\rm b}(X)\times C_{\rm b}^d (Y),\ \phi (x)+\bpsi(y)^\top  
\frac{\d\bmu}{\d\eta}(x)\le c(x,y)\Bigg\},
\end{align*}
thus proving the duality formula (\ref{dua}) in the case where $c$ is bounded continuous.

Consider the general case where $c$ is lower semi-continuous, possibly taking values in $\R\cup\{\infty\}$. As in \cite{V03}, we can write $c=\sup c_n$ where each $c_n$ is continuous bounded and $c_n$ is nondecreasing in $n$. For $(\phi,\bpsi)\in \Phi_c$, we denote $\varphi^{d}(\phi,\bpsi):=\int_X\phi\, \d\eta+\int_Y \bpsi^\top\d\bnu$. Also write $I_n(\pi)=\int_{X\times Y}c_n\d\pi$. We aim to show that
\begin{align}\inf_{\pi\in\Pi_\eta(\bmu,\bnu)}I(\pi)&\leq\sup_n\inf_{\pi\in\Pi_\eta(\bmu,\bnu)}I_n(\pi) \leq\sup_n\sup_{(\phi,\bpsi)\in \Phi_{c_n}}\varphi^{d}(\phi,\bpsi)\leq \sup_{(\phi,\bpsi)\in \Phi_c}\varphi^{d}(\phi,\bpsi).\label{s}\end{align}The second inequality follows from the first part of the proof, and the third inequality follows from that $\{c_n\}$ is nondecreasing in $n$, so it suffices to prove the first equality. 

Since $\d\eta/\d\mub$ is bounded, $\Pi_\eta(\bmu,\bnu)$ is tight. Consider a minimizing sequence $\{\pi_{n,k}\}$ for $\inf I_n(\pi)$. By Prokhorov's theorem, we can extract a subsequence, say $\pi_{n,k}\to\pi_n$ weakly as $k\to\infty$. Note that $\pi_n\in \Pi_\eta(\bmu,\bnu)$ since $\d\bmu/\d\eta$ is continuous. Thus the infimum is attained at $\pi_n$. Again by Prokhorov's theorem, $\pi_n\to\pi_*$ up to extracting a subsequence. By monotone convergence, $I_n(\pi_*)\to I(\pi_*)$. Thus for any $\ee>0$, we can find $N,M$ such that
$$\inf_{\pi\in\Pi_\eta(\bmu,\bnu)}I(\pi)\leq I(\pi_*)<I_N(\pi_*)+\ee<I_N(\pi_M)+2\ee.$$Letting $\ee\to 0$ proves the first inequality of (\ref{s}). Combining with the trivial bound $$\inf_{\pi\in\Pi_\eta(\bmu,\bnu)}I(\pi)\geq \sup_{(\phi,\bpsi)\in \Phi_c}\varphi^d(\phi,\bpsi)$$completes the proof of (\ref{dua}).

To show that the infimum of (\ref{dua}) is attained we still apply Prokhorov's theorem. For a minimizing sequence $\{\pi_k\}$ it has a subsequence converging to $\pi_*\in\Pi_\eta(\bmu,\bnu)$ (since $\d\bmu/\d\eta$ is continuous) and 
$$I(\pi_*)=\lim_{n\to\infty}I_n(\pi_*)\leq\lim_{n\to\infty}\limsup_{k\to\infty}I_n(\pi_k)\leq \limsup_{k\to\infty}I(\pi_k)=\inf_{\pi\in\Pi_\eta(\bmu,\bnu)}I(\pi).$$
This shows the desired attainability.
\end{proof}

\subsection{Proofs of results in Section \ref{W}}
\label{app:Pf6}

\begin{proof}[Proof of Theorem \ref{thm:rep}]
First we prove the easy direction. Suppose that $\kappa_{\bz}\in \K_{\bz},~\hat{\kappa}\in \S_{b,1},$ and $\widetilde{\kappa}_{\bz'}\in \widetilde{\K}_{\bz'}$ for all $\bz,\bz'\in\rd$. Fix a measurable set $B\subseteq Y $. Since $\kappa_{\bz}\in \K_{\bz}$, for $V\subseteq [0,1]$,
$$\int_{A_\bz}\kappa_\bz^x(\{\bz\}\times V)\mu_\bz(\d x)=\tau(V).$$
Therefore we have
\begin{align*}\int_X \kappa^x(B)\bmu(\d x)   &=\int_X \int_{\cR }\int_{[0,1]}\kappa_\bz^x(\bmu'(x),\d u)\hat{\kappa}^{(\bmu'(x),u)}(\d \bz',\d u')\widetilde{\kappa}_{\bz'}^{(\bz',u')}(B)\bmu'(x)\mub(\d x)\\
    &=\int_{\rd}\int_{A_\bz}\int_{\cR }\int_{[0,1]}\kappa_\bz^x(\{\bz\}\times \d u)\hat{\kappa}^{(\bz,u)}(\d \bz',\d u')\widetilde{\kappa}_{\bz'}^{(\bz',u')}(B)\bz\mu_\bz(\d x)\mm(\d\bz)\\
    % &=\int_{\rd\times[0,1]}\int_{\rd\times[0,1]}\hat{\kappa}^{(\bz,u)}(\d \bz',\d u')\widetilde{\kappa}_{\bz'}^{(\bz',u')}(B)\left(\int_{A_\bz}\kappa_\bz^x(\{\bz\}\times \d u)\mu_\bz(\d x)\right)\bz \mm(\d \bz)\\
    &=\int_{\cR }\int_{\cR }\hat{\kappa}^{(\bz,u)}(\d \bz',\d u')\widetilde{\kappa}_{\bz'}^{(\bz',u')}(B)\bz\tau(\d u)\mm(\d \bz).
\end{align*}
Since $\hat{\kappa}\in \S_{b,1},$ it holds for $Z'\subseteq\rd$ and $V\subseteq[0,1]$,
$$\int_{\cR }\hat{\kappa}^{(\bz,u)}(Z'\times V)\bz\tau(\d u)\mm(\d \bz)= \int_{Z'}\bz'\tau(V)\nn(\d\bz').$$
This gives
\begin{align*}\int_X \kappa^x(B)\bmu(\d x)&= \int_{\cR }\widetilde{\kappa}_{\bz'}^{(\bz',u')}(B)\bz'\tau(\d u')\nn(\d\bz')\\
&=\int_{\cR }\int_{B_{\bz'}}\bone_B(y)\bz'\widetilde{\kappa}_{\bz'}^{(\bz',u')}(\d y)\tau(\d u')\nn(\d\bz').
\end{align*}Using $\widetilde{\kappa}_{\bz'}\in \widetilde{\K}_{\bz'}$, we have that for $B\subseteq Y $,
$$\int_{[0,1]}\widetilde{\kappa}_{\bz'}^{(\bz',u')}(B)\tau(\d u')=\nu_{\bz'}(B).$$
We conclude that
$$\int_X \kappa^x(B)\bmu(\d x)= \int_{\rd}\int_{B_{\bz'}}\bone_B(y)\bz'\nu_{\bz'}(\d y)\nn(\d\bz')=\bnu(B).$$

Consider now  kernels $\kappa_{\bz}\in \K_{\bz},~\widetilde{\kappa}_{\bz'}\in \widetilde{\K}_{\bz'},~\bz,\bz'\in\rd$ and $\kappa\in\K(\bmu,\bnu)$ as fixed, where $\kappa_\bz$ is backward Monge and $\widetilde{\kappa}_{\bz'}$ is Monge.   Denote by $T_\bz$ the inverse   of $\kappa_\bz$ which can be chosen as any Monge transport from $([0,1],\tau)$ to $(A_\bz,\mu_\bz)$. More precisely, we have for any $B\subseteq Y $,
\begin{align}
    \int_{X }\kappa^x(B)\mub(\d x)=\int_{\cR }\kappa^{T_\bz(u)}(B)\tau(\d u)\mm(\d\bz)\label{eq:tz}
\end{align}
In this case, we may compose the kernels to get a kernel
\begin{align}\hat{\kappa}^{(\bz,u)}(D):=\int_{\rd}\int_{B_{\bz'}}\kappa^{T_\bz(u)}(\d y)\widetilde{\kappa}_{\bz'}^y(D)\nn(\d \bz'),\label{eq:kappahat}\end{align}
as illustrated by Figure \ref{fig:commute}.

To show $\hat{\kappa}\in \S_{b,1}$, it suffices to show that for $Z'\subseteq \rd$ and $V\subseteq[0,1]$,
\begin{enumerate}[(i)]
    \item $\int_{\cR }\hat{\kappa}^{(\bz,u)}(Z'\times V)\tau(\d u)\mm(\d\bz)=\tau(V)\nn(Z')$;\label{t1}
    \item $\int_{\cR }\hat{\kappa}^{(\bz,u)}(Z'\times V)\bz\tau(\d u)\mm(\d\bz)= \int_{Z'}\bz'\tau(V)\nn(\d\bz').$\label{t2}
\end{enumerate}
To prove \eqref{t1}, we first claim that for $B\subseteq Y $,
\begin{align}
    \nu_{\bz'}(B)=\int_{X }\kappa^x(B\cap B_{\bz'})\mub(\d x).\label{eq:1}
\end{align}
This is a direct consequence of the uniqueness of disintegration and for $B\subseteq Y $,
\begin{align}
    \nub(B)&=\int_{\rd}\int_{A_\bz}\kappa^x(B)\mu_\bz(\d x)\mm(\d \bz)\label{step}\\
    &=\int_{\rd}\int_{A_\bz}\int_{\rd}\kappa^x(B\cap B_{\bz'})\nn(\d\bz')\mu_\bz(\d x)\mm(\d \bz)\nonumber\\
    &=\int_{\rd}\int_X \kappa^x(B\cap B_{\bz'})\mub(\d x)\nn(\d\bz').\nonumber
\end{align}
It follows that using \eqref{eq:1} in the second equality, \eqref{eq:tz} in the third, \eqref{eq:kappahat} in the fourth, that
\begin{align*}
    &\hspace{0.5cm}\tau(E_1)Q(E_2)\\&=\int_{\rd}\int_{B_{\bz'}}\widetilde{\kappa}_{\bz'}^y(E_1\times E_2)\nu_{\bz'}(\d y)\nn(\d \bz')\\
    &=\int_{\rd}\int_{B_{\bz'}}\widetilde{\kappa}_{\bz'}^y(E_1\times E_2)\int_{X }\kappa^x(\d y\cap B_{\bz'})\mub(\d x)\nn(\d \bz')\\
    &=\int_{\rd}\int_{B_{\bz'}}\widetilde{\kappa}_{\bz'}^y(E_1\times E_2)\int_{\cR }\kappa^{T_\bz(u)}(\d y\cap B_{\bz'})\tau(\d u)\mm(\d\bz)\nn(\d \bz')\\
    &=\int_{\cR }\hat{\kappa}^{(\bz,u)}(E_1\times E_2)\tau(\d u)\mm(\d\bz).
\end{align*}
To show \eqref{t2}, we first note that by definition of $\widetilde{\kappa}_{\bz'}$, for all $\bz'$ and $V\subseteq [0,1]$,
$$\tau(V)=\int_{B_{\bz'}}\widetilde{\kappa}_{\bz'}^y(\bz',V)\nub(\d y).$$
Therefore, since $\kappa\in\K(\bmu,\bnu)$, for $Z'\subseteq\rd$ and $V\subseteq[0,1]$, we have
\begin{align*}
    \int_{Z'}\bz'\tau(V)\nn(\d\bz')&=\int_{Z'}\int_{B_{\bz'}}\widetilde{\kappa}_{\bz'}^y(\bz',V)\nub(\d y)\bz'\nn(\d \bz')\\
    &=\int_{\rd}\int_{B_{\bz'}}\bone_{\{\bnu'(y)\in Z'\}}\widetilde{\kappa}_{\bz'}^y(\bz',V)\nub(\d y)\bz'\nn(\d \bz')\\
    &=\int_{\rd}\int_{B_{\bz'}}\widetilde{\kappa}_{\bz'}^y(Z'\times V)\bnu(\d y)\nn(\d \bz')\\
    &=\int_{\rd}\int_{B_{\bz'}}\widetilde{\kappa}_{\bz'}^y(Z'\times V)\int_X \kappa^x(\d y)\bmu'(x)\mub(\d x)\nn(\d\bz').
\end{align*}
By \eqref{eq:tz} and \eqref{eq:kappahat}, we obtain for an arbitrary $Z'\subseteq \rd$ that
\begin{align*}
    &\hspace{0.5cm}\int_{Z'}\bz'\tau(V)\nn(\d\bz')\\   &= \int_{\rd}\int_{B_{\bz'}}\widetilde{\kappa}_{\bz'}^y(Z'\times V)\int_{\cR }\kappa^{T_\bz(u)}(\d y)\bz\tau(\d u)\mm(\d\bz)\nn(\d \bz')\\
    &=\int_{\cR }\int_{\rd}\int_{B_{{\bz'}}}\kappa^{T_\bz(u)}(\d y)\widetilde{\kappa}_{{\hat{\bz}}}^y(Z'\times V)\nn(\d \bz')\bz\tau(\d u)\mm(\d\bz)\\
    &=\int_{\cR }\hat{\kappa}^{(\bz,u)}(Z'\times V)\bz\tau(\d u)\mm(\d\bz).
\end{align*}
This finishes the proof that $\hat{\kappa}\in\S_{b,1}$. Next we show that after composing these kernels we get back $\kappa$, i.e., for $\bmu'(x)=\bz$ and $B\subseteq Y$, that
\begin{align}
    \kappa^x(B)=\int_{\cR }\int_{[0,1]}\kappa_\bz^x(\{\bz\}\times \d u)\hat{\kappa}^{(\bz,u)}(\d\bz',\d u')\bone_{\{S_{\bz'}(u')\in B\}}.\label{eq:getbackkappa}
\end{align}
Since $T_\bz$ and $\kappa_\bz$ forms inverses of each other, we have
\begin{align*}
    \kappa^x(B)  &=\int_{[0,1]}\kappa_\bz^x(\{\bz\}\times \d u)\kappa^{T_\bz(u)}(B)\\
    &=\int_{[0,1]}\kappa_\bz^x(\{\bz\}\times \d u)\int_{\rd}\int_{B_{{\bz}'}}\kappa^{T_\bz(u)}(\d y)\nn(\d {\bz}')\bone_{\{y\in B\}}.
\end{align*}
Similarly, since $S_{\bz'}$ and $\widetilde{\kappa}_{\bz'}$ are inverses of each other, it holds that
$$\bone_{\{y\in B\}}=\widetilde{\kappa}_{{\bz}'}^y( {\bz}',(S_{\bz'})^{-1}(B)).$$
Therefore, using \eqref{eq:kappahat} in the last step yields
\begin{align*}
    &\hspace{0.5cm}\kappa^x(B)\\
    &=\int_{[0,1]}\kappa_\bz^x(\{\bz\}\times \d u)\int_{\rd}\int_{B_{{\bz}'}}\kappa^{T_\bz(u)}(\d y)\widetilde{\kappa}_{{\bz}'}^y( {\bz}',(S_{\bz'})^{-1}(B))\nn(\d {\bz}')\\
    &=\int_{\cR }\int_{[0,1]}\kappa_\bz^x(\{\bz\}\times \d u)\int_{B_{{\bz}'}}\kappa^{T_\bz(u)}(\d y)\widetilde{\kappa}_{{\bz}'}^y( \{{\bz}'\}\times\d u')\bone_{\{S_{\bz'}(u')\in B\}}\nn(\d {\bz}')\\
    &=\int_{\cR }\int_{[0,1]}\kappa_\bz^x(\{\bz\}\times \d u)\int_{\rd}\int_{B_{\widetilde{\bz}'}}\bone_{\{S_{\bz'}(u')\in B\}}\kappa^{T_\bz(u)}(\d y)\widetilde{\kappa}_{\widetilde{\bz}'}^y(\d {\bz}',\d u')\nn(\d \widetilde{\bz}')\\
    &=\int_{\cR }\int_{[0,1]}\kappa_\bz^x(\{\bz\}\times \d u)\hat{\kappa}^{(\bz,u)}(\d\bz',\d u')\bone_{\{S_{\bz'}(u')\in B\}},
\end{align*}
proving \eqref{eq:getbackkappa}, hence concluding the proof.
\end{proof}

\begin{proof}[Proof of Corollary \ref{coro:commute}]
    The existence of a backward martingale Monge coupling follows from Proposition \ref{prop:ssww} and Theorem 2.1 of \cite{NWZ22}. We let $\hat{\kappa}$ in \eqref{eq:rep} be induced by the Monge map $h$ in the $\rd$ dimension and identity in the $[0,1]$ dimension. Thus, there exists $\kappa\in \K(\bmu,\bnu)$ given by \eqref{eq:rep} such that
    \begin{align}
        \kappa^x(B)=\int_{[0,1]}\kappa_{\bmu'(x)}^x(\bmu'(x),\d    u)\widetilde{\kappa}_{h(\bmu'(x))}^{(h(\bmu'(x)),u)}(B).\label{eq:rep2}
    \end{align}
    Since $\bmu$ is jointly atomless, each $\mu_{\bz}$ is atomless, hence we may pick $\kappa_\bz$ that is Monge for $\bz\in\rd$. As a consequence, $\kappa$ is given by a composition of two Monge maps, hence is Monge. Denote by $f$ the map that induces $\kappa$. It is then immediate from \eqref{eq:rep2} that $\bnu'(f(x))=h(\bmu'(x))$.    
    %Recall from Theorem 3 of \cite{wang2022simultaneous} the parity formula
%     \begin{align}
%     \label{eq:rep}
% \kappa(x;B)=\int_{\R_+^d\times[0,1]}\int_{[0,1]}\kappa_{\bmu'(x)}(x;(\bmu'(x),\d u))\hat{\kappa}((\bmu'(x),u);&(\d \bz',\d u'))\widetilde{\kappa}_{\bz'}^{(\bz',u')}(B)\nonumber\\
% &\text{for }x\in \X \text{ and }B\subseteq Y,
% \end{align}
% where $\kappa_\bz,\tilde{\kappa}_{\bz'}$ are certain kernels in $\K(\mu_\bz,\delta_\bz\times\tau)$ and $\K(\delta_{\bz'}\times\tau,\nu_{\bz'})$ respectively, and $\hat{\kappa}$ is given by the Monge map $h$. This means for $y\in B$
\end{proof}

\begin{proof}[Proof of Theorem \ref{2way}]
Assume that $\mm=\nn$ and $c(x,y)$  is continuous. Then the martingale transport is unique, so that any $\hat{\kappa}\in\M_{b,1}$ must be the identity in the first coordinate. The transport cost \eqref{eq:costafterrep} then simplifies into
\begin{align*}
    \C(\kappa)&=\int_{\cR }\int_{[0,1]}\hat{\kappa}^{(\bz,u)}(\{\bz\}\times \d u')\left(\int_{A_\bz}\int_{B_{\bz}} c(x,y)\kappa_\bz^x(\{\bz\}\times \d u)\widetilde{\kappa}_{\bz}^{(\bz,u')}(\d y)\mu_{\bz}(\d x)\right)\mm(\d \bz)\\
    &=\int_{\rd}\int_{A_\bz}\int_{B_{\bz}} c(x,y)\int_{[0,1]}\int_{[0,1]}\kappa_\bz^x(\{\bz\}\times \d u)\hat{\kappa}^{(\bz,u)}(\{\bz\}\times \d u')\widetilde{\kappa}_{\bz}^{(\bz,u')}(\d y)\mu_{\bz}(\d x)\mm(\d \bz),
\end{align*}
so that
\begin{align}\inf_{\hat{\kappa}\in \M_{b,1}}\C(\kappa)    &=\int_{\rd}\Bigg(\inf_{\hat{\kappa}\in \M_{b,1}}\int_{A_\bz}\int_{B_{\bz}} c(x,y)\int_{[0,1]}\int_{[0,1]}\kappa_\bz^x(\{\bz\}\times \d u)\nonumber\\
    &\hspace{2.5cm}\hat{\kappa}^{(\bz,u)}(\{\bz\}\times \d u')\widetilde{\kappa}_{\bz}^{(\bz,u')}(\d y)\mu_{\bz}(\d x)\Bigg)\mm(\d \bz)\nonumber\\
    &\geq \int_{\rd}\left(\inf_{\kappa\in\K(\mu_{\bz},\nu_{\bz})}\int_{A_\bz}\int_{B_{\bz}} c(x,y)\kappa^x(\d y)\mu_{\bz}(\d x)\right)\mm(\d \bz)\label{eq:ineqiseq}\\
    &=\int_{\R_+^d}\I_c(\mu_{\bz},\nu_{\bz}) P (\d {\bf z}).\nonumber
\end{align}
We next show the inequality in \eqref{eq:ineqiseq} is in fact an equality, so that (i) is equivalent to (iii). Recall from the disintegration theorem that the map
$$\R_+^d\to \P(X)\times\P(Y),\ \bz\mapsto (\mu_{\bz},\nu_{\bz})$$ is measurable. By Corollary 5.22 in \cite{V09} and since $c$ is continuous, there exists a measurable map $\bz\mapsto \pi_{\bz}$ such that for each $\bz$, $\pi_{\bz}$ is an optimal transport plan from $\mu_{\bz}$ to $\nu_{\bz}$. We then define the average measure
$$\pi:=\int_{\R^d_+}\pi_{\bz}\mm(\d \bz).$$
It is straightforward to check using (\ref{o})  that $\pi\in\Pi(\bmu,\bnu)$. Alternatively, using the kernel formulation, this means there exists a well-defined  stochastic kernel $\kappa\in\K(\bmu,\bnu)$ such that $\kappa\in\K(\mu_{\bz},\nu_{\bz})$ is an optimal transport from $\mu_{\bz}$ to $\nu_{\bz}$. Therefore, \eqref{eq:ineqiseq} is an equality, and equality holds if and only if
$$\I_c(\mu_{\bz},\nu_{\bz})=\int_{A_{\bf z}}\int_Yc(x,y)\kappa(x,\d y)\mu_{\bf z}(\d x).$$
That is, $\kappa$ is optimal from $\mu_{\bz}$ to $\nu_{\bz}$ for $P$-a.s.~$\bz$. This gives the equivalence of (ii) and (iii).
\end{proof}

\begin{proof}[Proof of Proposition \ref{newdua}]
To show the $\geq$ direction, consider any $\pi\in\Pi(\bmu,\bnu)$ and any $(\phi,\psi)\in\widetilde{\Phi}(c)$. Recall from Theorem \ref{2way} that
$$\pi(\{(x,y)\mid \bmu'(x)\neq \bnu'(y)\})=0.$$
It then holds that 
$$\int_X\phi\, \d\mub+\int_Y \psi\, \d \nub=\int_{X\times Y}\phi(x)+\psi(y)\pi(\d x,\d y)\leq \int_{X\times Y}c\, \d\pi.$$
This proves the $\geq$ in \eqref{dua2}.

Using \eqref{2ic} and the classic duality, it suffices to prove
\begin{align*}
    &\hspace{0.5cm}\sup\left\{\int_X\phi\, \d\mub+\int_Y \psi\, \d \nub\mid (\phi,\psi)\in\widetilde{\Phi}_c\right\}\\
    &\geq \int_{\R^d_+}\sup\left\{\int_X\phi_\bz\, \d\mu_\bz+\int_Y \psi_\bz\, \d \nu_\bz\mid\phi_\bz(x)+\psi_\bz(y)\leq c(x,y)\right\}\mm(\d \bz).
\end{align*}
By Theorem 1.39 of \cite{S15}, the suprema on the right-hand side are attained for bounded continuous functions $\phi_\bz,\psi_\bz$. By Theorem 18.19 of \cite{AB06}, there exists a measurable selection $\bz\to(\phi_\bz,\psi_\bz)$ where each $(\phi_\bz,\psi_\bz)$ is a maximizer.  We define $\phi(x)=\phi_{\bmu'(x)}(x)$ and $\psi(y)=\psi_{\bnu'(y)}(y)$. Since $\bz\mapsto\phi_{\bz}(x)$ is measurable and $x\mapsto \phi_{\bz}(x)$ is continuous, we have $(\bz,x)\mapsto \phi_\bz(x)$ is jointly measurable, and hence $\phi,\psi$ are measurable. Moreover, $(\phi,\psi)\in L^1(\mub)\times L^1(\nub)$ since $c$ is bounded. Evidently, $(\phi,\psi)\in\widetilde{\Phi}_c$. It also follows from the disintegration theorem that
$$\int_X\phi\, \d\mub=\int_{\R^d_+}\int_X\phi_\bz\, \d\mu_\bz \mm(\d \bz).$$
This proves the desired inequality and hence \eqref{dua2}.

By Theorem \ref{2way}, the infimum in \eqref{dua2} is attained. Our construction of the maximizers $\phi,\psi$ above also implies that the supremum is attained. 
\end{proof}

\begin{proof}[Proof of Proposition \ref{d2}]

We first note that, since $\mu_1\sim\mub$ and $\nu_1\sim\nub$, by  Lemma 3.5 of \cite{SSWW19}, $\bmu\simeq \bnu$ is equivalent to 
$$\left(\frac{\d\mu_1}{\d\mu_1},\frac{\d\mu_2}{\d\mu_1}\right)|_{\mu_1}\dd\left(\frac{\d\nu_1}{\d\nu_1},\frac{\d\nu_2}{\d\nu_1}\right)|_{\nu_1}.$$
By proper transformations we may without loss of generality assume that $\mu_1$ and $\nu_1$ are standard Gaussian, which we denote by $\chi$. We then have
$$\left(\frac{\d\mu_2}{\d \chi}\right)|_\chi\dd\left(\frac{\d\nu_2}{\d \chi}\right)|_\chi.$$

Suppose that $\mu_2 = N( \bm,\Sigma )$ and $\nu_2 = N(\bn,\Omega )$. Plugging in the densities we obtain (where $\mathbf{Z}$ is a standard Gaussian random vector)
\begin{align*}
&\hspace{0.5cm}\sqrt{\frac{1}{\det \Sigma }}\exp\left(-\frac{1}{2}((\mathbf{Z}- \bm)^{\top}\Sigma ^{-1}(\mathbf{Z}- \bm)-\mathbf{Z}^{\top}\mathbf{Z})\right)\\
&\dd\sqrt{\frac{1}{\det \Omega }}\exp\left(-\frac{1}{2}((\mathbf{Z}-\bn)^{\top}\Omega ^{-1}(\mathbf{Z}-\bn)-\mathbf{Z}^{\top}\mathbf{Z})\right).
\end{align*}
Taking logarithm we obtain 
\begin{align*}
&\hspace{0.5cm}(\mathbf{Z}- \bm)^{\top}\Sigma ^{-1}(\mathbf{Z}- \bm)-\mathbf{Z}^{\top}\mathbf{Z}+\log {\det \Sigma }\dd(\mathbf{Z}-\bn)^{\top}\Omega ^{-1}(\mathbf{Z}-\bn)-\mathbf{Z}^{\top}\mathbf{Z}+\log {\det \Omega }.
\end{align*}
Using (5) in \cite{G63} we can compute the Laplace transforms, so that for all $t$,
\begin{align*}
&\frac{\exp(-2( t \Sigma ^{-1} \bm)^{\top}(I-2  t (\Sigma ^{-1}-I))^{-1}( t \Sigma ^{-1} \bm))}{|\det(I-2  t (\Sigma ^{-1}-I))|^{1/2}}\times \exp\left( t ( \bm^{\top}\Sigma  \bm+\log\det\Sigma )\right)\\
&\hspace{2.5cm}=\frac{\exp(-2( t \Omega ^{-1}\bn)^{\top}(I-2  t (\Omega ^{-1}-I))^{-1}( t \Omega ^{-1}\bn))}{|\det(I-2  t (\Omega ^{-1}-I))|^{1/2}}\times \exp\left( t (\bn^{\top}\Omega \bn+\log\det\Omega )\right).
\end{align*}

After squaring both sides, we may recognize either side as a product of a rational function in $t$ and an exponential of a rational function in $t$ (see e.g.,  \cite{MP92}, Theorem 3.2a.2). The rational functions on both sides must coincide. Thus, for all $t$,
\begin{align}
    |\det(I-2t(\Sigma^{-1}-I))|=|\det(I-2t(\Omega^{-1}-I))|.\label{a}
\end{align}
Taking logarithm of the rest we see that the Taylor coefficients around $t=0$ of $-2( t \Sigma ^{-1} \bm)^{\top}(I-2  t (\Sigma ^{-1}-I))^{-1}( t \Sigma ^{-1} \bm)$ and $t ( \bm^{\top}\Sigma  \bm+\log\det\Sigma )$ separate. This yields  
\begin{align} &(\Sigma^{-1}\bm)^{\top}(I-2t(\Sigma^{-1}-I))^{-1}(\Sigma^{-1}\bm) =(\Omega^{-1}\bn)^{\top}(I-2t(\Omega^{-1}-I))^{-1}(\Omega^{-1}\bn).\label{b}
\end{align}
From (\ref{a}), we have that the characteristic polynomials of $\Sigma$ and $\Omega$ coincide. Since both of them are symmetric and positive definite, they have the same eigenvalues counted with multiplicity. Writing $\Sigma^{-1}=PDP^{-1}$ and $\Omega^{-1}=QD'Q^{-1}$ with $P,Q$ orthogonal, we have that there is an elementary permutation matrix $E$ such that $D=ED'E^{-1}$. This gives $\Sigma^{-1}=(PEQ^{-1})\Omega^{-1}(PEQ^{-1})^{-1}$. Plugging this into the (\ref{b}), we have for all $t$,
\begin{align*}&\hspace{0.5cm}((PEQ^{-1})^{-1}\Sigma^{-1}\bm)^{\top}(I-2t(\Omega^{-1}-I))^{-1}((PEQ^{-1})^{-1}\Sigma^{-1}\bm)\\
&=(\Omega^{-1}\bn)^{\top}(I-2t(\Omega^{-1}-I))^{-1}(\Omega^{-1}\bn).\end{align*}
By expanding the term $(I-2t(\Omega^{-1}-I))^{-1}$ and comparing the coefficients in the expansion, we have for any $k\geq 2$, 
$$((PEQ^{-1})^{-1}\bm)^{\top}\Omega^{-k}((PEQ^{-1})^{-1}\bm)=\bn^{\top}\Omega^{-k}\bn.$$
Since $\Omega^{-1}=QD'Q^{-1}$, we have
\begin{align}((PE)^{-1}\bm)^{\top}(D')^k((PE)^{-1}\bm)=(Q^{-1}\bn )^{\top}(D')^k(Q^{-1}\bn).\label{ev}\end{align}
Since $\Omega$ is positive definite, $D'$ is diagonal and has positive entries along the diagonal. Denote $\lambda_1,\dots,\lambda_{\ell}$ the distinct eigenvalues (or distinct diagonal entries) of $D'$ and $S_1,\dots,S_{\ell}$ the corresponding eigenspaces with dimensions $d_1,\dots,d_{\ell}$. The system of equations (\ref{ev}) then becomes $\ell$ linearly independent equations since the rank of the  Vandermonde matrix formed by diagonal entries of $D'$ is at most $\ell$. In this way, (\ref{ev}) reduces to $\ell$ restrictions that the lengths of the vectors $(PE)^{-1}\bm$ and $Q^{-1}\bn$ are the same on each $S_\ell$. Hence, there exists an orthogonal matrix $O$ consisting of $\ell$ blocks on the subspaces $S_{\ell}$, each of which is an element in $\mathcal O(d_{\ell})$ (the set of orthogonal matrices of dimension $d_{\ell}$), such that $Q^{-1}\bn=O(PE)^{-1}\bm$. Thus $\bn=QO(PE)^{-1}\bm=(QOQ^{-1})(PEQ^{-1})^{-1}\bm$. Since $D'$ is a multiple of identity on each $S_{\ell}$, it commutes with $O$ on each block, hence $D'$ commutes with $O$. Therefore, the matrix
$$(PEQ^{-1})^{-1}\Sigma^{-1}(PEQ^{-1})=\Omega^{-1}=QD'Q^{-1}$$commutes with $QOQ^{-1}$. We conclude that
$$\Omega^{-1}=(QO(PE)^{-1})^{-1}\Sigma^{-1}(QO(PE)^{-1}).$$
That is, there exists a matrix $M:=QO(PE)^{-1}$ such that $\Omega^{-1}=M^{-1}\Sigma^{-1}M$ and $\bn=M\bm$. Therefore, the linear map $M$ transports $\mu_2$ to $\nu_2$. Since $M$ is orthogonal, it also transports $\chi=\mu_1$ to $\chi=\nu_1$. This concludes the proof.
\end{proof}

A natural question to ask is whether Proposition \ref{d2} extends to dimensions $d>2$. In this case, computation of Laplace transforms yields that instead of the relation (\ref{a}) above, we have for all $\mathbf{t}=\{t_j\}_{2\leq j\leq d}$ that $$\left|\det\left(I-2\sum_{j=2}^d t_j(\Sigma_j^{-1}-I)\right)\right|=\left|\det\left(I-2\sum_{j=2}^d t_j(\Omega_j^{-1}-I)\right)\right|$$and our goal is to provide an orthogonal matrix $P$ such that for any $2\leq j\leq d$,  $\Sigma_j^{-1}=P\Omega_j^{-1}P^{-1}$.
This is related to the simultaneous similarity of matrices problem, which was solved in \cite{F83} in the complex case. \cite{F83} proved that given some mild conditions on the characteristic polynomial
$$p(\lambda,x):=\det\left(\lambda I-\sum_{j=1}^dA_jx^j\right),$$ there are only finitely many orbits of tuples of symmetric matrices $(A_1,\dots,A_d)$ under the action of simultaneous conjugation by an orthogonal matrix.
An open problem was raised whether the same holds for real-valued matrices in \cite{F83}. A counterexample was provided later in \cite{S98} with matrices that are not positive definite. In addition, note that to apply to our situation, we need a single orbit instead of a finite number of them. Nevertheless, we are not aware of counterexamples in the case $d>2$ to Proposition \ref{d2}. If two-way transports exist between tuples of Gaussian measures while no linear transport exists, it is interesting to know what such a transport looks like.

\subsection{On the Wasserstein distance between vector-valued measures}\label{sec:W}

The aim of this section is to  propose a notion of the Wasserstein distance between $\R^d_+$-valued probability measures on a Polish space $X$ equipped with a metric $\rho$, using the optimal cost in simultaneous transport. Throughout this section, we consider the reference measure $\eta=\bar{\mu}$ and a number $p\geq 1$. 

Let us first recall the classic definition of the Wasserstein distance. Consider a Polish space $(X,\rho)$ and define
$$\P_p(X):=\left\{\mu\in\P(X)\mid \int_X \rho(x_0,x)^p\mu(\d x)<\infty\text{ for some }x_0\in X\right\}.$$
The Wasserstein distance between probability measures $\mu,\nu\in\P_p(X)$ is the metric given by 
$$\W_p(\mu,\nu):=\left(\inf_{\pi\in\Pi(\mu,\nu)}\int_{X\times X}\rho(x,y)^p\pi(\d x,\d y)\right)^{1/p}.$$
The space $(\P_p(X),\W_p)$ is again a Polish space.

For $\R^d_+$-valued measures, we may similarly define
$$\P(X)_{p,\rho}^d:=\left\{\bmu\in\P(X)^d\mid \int_X\rho(x,x_0)^p\mub(\d x)<\infty \mbox{ for some $x_0\in X$}\right\}.$$ The following consequence of  Theorem \ref{2way} provides a collection $\cE$  of $\R^d_+$-valued probability measures $\cE\subseteq \P(X)^d_{p,\rho}$ such that for any $\bmu,\bnu\in \cE$, $\W_p(\bmu,\bnu)=\W_p(\bnu,\bmu)<\infty$. We recall the equivalence relation $\simeq$ from Section \ref{61}.

\begin{proposition}\label{sym} Let $\bmu,\bnu\in\P(X)^d$ and suppose that both $\Pi(\bmu,\bnu)$ and $\Pi(\bnu,\bmu)$
are non-empty and $c(x,y)$ is continuous and symmetric in $x,y$. Then 
$$\I_c(\bmu,\bnu)=\I_{\widetilde{c}}(\bnu,\bmu)$$ where $\widetilde{c}(y,x)=c(y,x)$.
\end{proposition}

\begin{proof}[Proof of Proposition \ref{sym}]
By Theorem \ref{2way}, we have 
$$\I_c(\bmu,\bnu) =\int_{\R_+^d}\I_c(\mu_{\bz},\nu_{\bz}) P (\d {\bf z})=\int_{\R_+^d}\I_{\widetilde{c}}(\nu_{\bz},\mu_{\bz}) P (\d {\bf z})=\I_{\widetilde{c}}(\bnu,\bmu),$$where the second step follows since the classic optimal transport problem is symmetric. 
\end{proof}
 The upshot of Proposition \ref{sym} is that, for $\bmu,\bnu$ belonging to the same equivalence class we can define the Wasserstein distance 
$$\W_p(\bmu,\bnu):=\left(\inf_{\pi\in\Pi(\bmu,\bnu)}\int_{X\times X}\rho(x,y)^p\pi(\d x,\d y)\right)^{1/p}.$$
By the Decomposition Theorem, for $\bmu,\bnu\in \cE_P$, we have
$$\W_p(\bmu,\bnu)^p=\int_{\R_+^d}\W_p(\mu_{\bz},\nu_{\bz})^p P (\d {\bf z}).$$
The following corollary then follows from standard results on the analysis on the space of random variables taking values in a Polish space; see \cite{C02}.

\begin{corollary}
For each $1\leq p<\infty$, the metric space $(\cE_P,\W_p)$ is complete and separable, hence a Polish space.
\end{corollary}

Since for each $\kappa\in\K(\bmu,\bnu)$,
$$\int_{X\times X} c(x,y)  \bar{\mu} \otimes \kappa (\d x,\d y)=\frac{1}{d}\sum_{j=1}^d\int_{X\times X} c(x,y)  \mu_j \otimes \kappa (\d x,\d y),$$
we have by taking infimum that
\begin{align}\W_p(\bmu,\bnu)^p\geq\frac{1}{d}\sum_{j=1}^d\W_p(\mu_j,\nu_j)^p.\label{wd}\end{align}
It is also straightforward to see that  \eqref{wd}  is not an equality in Example \ref{ex:gaussian2}.

In the case $d=1$, the metric $\W_p$ metrizes the weak topology if the space $X$ is compact. On the other hand, if $d>1$, in general the convergence of $\W_p(\bmu^{(n)},\bmu)$ to zero does not imply the weak convergence of each component of $\bmu^{(n)}$ to that of $\bmu$, even if $\bmu^{(n)},\bmu\in\mathcal E_P$ for $n\in\N$ and $X$ is compact.

\begin{example}
As a sanity check, let us consider the case where $\mu_1=\dots=\mu_d$ and $\nu_1=\dots=\nu_d$. Then according to discussions in Section \ref{sec:12}, the optimal transport from $\mu_1$ to $\nu_1$ is also optimal from $\bmu$ to $\bnu$. This means
$$\W_p(\bmu,\bnu)^p=\frac{1}{d}\sum_{j=1}^d\W_p(\mu_j,\nu_j)^p=\W_p(\mu_1,\nu_1)^p.$$In other words, in the trivial case where all measures are equal, our Wasserstein distance is the same as the classic Wasserstein distance between such measures.

As another sanity check, consider $d=1$, then for any $\mu,\nu$, both $\Pi(\mu,\nu)$ and $\Pi(\nu,\mu)$
are non-empty, so that $\W_p$ is a metric on $\P(X)_{p,\rho}$ and it coincides with the classic Wasserstein distance.
\end{example}

\begin{example}\label{ex0}
Suppose that $\bmu\in\P(\R)^d$, define $T(x)=ax+b$ for some $a>0$, $b\in\R$ and $\bnu:=\bmu\circ T^{-1}$. Consider the convex cost $c(x,y)=|x-y|^p,\ p\geq 1$. Then since the linear transformation is comonotone, the associated kernel $\kappa_T\in \K(\bar{\mu},\bar{\nu})$ is an optimal transport from $\bar{\mu}$ to $\bar{\nu}$. By arguments in Section \ref{sec:12}, $\kappa_T$ is also optimal from $\bmu$ to $\bnu$.  In particular, (\ref{wd}) is an equality. Moreover, by the arguments in Section \ref{sec:12}, in case $\mu_1,\dots,\mu_d$ have disjoint supports, (\ref{wd}) is also an equality.
\end{example}

\subsection{Dual MOT-SOT parity}
\label{dualparity} 
Duality for MOT was first established by \cite{BHP13} in the following form. Given probability measures $\mu,\nu$ on $\R$ with $\mu\lcx\nu$ and an upper semi-continuous cost function $c$, it holds\begin{align}
    &\inf_{\pi\in\M(\mu,\nu)}\int c(x,y)\pi(\d x,\d y)\nonumber\\&\hspace{2cm}=\sup\left\{\int\phi\,\d\mu+\int\psi\,\d\nu\mid\phi(x)+\psi(y)+h(x)(y-x)\leq c(x,y)\right\},\label{duamot}
\end{align}
where it is also noted that the supremum may not always be attained; see also \cite{BNT17}. Our goal in this section is to connect the dual problems in \eqref{dua} and \eqref{duamot} when the primal problems are connected via the MOT-SOT parity.

Let us consider two measures $P,Q$ supported on $[0,2]$ with mean $1$ (this extends natually to compactly supported measures), with $P\gcx Q$. Let $c(x,y)$ be a continuous cost function. We next construct measures $\bmu,\bnu$ on $[0,1]$ so that the corresponding SOT problem is connected to the MOT problem with marginals $P,Q$. Let $F,G$ be cdfs for $P,Q$ and assume they are continuously invertible.\footnote{These regularity conditions on $P,Q$  do not affect the non-attainability of the supremum in \eqref{duamot}. Indeed, it is the irreducibility of the martingale coupling that matters.}   Let $\tau$ be the Lebesgue measure on $[0,1]$ and define $\d\mu_1=F^{-1}\d\tau,~\d\mu_2=(2-F^{-1})\d\tau,~\d\nu_1=G^{-1}\d\tau,~\d\nu_2=(2-G^{-1})\d\tau$. In this case, $\bmu'$ and $\bnu'$  are injective.  By   \eqref{dua} and Example \ref{4}, the dual SOT problem solves
\begin{align*}
    &\hspace{0.5cm}\sup\Bigg\{\int\phi(x)\d x+\int\psi_1(y)G^{-1}(y)\d y+\int\psi_2(y)(2-G^{-1}(y))\d y\mid\\
    &\hspace{2cm}~{{ \phi(x)+\psi_1(y)F^{-1}(x)+\psi_2(y)(2-F^{-1}(x))\leq c(F^{-1}(x),F^{-1}(y))}}\Bigg\}\\
    &=\sup\Bigg\{\int\phi(x)\d x+\int\psi_1(y)G^{-1}(y)\d y+\int\psi_2(y)(2-G^{-1}(y))\d y\mid\\
    &\hspace{2cm}~{{ \phi(F(x))+\psi_1(G(y))x+\psi_2(G(y))(2-x)\leq c(x,y)}}\Bigg\}\\
    &=\sup\Bigg\{\int\phi(F^{-1}(x))\d x+\int\psi_1(G^{-1}(y))G^{-1}(y)\d y+\int\psi_2(G^{-1}(y))(2-G^{-1}(y))\d y\mid\\
    &\hspace{2cm}~{{ \phi(x)+(\psi_1(y)-\psi_2(y))x+2\psi_2(y)\leq c(x,y)}}\Bigg\}\\
    &=\sup\Bigg\{\int\phi(x)P(\d x)+\int\psi_1(y)yQ(\d y)+\int\psi_2(y)(2-y)Q(\d y)\mid\\
    &\hspace{2cm}~{{ \phi(x)+(\psi_1(y)-\psi_2(y))x+2\psi_2(y)\leq c(x,y)}}\Bigg\}\\
    &=\sup\Bigg\{\int\phi(x)P(\d x)+\int\psi(y)Q(\d y)\mid{{ \phi(x)+\psi(y)+(y-x)h(y)\leq c(x,y)}}\Bigg\},
\end{align*}
using change of variables. This is precisely the (pointwise) MOT duality \eqref{duamot}. The dual MOT-SOT parity can be phrased as follows: if $(\hat{\phi},\hat{\psi},h)$ is a dual optimizer for MOT and $(\phi,\psi_1,\psi_2)$ for SOT, then
$$\hat{\phi}(z)=\phi(F^{-1}(z))\ \text{ and }\ \hat{\psi}(z')=\phi_1(G^{-1}(z'))z'-\phi_2(G^{-1}(z'))z'+2\psi_2(G^{-1}(z)).$$

\section{A small review of   optimal transport in higher dimensions}\label{rev}

 As mentioned in the introduction, we briefly survey a few directions on generalizing the Monge--Kantorovich optimal transport problem in higher dimensions present in the existing literature. The closest to our setting is 
 \cite{W19} and \cite{G20} in point (\ref{point}) below.

\begin{enumerate}[(i)]
    \item The multi-marginal optimal transport problem is a generalization of the classic Monge--Kantorovich
    transport problem concerning couplings of more than two marginals. For example, the objective of the Kantorovich version of such problems is to minimize
    $$\int_{X_1\times \dots\times X_d}c(x_1,\dots,x_d)\pi(\d x_1,\dots,\d x_d)$$among measures $\pi\in\P(X_1\times\dots\times X_d)$ with marginals $\mu_1,\dots,\mu_d$. 
    A duality formula can be established. However, the existence of a Monge transport is a more delicate problem for dimension $d\geq 3$. This problem has applications in physics and economics. See \cite{P15} and \cite{S15} for a review and \cite{R98} for a rich treatment. A solution for the minimization problem with $c(x_1,\dots,x_d)=(x_1+\dots+x_d)^2$ is obtained by \cite{WW16} under some conditions on $(\mu_1,\dots,\mu_d)$.
    
    \item More generally, \cite{R91} considered the multivariate marginal problem. For a collection $\mathcal E$ of subsets of $\{1,\dots,n\}$, consider the set of measures on $X_1\times\dots\times X_n$ that have fixed projections onto each $\prod_{j\in J} X_j,\ J\in \mathcal E$.
    The existence of such measures is a non-trivial task.
    A duality formula in a more general context was established earlier by \cite{R84}. For more recent results, see \cite{G19,G21} for the special case where $\mathcal E$ consists of subsets of cardinality $k,\ k\leq n$. This problem is also connected to Monge--Kantorovich problem with linear constraints.

    %For $1\leq k<n$ they considered measures on $X_1\times\dots\times X_n$ that have fixed projections onto each $X_{i_1}\times\dots \times X_{i_k}$ where $1\leq i_1<\dots<i_k\leq n$. The existence of such measures is a non-trivial task. Assuming existence, \cite{G19,G21} also established the theories of duality and cyclical monotonicity. \com{cite Rus..}
    
    \item \cite{B20} generalized the classic Monge--Kantorovich transport problem to multiple measures, with both transports and transfers allowed, with the name ``vector-valued optimal transport". Given probability measures $\bmu=(\mu_1,\dots,\mu_d)$ and $\bnu=(\nu_1,\dots,\nu_d),$ one is allowed to transport not only from each $\mu_j$ to $\nu_j$, but also from each $\mu_j$ to $\nu_{j'}$ where $j\neq j'$ (this is called a transfer), but the costs may be different. That is, the cost function is matrix-valued with $d^2$ components and the goal is to minimize the total cost (such a setting does not apply to our main motivating example in Example \ref{ex:1}). The existence of a transport is guaranteed and duality is obtained. \cite{B20} also investigated an extension of the Wasserstein distances.
    
    \item Some earlier studies are in a similar direction as \cite{B20}. To list a few, in   \cite{CCG18,C18}  and \cite{RCLO18}, the notion of ``vector-valued optimal transport" was proposed. Inspired by the dynamic formulation of classic optimal transport with the $L^2$ cost, they took the Benamou-Brenier perspective and formulated an optimal transport problem between vector-valued measures using divergences in a network flow problem. Similarly as \cite{B20}, both transports and transfers are allowed. In addition, numerical algorithms are available and applications to image processing are discussed.
    
    \item More recently, \cite{C21} proposed a generalization of the Kantorovich--Rubinstein transport problem to higher dimensions, with the name ``optimal transport for vector measures". Consider a metric space $(X,\rho)$ and a signed measure $\eta$ on $X$ such that $\eta(X)=0$ and there exists $x_0\in X$ such that $\int_X\rho(x,x_0)\n{\eta}(\d x)<\infty$, where $\n{\eta}$ is the total variation norm of $\eta$. This problem deals with
    $$\inf_{\pi: P_1\pi-P_2\pi=\eta}\int_{X\times X}\rho(x,y)\n{\pi}(\d x,\d y)$$where $\pi$ is an $\R^d$-valued measure, and $P_1,P_2$ are projections onto the first two coordinates. Existence of $\pi$ is guaranteed. The Kantorovich--Rubinstein duality formula is extended.
    \item A recent monograph \cite{W19} %and paper \cite{GW20} 
    and PhD thesis \cite{G20}
    discussed the notions of vector-valued transport and optimal multi-partitions. This is similar to our work as such vector-valued transports are indeed simultaneous transports. However, most of their results  concern duality formulas, existence of dual optimizers, and the structure (e.g., existence and uniqueness) of the \emph{optimal} multi-partition, where $Y$ is a finite set and under certain atomless assumptions.\footnote{which explains the name ``multi-partitions". Due to the nature of the problem, it seems mathematically difficult to approximate the general theory by the special case where $Y$ is discrete.} A different notion of Wasserstein distance between $\bmu$ and $\bnu$ was formulated by choosing both $\bmu$ and $\bnu$ as the measures at origin, defined as
    $$\V_p(\bmu,\bnu):=\left(\inf_{\bta\in\M(X)^d}\W_p(\bmu,\bta)^p+\W_p(\bnu,\bta)^p\right)^{1/p}.$$
    An application to learning theory is also discussed. The only mathematical overlaps between our paper and \cite{W19} and \cite{G20} are Proposition \ref{prop:ssww} and Theorem \ref{duality1}, where our results offer more generality.\label{point}
\end{enumerate}

\section{Application to a labour market equilibrium model}\label{sec:equilibrium}

We discuss a matching equilibrium model in a labour market via simultaneous transport, similar to that in the classic transport setting.\footnote{We refer to \cite{G16}, \cite{BCN23}, and the references therein for variations of labour market equilibrium models using an optimal transport approach.}
First, we state the relevant version of the duality formula in Theorem \ref{duality1}.
Suppose that $\bmu=(\mu_1,\dots,\mu_d)$ is a vector of probabilities on $X$, $\bnu=(\nu_1,\dots,\nu_d)$ is a vector of probabilities on $Y$, and $\eta\sim \mub$. Assume that $X$ and $Y$ are compact,
and $g:X\times Y\to [-\infty,\infty)$ is upper semi-continuous.
The duality formula, with a maximization in place of a minimization in \eqref{dua},  is 
\begin{align}&\hspace{0.5cm}\sup_{\pi\in\Pi_\eta(\mu,\bnu)}\int_{X\times Y}g \, \d \pi =\inf_{(\phi,\bpsi)\in \Phi_g}\int_X\phi\,  \d\eta+\int_Y \bpsi^{\top} \, \d \bnu,\label{eq:dua}\end{align}where  \begin{align*}&\Phi_g :=\Bigg\{(\phi,\bpsi)\in C(X)\times C^d (Y)\mid \phi (x)+\bpsi(y) \cdot
\frac{\d\bmu}{\d\eta}(x)\ge g(x,y)\Bigg\}.\end{align*}     
 
Let $x\in X$ represent worker labels and $y \in Y$ represent firms. 
The interpretation of $\eta$, $\bmu$ and $\bnu$ is given below.
\begin{enumerate}
\item $\eta$ is the distribution of the workers, i.e., how much proportion of the workers are labelled with $x\in X$.
In a discrete setting of $n$ workers in total, it would not hurt to imagine that $\eta(x)=1/n$; i.e., each worker has a different label.
\item There are $d$ types of skills in this production problem.
Workers with the same label have the same skills. 
The distribution $\mu_i$ describes the supply of type-$i$ skill  provided by the workers. 
In a discrete setting, $\mu_i(x)$ is  the  type-$i$ skill   provided by each worker label $x$. 
We denote by $\bmu'=\d \bmu/\d \eta$, that is, the (per-worker) skill vector. 
\item The distribution $\nu_i$ describes the demand of type-$i$ skill from the firms. In a discrete setting,  $\nu_i(y)$ is the type-$i$ skill demanded by each firm $y$. 
\end{enumerate} 
Assume that the total demand and the total supply of skills are equal, and hence both $\bmu$ and $\bnu$ are normalized to have total mass of $(1,\dots,1)$. A matching between the workers and the firms is an element $\pi$ of  $\Pi_{\eta}(\bmu,\bnu)$.
Let $g(x,y)$ represent  the production  of firm $y$ hiring worker $x$ per unit of worker. 
For a given matching $\pi$, 
 the total production in the economy is 
 $ 
  \int  g \, \d \pi .
 $ 
%and the total production of firm $y$ if all its demand  of labour is provided by worker $x$ is
%$$
%  G(x,y)  \nu(y).
%$$
%The ratio $\mu(x)/\nu(y)$, if no larger than $1$, represents the proportion of labour demanded by firm $y$ supplied by worker $x$, if $x$ is hired by $y$.
%We assume $\mu(x)/\nu(y)\le 1$ as there are typically more workers than firms. 

Take two arbitrary functions  $w:X\to \R$  and $\mathbf p:Y\to \R^d$.
As usual, $w(x)$  represents the wage of worker $x$.
The function $\mathbf p$ represents  the profit-per-skill vector of  firm  $y$ in the following sense:
if firm $y$ employs a skill vector $\mathbf q \in \R_+^d$, then the total profit of the firm is
$\mathbf p(y)\cdot \mathbf q $.  
 Taking $\mathbf q=\bmu'(x)$, the profit generated from hiring each worker $x$ is $\mathbf p(y) \cdot \bmu'(x)$. 
%The total demand of $y$ will be supplied by several workers, 
%say, $x\in X_{y}$, each with skill vector  $\bmu'(x)$. 
%In other words, the total demand of $y$ is 
%$$
 %\int_{ X_y} \bmu' \, \d  \eta  =  \bnu(X_y).
%$$ 
The total profit of all firms is 
$$
\int_{X\times Y }\mathbf p(y)\cdot \bmu'(x) \pi(\d x,\d y) = \int_{Y}\mathbf p^{\top} \d \bnu,
$$ 
which follows from the definition of $\pi$. 

For worker $x$, their  objective is to choose a firm to maximize their wage, that is 
$$
  \max_{y\in Y}\left\{  g(x,y) -    \mathbf p(y) \cdot  \bmu'(x) \right\}.
$$ 
For firm $y$, its   objective is  to hire  workers to maximize its  profit, that is
$$
  \max_{x\in X}\left \{   g(x,y) -   w(x) \right\}.
$$ 
For a social assignment $(w,\mathbf p)$ and a matching $\pi\in \Pi_{\eta}(\bmu,\bnu)$,
an \emph{equilibrium} is attained   if 
\begin{enumerate}[(a)]\item
the social assignment is optimal, that is
$$ w(x)= \max_{y\in Y}\left\{  g(x,y) -       \mathbf p(y) \cdot  \bmu'(x)  \right\}$$ 
 and  $$
\mathbf p(y)\cdot  \bmu'(x_y) = g(x_y, y)-w(x_y)=  \max_{x\in X}\left \{   g(x,y) -   w(x) \right\}.
$$
\item the total production in the economy is at least as large as the total wage plus the total profit, that is,  
\begin{equation}\label{eq:equilibrium4}
 \int_{X\times Y} g \,\d \pi  \ge     \int_X w \,\d\eta +\int_Y \mathbf p^{\top}   \d \bnu.
 \end{equation}
 \end{enumerate}  
 Since  (a) implies 
  \begin{equation}\label{eq:equilibrium6}
  w(x)+\mathbf p(y) \cdot \bmu'(x) \ge  g(x,y) 
  \end{equation}
   for all $x\in X$ and $y\in Y$,
   integrating \eqref{eq:equilibrium6} with respect to $\pi$ gives
   $$
     \int_X w \,\d\eta +\int_Y \mathbf  p^{\top}  \d \bnu \ge  \int_{X\times Y} g \,\d \pi   ,$$
   and hence, 
     \eqref{eq:equilibrium4} has to hold as an equality,  
 and this implies the duality \eqref{eq:dua}.
Again, an equilibrium exists if and only if duality  holds with both the infimum and the supremum attained. 
In the finite-state setting, the above attainability is automatic.

\end{document}